\documentclass[12pt, draftclsnofoot, onecolumn]{IEEEtran}
\usepackage{float}
\usepackage{color}

\ifCLASSINFOpdf
\usepackage[pdftex]{graphicx}
\usepackage{epstopdf}
\else
\usepackage[dvips]{graphicx}
\fi

\usepackage{amsmath,amssymb,amsthm}
\usepackage{multirow}
\usepackage{hhline}
\usepackage[caption=false,subrefformat=parens,labelformat=parens]{subfig}
\usepackage{balance}
\usepackage{tikz}
\usepackage[compress]{cite}
\usepackage{hyperref}
\usepackage{setspace}
\usepackage{amsthm}
\newtheorem{theorem}{Theorem}
\newtheorem{proposition}{Proposition}
\newtheorem{lemma}{Lemma}
\newtheorem{corollary}{Corollary}

\usepackage{acronym}
\acrodef{2D}{two-dimensional}
\acrodef{3D}{three-dimensional}
\acrodef{4G}{fourth generation}
\acrodef{5G}{fifth generation}
\acrodef{6G}{sixth generation}
\acrodef{ADT}{angle diversity transmitter}
\acrodef{ADR}{angle diversity receiver}
\acrodef{ANSI}{American national standards institute }
\acrodef{AP}{access point}
\acrodef{APD}{avalanche photodiode}
\acrodef{AWGN}{additive white Gaussian noise}
\acrodef{AWGR}{high port count arrayed waveguide grating router}
\acrodef{BER}{bit error ratio}
\acrodef{CSI}{channel state information}
\acrodef{CPC}{compound parabolic concentrator}
\acrodef{CoMP}{coordinated multipoint}
\acrodef{DC}{direct current}
\acrodef{DCO{-}OFDM}{direct current-biased optical orthogonal frequency division multiplexing}
\acrodef{DFB}{distributed feedback laser}
\acrodef{EGC}{equal gain combining}
\acrodef{FEC}{forward error correction}
\acrodef{FF}{fill factor}
\acrodef{FFT}{fast Fourier transform}
\acrodef{FOV}{field of view}
\acrodef{FSO}{free space optical}
\acrodef{IEC}{International electrotechnical commission}
\acrodef{IFFT}{inverse fast Fourier transform}
\acrodef{IM{-}DD}{intensity modulation and direct detection}
\acrodef{IM{/}DD}{intensity modulation and direct detection}
\acrodef{IR}{infrared}
\acrodef{ISI}{inter-spot interference}
\acrodef{JT}{joint transmission}
\acrodef{LSC}{luminescent solar concentrator}
\acrodef{LD}{laser diode}
\acrodef{LED}{light-emitting diode}
\acrodef{LiFi}{light fidelity}
\acrodef{MPE}{maximum permissible exposure}
\acrodef{MHP}{most hazardous position}
\acrodef{MRC}{maximal ratio combining}
\acrodef{NIR}{near infrared}
\acrodef{NOMA}{non-orthogonal multiple access}
\acrodef{PAM}{pulse amplitude modulation}
\acrodef{PIN}{positive-intrinsic-negative}
\acrodef{PSD}{power spectral density}
\acrodef{PD}{photodiode}
\acrodef{DCO{-}OFDM}{DC biased optical orthogonal frequency division multiplexing}
\acrodef{OFDM}{orthogonal frequency division multiplexing}
\acrodef{OFDMA}{orthogonal frequency division multiple access}
\acrodef{OOK}{on-off keying}
\acrodef{OW}{optical wireless}
\acrodef{OWC}{optical wireless communication}
\acrodef{QAM}{quadrature amplitude modulation}
\acrodef{RF}{radio frequency}
\acrodef{RIN}{relative intensity noise}
\acrodef{RSS}{received signal strength}
\acrodef{SINR}{signal-to-interference-plus-noise ratio}
\acrodef{SNR}{signal-to-noise ratio}
\acrodef{TIA}{transimpedance amplifier }
\acrodef{TEM}{transverse electromagnetic}
\acrodef{VCSEL}{vertical cavity surface emitting laser}
\acrodef{VLC}{visible light communication}
\acrodef{WiFi}{wireless fidelity}

\acrodef{NCD}{no constraint on dimensions}
\acrodef{MCD}{moderately constrained dimensions}
\acrodef{SCD}{strictly constrained dimensions}
\acrodef{FFOT}{fused fibre-optic taper}

\newcommand\raisepunct[1]{\,\mathpunct{\raisebox{0.4ex}{#1}}}

\begin{document}

\title{Design and Optimisation of High-Speed Receivers for 6G Optical Wireless Networks}

%
% \author{Elham~Sarbazi,~\IEEEmembership{Member,~IEEE,}
% 	Hossein~Kazemi,~\IEEEmembership{Member,~IEEE,}
% 	Mohammad~Dehghani~Soltani,~\IEEEmembership{Member,~IEEE,}
% 	Majid~Safari,~\IEEEmembership{Member,~IEEE,}
% 	and~Harald~Haas,~\IEEEmembership{Fellow,~IEEE,}
% 	and~\textcolor{red}{add new names},~\IEEEmembership{Fellow,~IEEE}}% <-this % stops a space

\author{\IEEEauthorblockN{Elham Sarbazi, Hossein Kazemi, Michael Crisp, Taisir El-Gorashi, \\Jaafar Elmirghani, Richard Penty, Ian White, Majid Safari and Harald Haas\vspace{-20pt}}
\thanks{This work has been submitted to the IEEE for possible publication. Copyright may be transferred without notice, after which this version may no longer be accessible.}
\thanks{Part of this work has been presented at IEEE Global Communications Conference (GLOBECOM 2022), 4–8 Dec 2022 \cite{ESarbazi2022Design}.}}

%\author{Michael~Shell,~\IEEEmembership{Member,~IEEE,}
%        John~Doe,~\IEEEmembership{Fellow,~OSA,}
%        and~Jane~Doe,~\IEEEmembership{Life~Fellow,~IEEE}% <-this % stops a space
%\thanks{M. Shell was with the Department
%of Electrical and Computer Engineering, Georgia Institute of Technology, Atlanta,

% \markboth{Journal of \LaTeX\ Class Files,~Vol.~14, No.~8, August~2015}%
% {Shell \MakeLowercase{\textit{et al.}}: Bare Demo of IEEEtran.cls for IEEE Journals}

\maketitle

%%%%%%%%%%%%%%%%%%%%%%%%%%%%%%%%%%%%%%%%%%%%%%%%%%%%%%%%%%%%%%%%%%%%%%%%%%%%%%%%%%%%%%%%%%%%%%%%%%%%
%%%%%%%%%%%%%%%%%%%%%%%%%%%%%%%%%%%%%%%%%%%%%%%%%%%%%%%%%%%%%%%%%%%%%%%%%%%%%%%%%%%%%%%%%%%%%%%%%%%%
\begin{abstract}
To achieve multi-Gb/s data rates in 6G optical wireless access networks based on narrow infrared (IR) laser beams, a high-speed receiver with two key specifications is needed: a sufficiently large aperture to collect the required optical power and a wide field of view (FOV) to avoid strict alignment issues. This paper puts forward the systematic design and optimisation of multi-tier non-imaging angle diversity receivers (ADRs) composed of compound parabolic concentrators (CPCs) coupled with photodiode (PD) arrays for laser-based optical wireless communication (OWC) links. Design tradeoffs include the gain-FOV tradeoff for each receiver element and the area-bandwidth tradeoff for each PD array. The rate maximisation is formulated as a non-convex optimisation problem under the constraints on the minimum required FOV and the overall ADR dimensions to find optimum configuration of the receiver bandwidth and FOV, and a low-complexity optimal solution is proposed. The ADR performance is studied using computer simulations and insightful design guidelines are provided through various numerical examples. An efficient technique is also proposed to reduce the ADR dimensions based on CPC length truncation. It is shown that a compact ADR with a height of $\leq0.5$~cm and an effective area of $\leq0.5$~cm$^2$ reaches a data rate of $12$~Gb/s with a half-angle FOV of $30^\circ$ over a $3$~m link distance.
\end{abstract}


%%%%%%%%%%%%%%%%%%%%%%%%%%%%%%%%%%%%%%%%%%%%%%%%%%%%%%%%%%%%%%%%%%%%%%%%%%%%%%%%%%%%%%%%%%%%%%%%%%%%
%%%%%%%%%%%%%%%%%%%%%%%%%%%%%%%%%%%%%%%%%%%%%%%%%%%%%%%%%%%%%%%%%%%%%%%%%%%%%%%%%%%%%%%%%%%%%%%%%%%%
\setstretch{1.45} 
\section{Introduction} \label{Sec:1}
\Ac{6G} wireless networks are envisioned to create an advanced communication infrastructure to support ubiquitous mobile ultra-broadband, ultra-high-speed with low-latency communication, and ultra-high data density services \cite{Chowdhury6G}. \Ac{OWC} based on \ac{IR} laser beams is considered as an enabling technology for \ac{6G} with the aim to achieve data rates of multi-Gb/s per user towards the realisation of Terabit/s wireless access networks \cite{Koonen:2020:NarrowBeamOWC, HKazemi2020tb, ESarbazi2020tb, HKazemi2022tb}. The accomplishment of such an ambitious goal by means of narrow laser beams requires to rethink common approaches used for conventional \ac{LiFi} receivers which are primarily tailored to the wide Lambertian emission profile of inexpensive light sources. Such receivers typically use silicon \acp{PD} with an intrinsically wide \ac{FOV} and a relatively large photosensitive area, whose detection bandwidth does not exceed few tens of MHz.
Optical wireless receivers consist of an optical concentrator, a photodetector and a \ac{TIA} \cite{StephenBAlexander1997}. Two types of frequently used photodetectors are \acp{APD} and \ac{PIN} diodes. While \acp{APD} offer higher sensitivities due to their internal current gain mechanism, \ac{PIN} diodes have a simpler structure and a much higher bandwidth. Unlike \acp{APD}, \ac{PIN} diodes do not need a high bias voltage to operate, which renders them a more qualified option for integration with battery-operated mobile devices. Though a \ac{TIA} with a high transimpedance gain is required to amplify the photocurrent generated by a \ac{PIN} \ac{PD} to compensate for the lack of internal current gain \cite{Personick1973Receiver}. Receivers based on \acp{APD} operate in shot noise-limited regime, and those using \ac{PIN} \acp{PD} are thermal noise-limited.

An optical concentrator is composed of an imaging or a non-imaging component \cite{JosephKahn1997}. Imaging optics is often used to form an image of a light source on the detector plane. Imaging components can be manufactured in compact form and they can provide a diffraction limited spatial resolution. However, they have a limited acceptance angle and hence a very strict requirement for alignment. This necessitates the use of complex alignment systems with prohibitive implementation costs especially for mobile applications \cite{YKaymak2018survey, MAbadi2019space, MFernandes2021adaptive, FGuiomar2022coherent}. Another disadvantage of imaging components is that their performance severely degrades with imperfect manufacturing and assembly errors.
With non-imaging optics, image formation is not a concern, instead the main purpose is efficient light collection and concentration. Compared to their imaging counterparts, non-imaging components have a wider acceptance angle for a given concentration gain and they are less prone to alignment issues. Also, they are generally larger in size, however, their size can be reduced at the cost of a slight loss in their concentration performance \cite{RWinston2005,P_DTsonev2019,Winston1975Principles,ARabl1976Optical,Welford1978Optics}. Therefore, non-imaging optics offers more flexibility to meet the design requirements. 

The challenges of implementing Gb/s optical wireless receivers for narrow laser beams entail the design of a high-performance optical front-end system for efficient collection and detection of the incident light at high speed. Mobile devices need a compact and low-cost receiver with two key specifications: a large aperture to capture sufficient optical power and a wide \ac{FOV} to keep connection with the transmitter by eliminating the strict alignment requirement. Designing such a receiver is associated with a twofold challenge. Foremost, there is a tradeoff between the bandwidth and the photosensitive area of a \ac{PD}. Thus, a high bandwidth comes with a small area, which restricts the collected optical power. The low collection efficiency of a small \ac{PD} can be compensated by using appropriate light-focusing optics to boost the received \ac{SNR}, in which case a governing tradeoff between the gain and \ac{FOV} of the optical component comes into play according to the etendue conservation law \cite{RWinston2005}. Hence, an improved \ac{SNR} can be obtained in exchange for a reduced \ac{FOV} and designing a high-speed optical wireless receiver realising a large collecting aperture and a wide \ac{FOV} at once remains a major challenge.

%---------------------------------------------------------------------------------------------------
\subsection{Related Works}
To date, several studies in the literature have addressed the existing challenges in the design of high-speed optical wireless receivers and have reported the achieved \ac{FOV} and data rates through various experimental works. In the following, we review an anthology of closely connected works and identify the remaining issues and research gaps.

In \cite{TKoonen2020Novel,TKoonen2021Beam}, Koonen \textit{et al.} proposed an imaging receiver design based on a Fresnel lens with a large aperture of diameter $50$~mm and a full-angle \ac{FOV} of $10^{\circ}$, whereby they experimentally demonstrated a data rate of $1$ Gb/s. In \cite{Umezawa2018High, Umezawa2021Array}, Umezawa \textit{et al.} designed and fabricated an $8\times8$ array of \ac{PIN} \acp{PD} followed by an aspheric lens of $15$ mm in diameter, offering a full-angle \ac{FOV} of $6^{\circ}$. The authors also reported the achievement of $25$ Gb/s transmission over a $10$~m link. These studies corroborated that an array of \acp{PD} equipped with an optical concentrator can yield an enhanced \ac{SNR} and a wider \ac{FOV} compared with a single \ac{PD} using similar optics. 

In \cite{SCollins2014F}, Collins \textit{et al.} proposed the use of a fluorescent concentrator made of a quantum dot material for free space optical communications. They showed that such a concentrator can realise an optical gain $50$ times higher than an etendue conserving concentrator with the same \ac{FOV}. In \cite{PManousiadis2016wide}, Manousiadis \textit{et al.} demonstrated the performance of a flat fluorescent optical antenna coupled with an \ac{APD}, attaining a \ac{FOV} of $60^\circ$ and an optical gain of $12$. In \cite{YDong2017nanopatterned}, Dong \textit{et al.} exploited a similar design to achieve the same \ac{FOV} with an optical gain of $3.2$. Although this type of optical receivers are used as a means of circumventing the etendue limitation, they have a low modulation bandwidth of no greater than $50$~MHz due to the long fluorescence lifetime. This bandwidth is not adequate for multi-Gb/s receiver design. In addition, fluorescent receivers are only applicable to specific wavelengths in the visible light spectrum. State-of-the-art high-speed receiver designs are discussed as follows.

%---------------------------------------------------------------------------------------------------
In \cite{OAlkhazragi2021wide, BOoi2022wide}, Alkhazrag and Ooi \textit{et al.} proposed and demonstrated an imaging receiver structure based on \acp{FFOT}, featuring a \ac{FOV} of $25^{\circ}$ and an optical gain of $120$. Unlike conventional imaging optics, \acp{FFOT} comprise hundreds of thousands of tapered optical fibres and as a result they have larger dimensions and a longer height as compared to the focal length of imaging lenses. On the upside, they can provide ease of alignment. However, the main disadvantage of \acp{FFOT} is that they are not widely available and their manufacturing cost is high. In \cite{Pham2022Auto}, Pham \textit{et al.} designed an imaging receiver with automatic alignment for beam-steered \ac{IR} light communication links. The proposed receiver is constructed of an imaging lens and a single high-speed photodetector while the whole setup is equipped with a motorised actuator. Although the half-angle \ac{FOV} is only $0.6^\circ$, the actuator automatically adjusts the receiver orientation by using a control algorithm to ensure that the incident light is within the \ac{FOV}. With this receiver setup, the authors reported data rates of up to $2$~Gb/s for a half-power beam diameter of $12$~cm. However, such a design involves bulky and costly components which is not suitable for compact and low-cost receiver design purposes.

In \cite{MDSoltani2022high}, Soltani \textit{et al.} proposed a multi-element imaging receiver design based on an \textit{array of arrays} structure for laser-based \ac{OWC} links. In this structure, each receiver element is composed of a \ac{PD} array and a commercial aspheric lens by Thorlabs \cite{DataSheetAsphericLens}, and multiple receiver elements are put together to form the outer array of \ac{PD} arrays to improve the overall optical gain performance. In order to preserve the receiver bandwidth, each \ac{PD} is assumed to be followed by a separate \ac{TIA} of its own and \ac{MRC} or threshold-based \ac{EGC} is then applied to process the output signals of individual \acp{PD}. The authors formulated an optimisation problem aiming to find the optimum configuration of the proposed structure in terms of the \ac{PD} side length, the spacing between the \acp{PD} in each array and the lens to array distance subject to specific design constraints including the \ac{FOV} and \ac{BER} requirements. For a minimum required full-angle \ac{FOV} of $15^{\circ}$ and a \ac{BER} of $10^{-3}$, the optimum solution suggests a design with $4$~cm$^2$ total area and a height of $1.8$~mm, achieving data rates of $23.8$~Gb/s using \ac{OOK} modulation and $21.1$~Gb/s based on \ac{DCO{-}OFDM}. However, the optimised design uses an aggregate of $64$ lenses and over $2300$ \acp{PD} and \acp{TIA}, resulting in a high hardware complexity. Moreover, with the considered lens having a focal length of $820$~{\textmu}m, the receiver performance is optimised within a $43$~{\textmu}m distance from the lens, and as a consequence it is highly sensitive to any displacements in the optical assembly. Any alignment error of a few tens of micrometer degrades the promised \ac{FOV} as well as the maximum delivered data rate. Such a design is extremely challenging to implement if not impractical, and it comes at a prohibitive fabrication cost.

%---------------------------------------------------------------------------------------------------
\subsection{Contributions}
\Acp{ADR} are well-known as a promising solution to provide a wide \ac{FOV} and a high optical gain simultaneously for \ac{OWC} systems \cite{JosephKahn1997, JCarruther2000angle, ZChen2014angle}. An \ac{ADR} design relieves strict alignment requirements and enables connectivity to multiple access points. Besides, the use of \acp{ADR} not only improves the mobility performance and allows for a seamless handover, it also addresses the beam obstruction issues. For laser-based \ac{OWC} systems that demand highly efficient light concentration, a non-imaging receiver design based on \acp{CPC} is more appealing as it provides maximum concentration gain determined by the etendue conservation law and more flexibility in the range of acceptance angles.

In this paper, we consider a multi-tier non-imaging \ac{ADR} architecture based on \acp{CPC} where each \ac{CPC} is coupled with a \ac{PD} array. After a thorough analytical modelling of the receiver optics and the subsequent signal combining schemes, we present the underlying design tradeoffs that control the receiver bandwidth and \ac{FOV}. Then, we proceed to rate maximisation problems to find the optimum receiver configuration. Problem formulations are presented under a unified optimisation framework taking into account the \ac{FOV} constraint as well as additional constraints on the receiver dimensions. To the best of our knowledge, this is the first study that addresses the interplay between the design tradeoffs systematically and deals with the optimisation of the non-imaging \ac{ADR} configuration for laser-based \ac{OWC} systems. Furthermore, we put forward a modified \ac{ADR} design based on truncated \acp{CPC} as an efficient approach to attain compact receiver dimensions. The proposed approach is applied to portable devices such as laptops and smart phones in which a confined space is available for \ac{6G} receiver integration. 

The main contributions of this paper are summarised as follows:
\begin{itemize}
\item An in-depth study of the fundamental tradeoffs in high-speed receivers and identifying the interrelated design tradeoffs in a non-imaging \ac{ADR}.
\item Proposing a two-stage signal processing scheme for the multi-element \ac{ADR} design to reduce the implementation complexity.
\item Building a unified and tractable analytical framework for evaluating the receiver performance considering various design parameters. The same framework applies to both short-range and long-range laser-based \ac{OWC} links regardless of the link distance.
\item Formulating and solving non-convex optimisation problems to acquire the optimum \ac{ADR} configuration under various constraints on its performance and dimensions.
\item Providing insightful and detailed discussions on the feasibility and scalability of the receiver design using extensive computer simulation results.
\item Proposing an effective modification to downsize the overall \ac{ADR} dimensions, leading to a low-complexity and compact receiver design for \ac{6G} laser-based optical wireless receivers.
\end{itemize}

The remainder of this paper is organised as follows. In Section~\ref{Sec:2}, the beam propagation model and eye safety considerations are introduced. In Section~\ref{Sec:3}, details of the receiver architecture and design tradeoffs are presented. In Section~\ref{Sec:4}, the achievable rate analysis and rate maximisation problems are discussed. In Section~\ref{Sec:5}, numerical examples with relevant discussions are provided. In Section~\ref{Sec:6}, the modified \ac{ADR} design and its performance evaluation are presented. Finally, in Section~\ref{Sec:7}, concluding remarks are drawn and future research directions are identified.

%%%%%%%%%%%%%%%%%%%%%%%%%%%%%%%%%%%%%%%%%%%%%%%%%%%%%%%%%%%%%%%%%%%%%%%%%%%%%%%%%%%%%%%%%%%%%%%%%%%%
%%%%%%%%%%%%%%%%%%%%%%%%%%%%%%%%%%%%%%%%%%%%%%%%%%%%%%%%%%%%%%%%%%%%%%%%%%%%%%%%%%%%%%%%%%%%%%%%%%%%
\setstretch{1.4} 
\section{System Model} \label{Sec:2}
In this study, we consider a single beam optical wireless link with a transmitter composed of a \acf{VCSEL} followed by a plano-convex lens. The receiver is located at a distance $D$ from the transmitter. 
%---------------------------------------------------------------------------------------------------
\subsection{Gaussian Beam Propagation} 
In an optical \ac{IM{/}DD} system, information bits are modulated onto the light intensity. The \ac{VCSEL} is assumed to have a Gaussian intensity profile. A Gaussian beam is primarily characterised by two parameters: beam waist $w_0$ and wavelength $\lambda$. Assuming the beam is travelling along the $z$ axis, the intensity distribution is given by \cite{Saleh2019}: 
\begin{equation}
	I(r,z) = \frac{2P_\mathrm{t}}{\pi w^2(z)}{\exp{\left( -\dfrac{2r^2}{w^2(z)} \right )}} \raisepunct{,}
	\label{Eq:2_1}
\end{equation}
where $P_\mathrm{t}$ is the transmit optical power, and $r$ and $z$ are the radial and axial positions, respectively. The beam radius $w(z)$ is expressed as \cite{Saleh2019}:
\begin{equation}
	w(z) = w_0 \sqrt{1+ {\left( \dfrac{z}{z_\mathrm{R}}\right)}^2}.
	\label{Eq:2_2}
\end{equation}
In addition, the Rayleigh range $z_\mathrm{R}$ is given by:
\begin{equation}
	z_\mathrm{R} = \frac{\pi {w^2_0} n}{\lambda} \raisepunct{,}
	\label{Eq:2_3}
\end{equation}
where $n$ represents the refractive index of the medium \cite{Saleh2019}. 

In this work, we only need the beam parameters after the lens transformation for the analysis and optimisation of the receiver performance. When refracted by a lens, the incident Gaussian beam is transformed into another Gaussian beam characterised by a different set of parameters: $w' _0$, $z'_\mathrm{R}$ and $w'(z)$. Details of how these parameters are obtained based on the parameters of a plano-convex lens can be found in \cite{Saleh2019}. The received optical power of the transformed beam within a circle of radius $\rho_0$ on the transverse plane at the receiver location (i.e., at $z=D$) is:
\begin{equation}
    P_{\mathrm{r}} = P_{\mathrm{t}} \left({1-\exp{\left( -\dfrac{2\rho_0 ^2}{({w'}(D))^2} \right )}} \right) \raisepunct{.}
    \label{Eq:Pr}
\end{equation}

%---------------------------------------------------------------------------------------------------
\subsection{Eye Safety Considerations} 
The maximum optical power emitted from a \ac{VCSEL} is subject to eye safety regulations. The laser safety regulations have been defined by the \ac{IEC} 60825 standard and the \ac{ANSI} \cite{IEC_Std608251,ANSI-Z136.4}. In these standards, the so-called \ac{MPE} is used as a metric for specifying the irradiance limit of a laser source, which must not be exceeded for the corresponding laser class. The \ac{MPE} value depends on the laser wavelength, the exposure time and the size of the apparent source. The latter is quantified using the subtense angle $\alpha$ which is defined as the plane angle subtended by an apparent source as viewed from a point in space. The subtense angle a critical parameter for eye safety assessment. It is a measure of the angular extent of the image formed on the retina and indicates how sharply light is focused on the retina. According to \cite{IEC_Std608251}, if $\alpha < 1.5$~mrad at the measurement distance, the source is classified as a \textit{point source}, otherwise as an \textit{extended source}.

A laser source is considered to be safe, if at any position in space, the fraction of the received power passing through the pupil of the human eye is less than the corresponding \ac{MPE} value multiplied by the pupil area. Therefore, the laser safety analysis often includes determining the \ac{MHP}, which is where the optical power collected by the pupil aperture is at its maximum. If the eye safety condition is met at this point, it ascertains that anywhere else in space is also eye-safe \cite{Henderson2003}. In this work, we have taken the eye safety limits into account by closely following the approach presented in \cite{Henderson2003}.

%%%%%%%%%%%%%%%%%%%%%%%%%%%%%%%%%%%%%%%%%%%%%%%%%%%%%%%%%%%%%%%%%%%%%%%%%%%%%%%%%%%%%%%%%%%%%%%%%%%%
%%%%%%%%%%%%%%%%%%%%%%%%%%%%%%%%%%%%%%%%%%%%%%%%%%%%%%%%%%%%%%%%%%%%%%%%%%%%%%%%%%%%%%%%%%%%%%%%%%%%
\begin{figure*}[t!]
    \centering
    \begin{minipage}[b]{\linewidth}
    \centering
    \subfloat[\label{Fig:Rx_a} 3D view ($1$-tier ADR)]{\includegraphics[height=0.14\textheight, keepaspectratio=true]{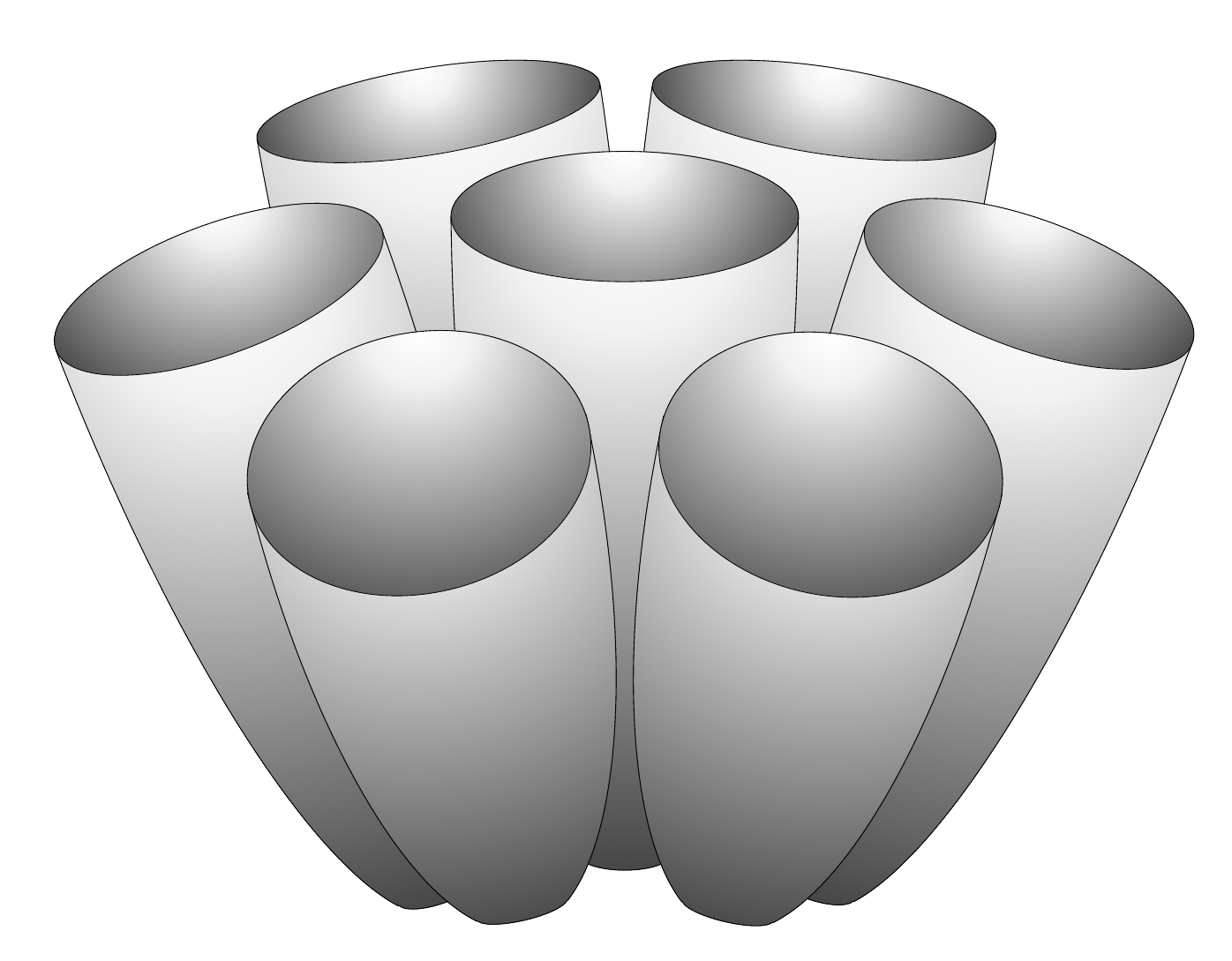}}\quad
    \subfloat[\label{Fig:Rx_b} top view ($1$-tier ADR)]{\includegraphics[height=0.16\textheight, keepaspectratio=true]{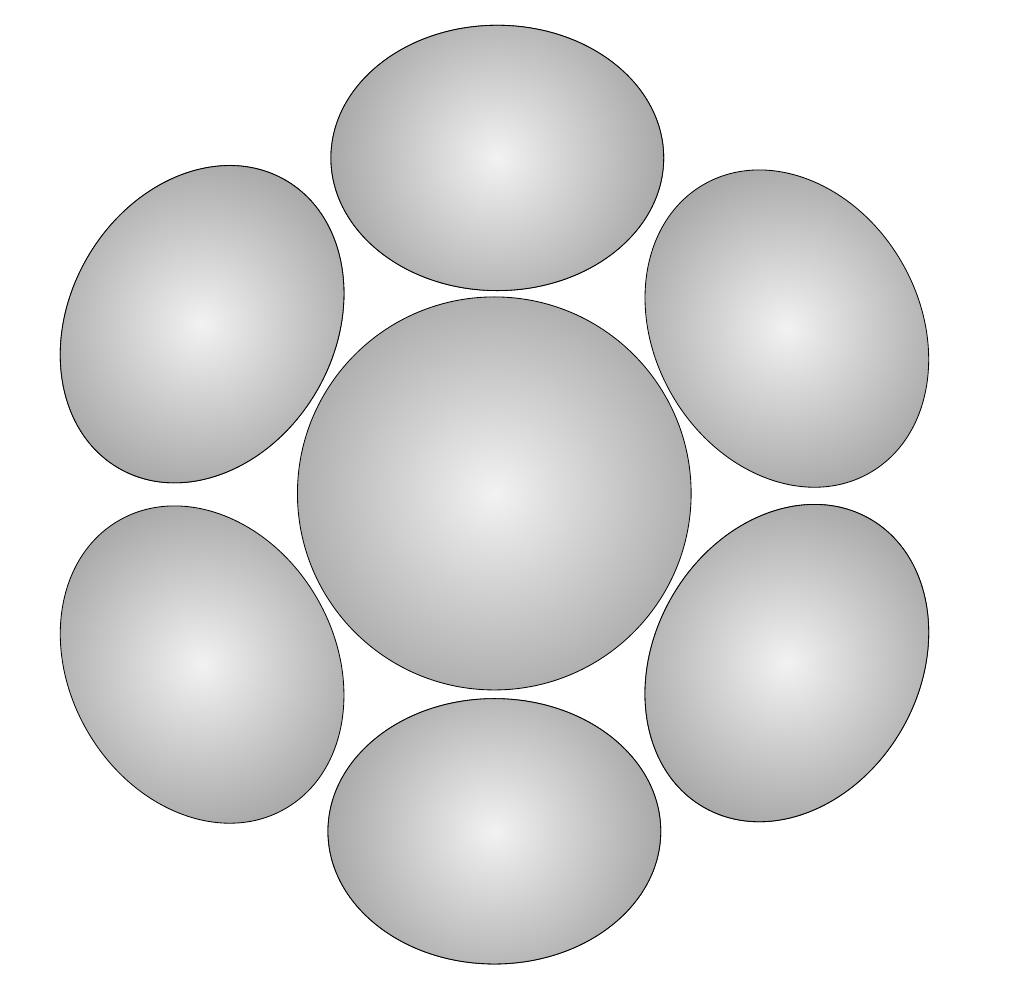}}\quad
    \subfloat[\label{Fig:Rx_c} top view ($2$-tier ADR)]{\includegraphics[height=0.2\textheight, keepaspectratio=true]{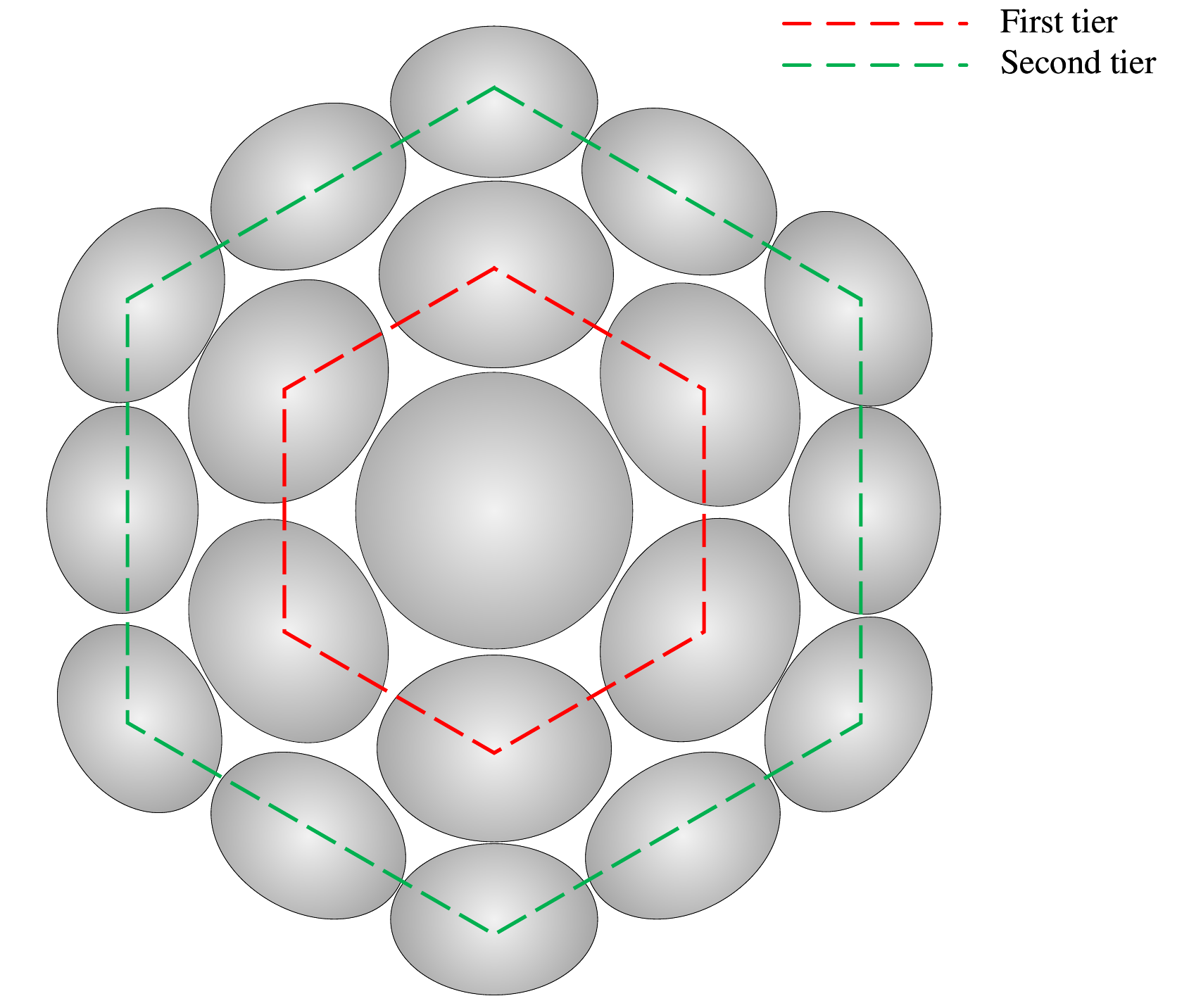}}
    \end{minipage}
    \caption{Angle diversity receiver structure.}
    \label{Fig:Rx}
    \vspace{-20pt}
\end{figure*}

\section{High-Speed Receivers for 6G OWC} \label{Sec:3}
%---------------------------------------------------------------------------------------------------
\subsection{Angle Diversity Receiver}
In this study, we consider a non-imaging \ac{ADR} design by using \acp{CPC}. \acp{CPC} are non-imaging concentrators that are commonly used in applications requiring efficient light collection. Among various imaging or non-imaging optical components, \acp{CPC} can approach the maximum theoretical concentration gain determined by the law of conservation of etendue \cite{RWinston2005}. In addition, they enable light collection over a relatively wide range of incident angles, thereby providing more flexible designs in terms of the range of acceptance angles as well as alignment tolerances in comparison with imaging optical components such as lenses. Hence, the principal advantage of angle diversity reception based on \acp{CPC} is that it allows the receiver to realise a high optical gain and a wide \ac{FOV} simultaneously.

Figs.~\subref*{Fig:Rx_a} and \subref*{Fig:Rx_b}, respectively, illustrate the \ac{3D} structure and the top view of an \ac{ADR} consisting of seven identical receiver elements oriented in the desired spatial directions. We refer to this design as $1$-tier \ac{ADR}. The overall \ac{FOV} can be improved by incorporating additional tiers in the \ac{ADR} structure. To this end, we add tiers according to a hexagonal layout. Let $N_\mathrm{tier}$ denote the number of tiers. The total number of receiver elements is given by:
\begin{equation}
    N_\mathrm{ADR} = 1 + \sum_{i=1}^{N_\mathrm{tier}} 6i.
\end{equation}
For instance, Fig.~\subref*{Fig:Rx_c} shows a $2$-tier \ac{ADR} which consists of $N_\mathrm{ADR}=19$ elements.

Fig.~\subref*{Fig:Rx_Elem} depicts the detailed schematic diagram for each \ac{ADR} element, comprising a \ac{CPC} paired with a \ac{2D} array of square-shaped \ac{PIN} \acp{PD}. Each \ac{PD} has a photosensitive area of $A_\mathrm{PD}=D_\mathrm{PD}^2$, with $D_\mathrm{PD}$ denoting the side length of the photosensitive area. The \ac{PD} array is a square array of size $\sqrt{N_\mathrm{PD}}\times\sqrt{N_\mathrm{PD}}$, where $N_\mathrm{PD}$ is the total number of \acp{PD} per array. For a given \ac{FF} such that $0<\mathrm{FF}\leq1$, the \ac{PD} array has a total area of: 
\begin{equation*}
A = \frac{N_\mathrm{PD}A_\mathrm{PD}}{\mathrm{FF}} \raisepunct{.}
\end{equation*}
To avoid compromising the receiver bandwidth, the output signals of the \ac{PD} array for an \ac{ADR} element are combined together using \ac{EGC} by assuming that every \ac{PD} is independently equipped with a \ac{TIA}, as shown in Fig.~\subref*{Fig:Rx_Elem}. Subsequently, the \ac{MRC} method is employed to combine the output signals of the $N_\mathrm{ADR}$ elements. This is depicted in Fig.~\subref*{Fig:Rx_MRC} for a $1$-tier \ac{ADR} with $N_\mathrm{ADR}=7$. The rationale for the proposed two-stage array processing approach is twofold. First, the use of \ac{EGC} for the \ac{PD} array aims to maximise the coupling efficiency between the light collection at the entrance aperture and the light detection (i.e., optical-to-electrical conversion) by the \ac{PD} array at the exit aperture of the \ac{CPC} under the design considerations. Second, the angular diversity paths are efficiently utilised by applying \ac{MRC} to the \ac{ADR} elements to maximise the received \ac{SNR}. This way the receiver implementation complexity due to the \ac{SNR} estimation per branch and the computation of \ac{MRC} coefficients scales with the number of receiver elements, $N_\mathrm{ADR}$. Since \ac{CPC} is an essential constituent of the non-imaging \ac{ADR}, we briefly describe its working principle and geometry in the following.

\begin{figure*}[t!]
    \centering
    \begin{minipage}[b]{\linewidth}
    \centering
    \subfloat[\label{Fig:Rx_Elem} Stage 1] {\includegraphics[height=0.22\textheight, keepaspectratio=true]{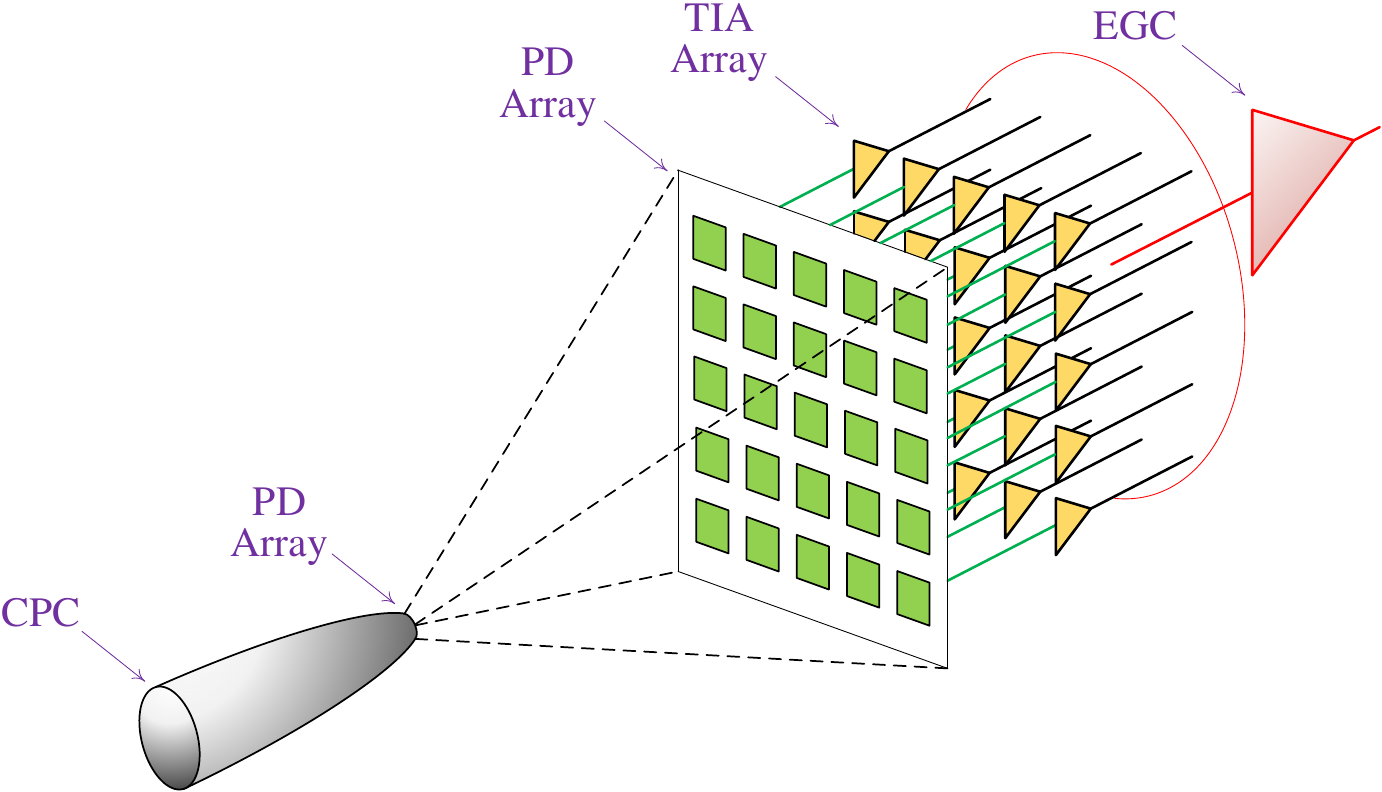}}\hspace{25pt}
    \subfloat[\label{Fig:Rx_MRC} Stage 2] {\includegraphics[height=0.22\textheight, keepaspectratio=true]{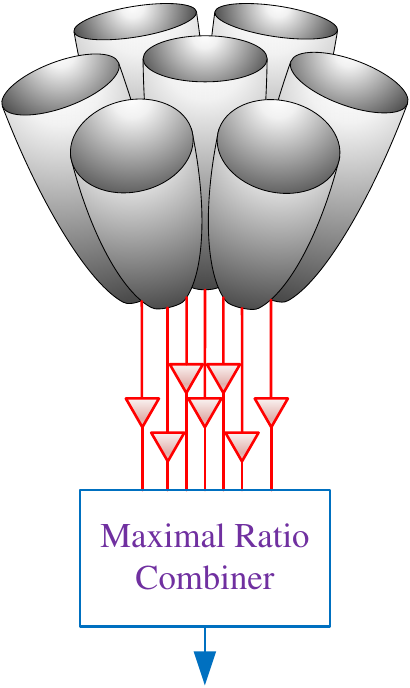}}
    \end{minipage}
    \caption{Two-stage ADR array processing: (a) Stage 1: EGC of the PD array output signals for an ADR element, (b) Stage 2: $7$-branch MRC for a $1$-tier ADR}
    \label{Fig:Rx_Array}
    \vspace{-20pt}
\end{figure*}

Fig.~\ref{Fig:CPC_CrossSection} shows the cross section of a \ac{CPC}. The outer surface of a \ac{CPC} is a parabolic arc rotated about an axis known as the \textit{CPC axis} or the \textit{rotational axis}. The outer surface is circularly symmetric about the \ac{CPC} axis with the desired diameter of the entrance aperture, $D_1$, at its light collecting end, and the diameter of the exit aperture, $D_2$, at its light concentration end. With this geometry, the incident light impinging on the entrance aperture finds its way to the exit aperture via multiple internal reflections, provided it is within the acceptance angle of the \ac{CPC}, $\theta_\mathrm{CPC}$.
The optical gain is measured by the geometrical concentration ratio, which is defined as the ratio between the areas of the entrance and exit apertures \cite{RWinston2005}, as follows:
\begin{equation}
    G_\mathrm{CPC} = \left(\frac{D_1}{D_2} \right)^2 = \frac{n^2 _\mathrm{CPC}}{\sin^2{\theta_\mathrm{CPC}}} \raisepunct{,}
    \label{Eq:CPC_gain}
\end{equation}
where $n_\mathrm{CPC}$ denotes the refractive index of the material used to fabricate the \ac{CPC} or fill it with. Here, $n_\mathrm{CPC}=1$ for a reflective hollow \ac{CPC}, and $n_\mathrm{CPC}>1$ for a dielectric \ac{CPC}. The second equality in \eqref{Eq:CPC_gain} represents the maximum theoretical gain for an ideal non-imaging concentrator\footnote{The maximum concentration gain is deduced from the edge-ray principle of non-imaging optics on the basis of the Snell's law and Fermat's principle \cite{RWinston2005}.}. The overall dimensions of the \ac{CPC}-based \ac{ADR} depends primarily on the dimensions of a \ac{CPC}. This includes the entrance aperture area which is equal to $\pi D_1^2/4$, and the length also referred to as the height of a \ac{CPC} which is given by \cite{RWinston2005}:
\begin{equation}
    L_\mathrm{CPC} = \frac{D_1+D_2}{2\tan{\theta_\mathrm{CPC}}} \raisepunct{.}
    \label{Eq:CPC_length}
\end{equation}
\begin{figure}[t!]
    \centering
    {\includegraphics[width=0.7\linewidth, keepaspectratio=true]{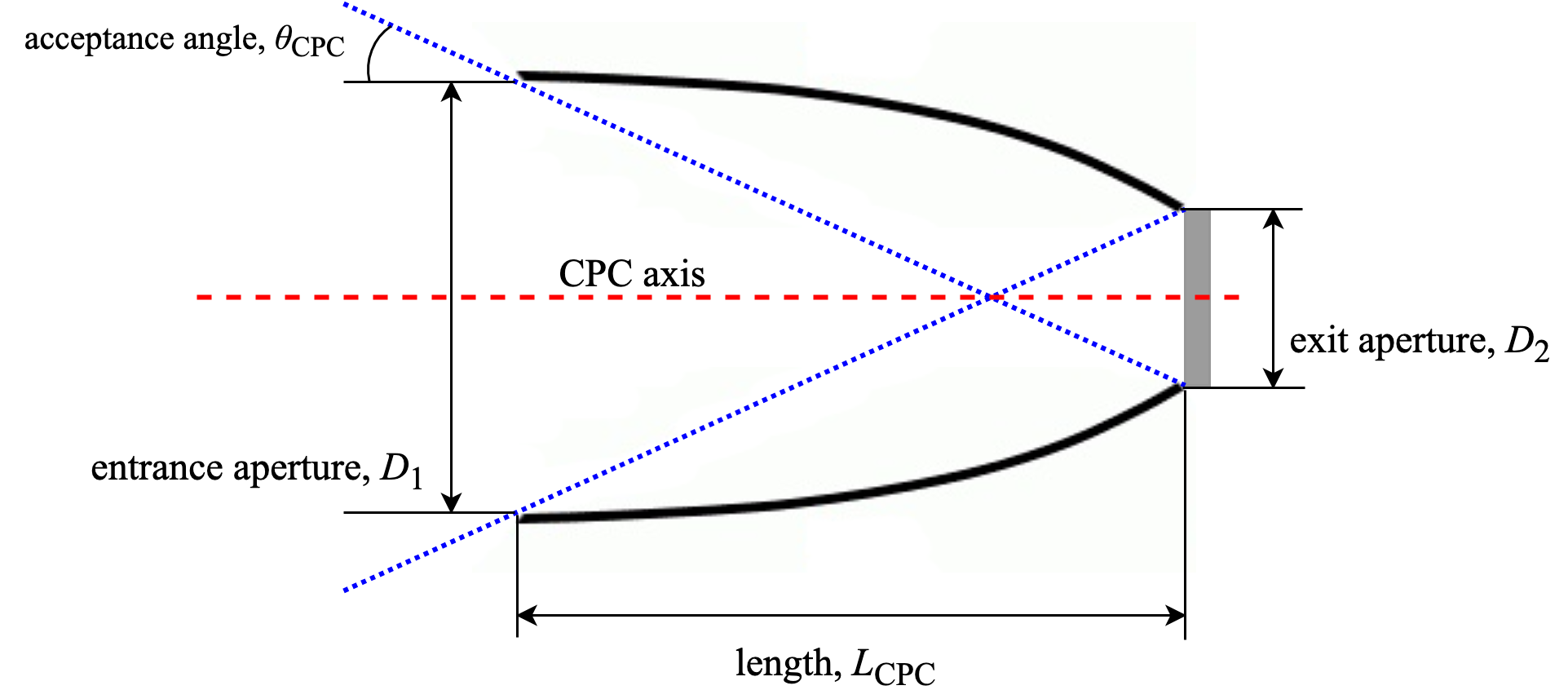}}
    \caption{Cross-section view of a CPC.}
    \label{Fig:CPC_CrossSection}
    \vspace{-20pt}
\end{figure}

%---------------------------------------------------------------------------------------------------
\subsection{Design Tradeoffs} 
Designing receivers with the aim of achieving Gb/s data rates presents a twofold challenge: 1) area-bandwidth tradeoff and 2) gain-\ac{FOV} tradeoff. These tradeoffs are briefly explained in the following.
%---------------------------------------------------------------------------------------------------
\subsubsection{Area-Bandwidth Tradeoff} \label{Sec:3_AreaBW}
The bandwidth of a \ac{PD} is expressed as \cite{StephenBAlexander1997}:
\begin{equation}
    B=\dfrac{1}{\sqrt{\left(2\pi R_{\mathrm{L}} C_\mathrm{p} \right)^2+\left( \dfrac{\ell}{0.44 v_{\mathrm{s}}}\right)^2}} \raisepunct{,}
    \label{Eq:BW}
\end{equation}
where $R_{\mathrm{L}}$ is junction series resistance plus the load resistance of the \ac{TIA}, $\ell$ denotes the thickness of the depletion region, and $v_{\mathrm{s}}$ is the carrier saturation velocity. Also, $C_\mathrm{p}$ is the junction capacitance given by $C_\mathrm{p} = \dfrac{\epsilon_0 \epsilon_{\mathrm{r}} A_\mathrm{PD}}{\ell}$ where $\epsilon_0$ is the permittivity in vacuum, $\epsilon_{\mathrm{r}}$ is the relative permittivity of the semiconductor, and $A_\mathrm{PD}$ denotes the area of the depletion region which constitutes the PD effective area. According to \eqref{Eq:BW}, as $\ell$ decreases, the junction capacitance $C_\mathrm{p}$ and therefore the left term in the denominator increases, however, the right term corresponding to the transit time of the PD decreases. Hence, there exists an optimum thickness for the depletion region which yields the maximum PD bandwidth. For this optimum value of $\ell$, denoted by $\ell_{\mathrm{opt}}$, the PD bandwidth is described as \cite{StephenBAlexander1997}:
\begin{equation}
    B=\dfrac{1}{\sqrt{\dfrac{4\pi\epsilon_0\epsilon_\mathrm{r} R_\mathrm{L}}{0.44 v_\mathrm{s}} A_\mathrm{PD}}} \raisepunct{.}
    \label{Eq:BW2}
\end{equation}
The area-bandwidth tradeoff for a \ac{PD} is readily represented by \eqref{Eq:BW2}. Note that this is in fact an upper bound of the PD bandwidth and any value of $\ell \neq \ell_{\mathrm{opt}}$ results in a lower bandwidth. For a square shaped \ac{PD} with the side length $D_{\mathrm{PD}} = \sqrt{A_{\mathrm{PD}}}$, \eqref{Eq:BW2} turns into:
\begin{equation}
    B=\dfrac{1}{K_\mathrm{PD} D_{\mathrm{PD}}} \raisepunct{,}
    \label{Eq:BW_Dpd}
\end{equation}
where $K_\mathrm{PD} = \sqrt{\dfrac{4\pi\epsilon_0\epsilon_\mathrm{r} R_\mathrm{L}}{0.44 v_\mathrm{s}}}$. Fig.~\ref{Fig:B_Dpd} illustrates the bandwidth $B$ versus $D_{\mathrm{PD}}$ for such a \ac{PD}. This graph is plotted based on \eqref{Eq:BW_Dpd} using the parameter values adopted from \cite{StephenBAlexander1997}. As shown in Fig.~\ref{Fig:B_Dpd}, $B$ exhibits a rapidly decreasing behaviour with respect to $D_\mathrm{PD}$. Therefore, in order to maximise the \ac{PD} bandwidth, the junction capacitance has to be minimised. Specifically, for bandwidths higher than $10$~GHz, a side length of less than $50$~{\textmu}m is required. This necessitates a very small \ac{PD} area and imposes a major challenge in the design of laser-based optical wireless receivers, since free-space optical signals need to be collected, aligned with and coupled into the confined photosensitive area of a miniaturised \ac{PD} with minimal loss.
\begin{figure*}[t!]
    \centering
    \begin{minipage}[b]{0.49\linewidth}
    {\includegraphics[width=\textwidth, keepaspectratio=true]{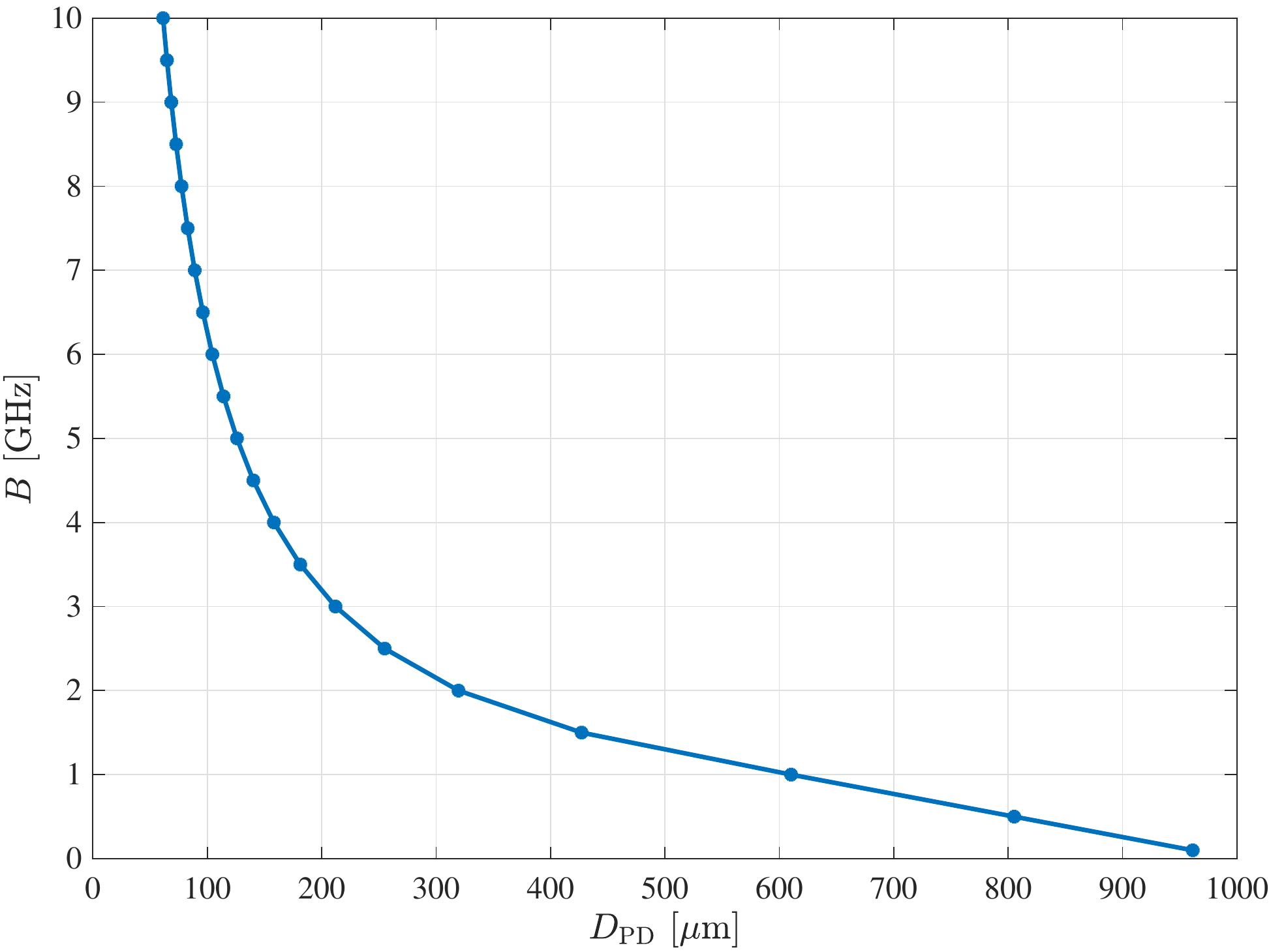}}
    \vspace{-20pt}
    \caption{PD bandwidth versus its side length.}
    \label{Fig:B_Dpd}
    \end{minipage}\hfill
    \begin{minipage}[b]{0.49\linewidth}
    \includegraphics[width=\textwidth, keepaspectratio=true]{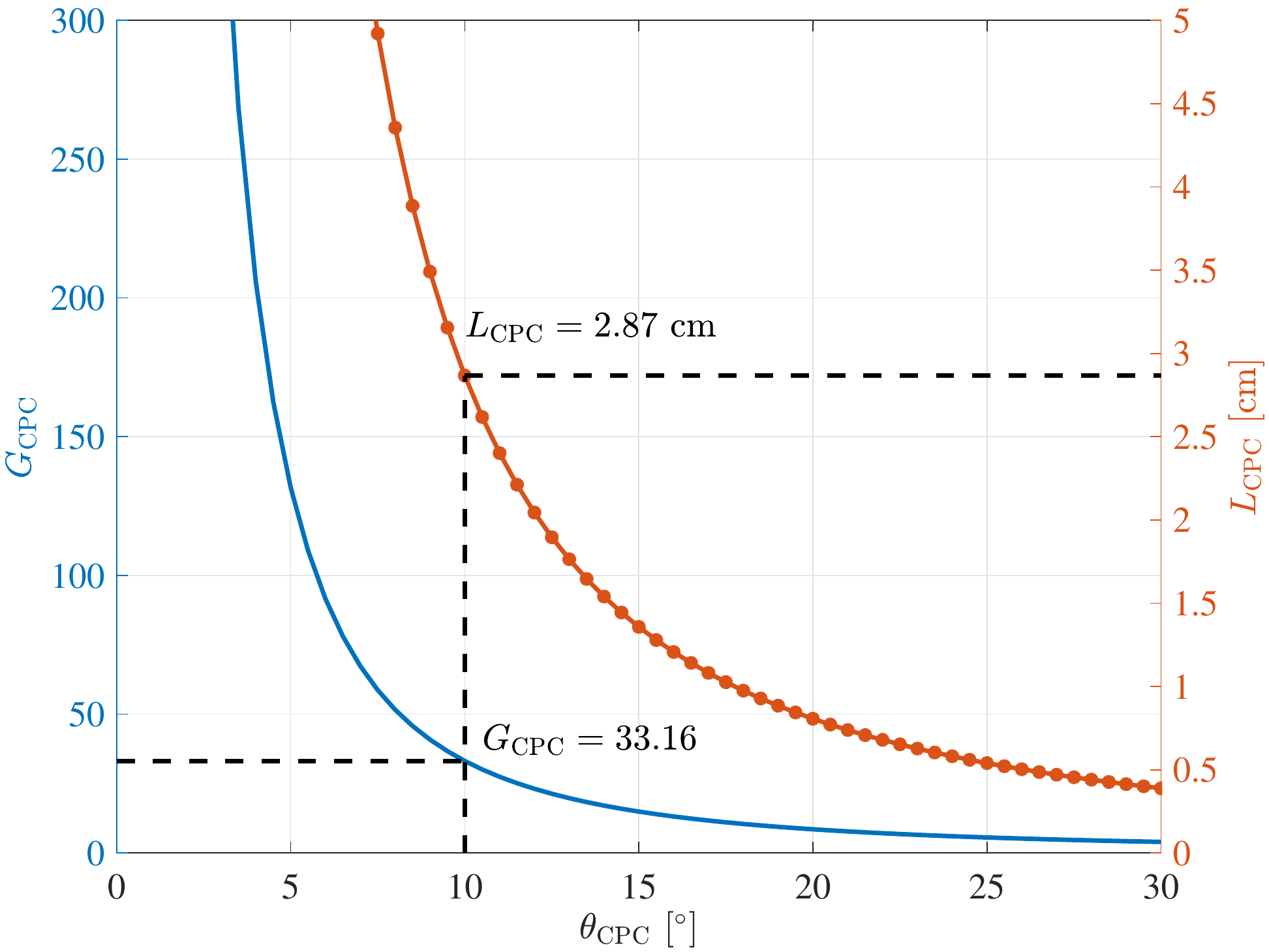}
    \vspace{-20pt}
    \caption{Gain-\ac{FOV} tradeoff of a CPC with $n_\mathrm{CPC}=1$ and $D_2 = 1.5$~mm.}
    \label{Fig:G_FOV_CPC}
    \end{minipage}
    \vspace{-20pt}
\end{figure*}
%---------------------------------------------------------------------------------------------------
\subsubsection{Gain-FOV Tradeoff} 
The small area of a high-bandwidth \ac{PD} can be compensated by using appropriate imaging or non-imaging optics. This improves the optical power collection efficiency of the optical receiver and thus the received \ac{SNR}. 
However, the use of light concentrators for increasing the collection area limits the receiver \ac{FOV} due to the law of conservation of etendue \cite{RWinston2005}. Fig.~\ref{Fig:G_FOV_CPC} illustrates the gain-FOV tradeoff for a CPC, where the gain and the length of a \ac{CPC} with $D_2=1.5$~mm are plotted against the acceptance angle, $\theta_\mathrm{CPC}$. This \ac{CPC} provides an optical gain of about $33$ for $\theta_\mathrm{CPC}=10^\circ$. Such a \ac{CPC} is $2.87$~cm long and it has an effective collection area of about $0.58$~cm$^2$.

%---------------------------------------------------------------------------------------------------
\subsection{Design Parameters}
The primary design parameters under consideration are the \ac{PD} bandwidth $B$, the half-angle \ac{FOV} of the \ac{ADR}, the total number of \acp{PD} per array $N_\mathrm{PD}$, and the number of ADR tiers $N_\mathrm{tier}$. The secondary parameters and how they are related to the primary design parameters are enlisted in the following.

Based on the \ac{ADR} structure, as shown in Fig.~\ref{Fig:Rx}, the half-angle \ac{FOV} is given by:
\begin{equation}
    \mathrm{FOV} = \theta_\mathrm{CPC} + N_\mathrm{tier} \theta_\mathrm{tilt},
    \label{Eq:FOV_theta_tilt}
\end{equation}
where $\theta_\mathrm{tilt}$ is the tilt angle between the adjacent ADR tiers. Assuming that the acceptance cones of the adjacent ADR tiers are touching but not overlapping, $\theta_{\mathrm{tilt}} = 2\theta_\mathrm{CPC}$. This assumption ensures a full angular acceptance for the ADR design.
Therefore, according to \eqref{Eq:FOV_theta_tilt}, the half-angle FOV of the ADR is related to the acceptance angle of \acp{CPC} by $\mathrm{FOV} = \theta_\mathrm{CPC} (2N_\mathrm{tier}+1)$. Consequently, the secondary parameter $\theta_\mathrm{CPC}$ is expressed in terms of the primary design parameter $\mathrm{FOV}$ as:
\begin{equation}
    \theta_\mathrm{CPC}= \dfrac{\mathrm{FOV}}{2N_\mathrm{tier}+1}\raisepunct{.}
    \label{Eq:FOV_theta}
\end{equation}
Noting that the \ac{ADR} \ac{FOV} satisfies $\mathrm{FOV}\leq\dfrac{\pi}{2}$, a useful corollary is deduced from \eqref{Eq:FOV_theta}.
%-----------------------------------------
% Corollary 1
%-----------------------------------------
\begin{corollary}\label{Corollary:1}
For the \ac{ADR} design with $N_\mathrm{tier}$ tiers, $\theta_\mathrm{CPC}$ is upper bounded as follows:
\begin{equation*}
    \theta_\mathrm{CPC} \leq \frac{\pi}{6} \raisepunct{.}
\end{equation*}
\end{corollary}
\noindent The diameter of the exit aperture of the \ac{CPC}, $D_2$, is related to the \ac{PD} array parameters via:
\begin{equation}
D_2 = D_\mathrm{PD} \sqrt{\dfrac{N_\mathrm{PD}}{\mathrm{FF}}} \raisepunct{,}
\end{equation}
where $\mathrm{FF}$ denotes the PD array \ac{FF}. From \eqref{Eq:BW_Dpd}, $D_\mathrm{PD} = \dfrac{1}{K_\mathrm{PD} B}$. Thus, $D_2$ can be expressed as:
\begin{equation}
    D_2 = \dfrac{1}{K_\mathrm{PD} B} \sqrt{\dfrac{N_\mathrm{PD}}{\mathrm{FF}}} \raisepunct{.} 
    \label{Eq:CPC_D2}
\end{equation}
Based on \eqref{Eq:CPC_gain}, the diameter of the entrance aperture of the \ac{CPC} is then given by:
\begin{equation}
    D_1 = D_2 \left(\dfrac{n_\mathrm{CPC}}{\sin\theta_\mathrm{CPC}}\right) = \dfrac{1}{K_\mathrm{PD} B} \sqrt{\dfrac{N_\mathrm{PD}}{\mathrm{FF}}} \left(\dfrac{n_\mathrm{CPC}}{\sin\theta_\mathrm{CPC}}\right) \raisepunct{.}  
    \label{Eq:CPC_D1}
\end{equation}
Substituting \eqref{Eq:CPC_D2} and \eqref{Eq:CPC_D1} into \eqref{Eq:CPC_length} yields:
\begin{equation}
    L_\mathrm{CPC} = 
    \dfrac{1}{2 K_\mathrm{PD} B} \sqrt{\dfrac{N_\mathrm{PD}}{\mathrm{FF}}} \left(\dfrac{\dfrac{n_\mathrm{CPC}}{\sin \theta_\mathrm{CPC}}+1}{\tan \theta_\mathrm{CPC}}\right) \raisepunct{.}
    \label{Eq:L_CPC2}
\end{equation}
The overall size of an \ac{ADR} with $N_\mathrm{tier}$ tiers can be described by the length $L_\mathrm{ADR} \approx L_\mathrm{CPC}$, and the top collection area $A_\mathrm{ADR}$. By using \eqref{Eq:CPC_D1} and \eqref{Eq:L_CPC2}, these are derived as:
\begin{equation}
    L_\mathrm{ADR} \approx \dfrac{K_1}{B} \left(\dfrac{n_\mathrm{CPC}+\sin \theta_\mathrm{CPC}}{\sin \theta_\mathrm{CPC}\tan \theta_\mathrm{CPC}}\right),
    \label{Eq:L_ADR}
\end{equation}
\begin{equation}
    A_\mathrm{ADR} 
    = \frac{\pi D^2 _1}{4} \left(1 + \sum_{i=1}^{N_\mathrm{tier}} 6i\cos(2i\theta_\mathrm{CPC}) \right) 
    = \frac{K_2}{B^2 \sin^2\theta_\mathrm{CPC}} \left(1 + \sum_{i=1}^{N_\mathrm{tier}} 6i\cos(2i\theta_\mathrm{CPC}) \right),
    \label{Eq:A_ADR}
\end{equation}
where $K_1=\dfrac{1}{2 K_\mathrm{PD}}\sqrt{\dfrac{N_\mathrm{PD}}{\mathrm{FF}}}$ and $K_2=\dfrac{\pi N_{\mathrm{PD}} n_\mathrm{CPC}^2}{4\mathrm{FF} K_\mathrm{PD}^2}$. Note that the tilt angle of the $i$th ADR tier relative to the CPC axis of the central ADR element is $\theta_i=i\theta_\mathrm{tilt}=2i\theta_\mathrm{CPC}$.

%%%%%%%%%%%%%%%%%%%%%%%%%%%%%%%%%%%%%%%%%%%%%%%%%%%%%%%%%%%%%%%%%%%%%%%%%%%%%%%%%%%%%%%%%%%%%%%%%%%%
%%%%%%%%%%%%%%%%%%%%%%%%%%%%%%%%%%%%%%%%%%%%%%%%%%%%%%%%%%%%%%%%%%%%%%%%%%%%%%%%%%%%%%%%%%%%%%%%%%%%
\section{Performance Analysis and Optimisation} \label{Sec:4}
For performance analysis, given the triple $(B,\mathrm{FOV},N_\mathrm{PD})$, we first need to determine the \ac{ADR} dimensions by taking the following steps. We thereupon proceed with performance optimisation.
\begin{enumerate}
\item For the given $\mathrm{FOV}$, the \ac{CPC} acceptance angle is obtained from \eqref{Eq:FOV_theta}. 
\item For the given $B$, the corresponding $D_\mathrm{PD}$ is obtained according to \eqref{Eq:BW_Dpd}.
\item The exit and entrance apertures of the \ac{CPC} are calculated based on \eqref{Eq:CPC_D2} and \eqref{Eq:CPC_D1}.
\item The height $L_\mathrm{ADR}$ and the overall area $A_\mathrm{ADR}$ are evaluated based on \eqref{Eq:L_ADR} and \eqref{Eq:A_ADR}. 
\end{enumerate}

In a fully aligned link, the middle element of the receiver mainly collects the incident optical power. To calculate the received optical power, the effective area of the receiver is approximated by the entrance aperture area of a CPC, within a circle of radius $\dfrac{D_1}{2}$. The corresponding $P_\mathrm{r}$ is thereby calculated based on \eqref{Eq:Pr} for $\rho_0=\dfrac{D_1}{2}$ in combination with \eqref{Eq:FOV_theta} and \eqref{Eq:CPC_D1}. It follows that:
\begin{equation}
    \mkern-18mu P_{\mathrm{r}} =\mathrm{FF} \times P_{\mathrm{t}} \left( {1-\exp{\left( -\dfrac{N_\mathrm{PD} n^2 _\mathrm{CPC}}{2 \mathrm{FF} \left[K_\mathrm{PD} B \sin\left(\dfrac{\mathrm{FOV}}{2N_\mathrm{tier}+1}\right) {w'}(D) \right]^2} \right)}} \right) \raisepunct{.}
    \label{Eq:Pr2}
\end{equation}

To ensure a high spectral efficiency with intensity modulation and direct detection, we assume the use of \ac{DC}-biased optical \ac{OFDM} in conjunction with adaptive \ac{QAM}. By properly choosing the variance of the \ac{OFDM} signal and the \ac{DC} bias, the electrical \ac{SNR} is given by:
\begin{equation}
    \mathrm{SNR} = \frac{\left(R_\mathrm{PD}P_\mathrm{r}\right)^2}{\sigma_\mathrm{n}^2} \raisepunct{,}
    \label{Eq:SNR}
\end{equation}
where $R_\mathrm{PD}$ is the PD responsivity; and $\sigma_\mathrm{n}^2 = N_0 B$ is the total noise variance with $N_0$ denoting the total noise \ac{PSD}. The total noise \ac{PSD} is \cite{ESarbazi2020tb}:
\begin{equation}
	N_0 = \frac{4\kappa T}{R_\mathrm{L}}F_\mathrm{n} N_\mathrm{PD} + 
	2q_\mathrm{e} R_\mathrm{PD}P_\mathrm{r}+
	\mathrm{RIN}\left(R_\mathrm{PD}P_\mathrm{r}\right)^2.
	\label{Eq:N0}
\end{equation}
In (\ref{Eq:N0}), the first term corresponds to the thermal noise, the second term is the shot noise \ac{PSD} of the receiver, and the third term is the \ac{PSD} of the \ac{RIN} of the \ac{VCSEL}, defined as the mean square of intensity fluctuations of the laser light normalised to the squared average intensity \cite{LColdren2012}. Also, $\kappa$ is the Boltzmann constant; $T$ is temperature in Kelvin; $R_\mathrm{L}$ is the load resistance; $F_\mathrm{n}$ is the noise figure of the \ac{TIA}; and $q_\mathrm{e}$ is the elementary charge.
The proposed receiver architecture entails using \ac{PD}-\ac{TIA} pairs, as shown in Fig.~\subref*{Fig:Rx_Elem}. In this case, the total noise is essentially dominated by the receiver thermal noise \cite{ESarbazi2020tb}. Therefore, (\ref{Eq:N0}) can be approximated by $N_0 \approx \dfrac{4\kappa T}{R_\mathrm{L}}F_\mathrm{n} N_\mathrm{PD}$. The achievable rate is given by:
\begin{equation}
    R = B\log_2\left(1+\frac{\mathrm{SNR}}{\Gamma}\right) = B\log_2\left(1 + \frac{\left(R_\mathrm{PD}P_\mathrm{r}\right)^2}{\Gamma N_0 B }\right),
    \label{Eq:Rate}
\end{equation}
where $B$ is the single-sided bandwidth of the OWC system which is determined by the receiver bandwidth\footnote{Here, the system bandwidth is assumed to be limited by the receiver as typically a large modulation bandwidth is available at the \ac{VCSEL}-based transmitter.}; and $\Gamma$ denotes the \ac{SNR} gap required to guarantee the target \ac{BER} performance. 

%---------------------------------------------------------------------------------------------------
\subsection{Rate Maximisation}
We now formulate an optimisation problem to maximise the achievable data rate for a \textit{given} PD array size $N_\mathrm{PD}$ and \ac{ADR} size $N_\mathrm{ADR}$. The aim of this optimisation is to find optimum values of $B$ and $\mathrm{FOV}$ while satisfying a number of design constraints.

%---------------------------------------------------------------------------------------------------
\subsubsection{With FOV constraint}
Consider the case where the only constraint in place is the minimum \ac{FOV} requirement. In this case, the rate maximisation problem can be stated as:
\begin{subequations}
\label{Eq:Max_1}
\begin{align}   
    \operatorname*{arg\;max}_{(B,\mathrm{FOV})} \ \ \ & R = B\log_2\left(1 + \frac{\left(R_\mathrm{PD}P_\mathrm{r}\right)^2}{\Gamma N_0 B }\right) \label{Eq:Max_1a}\\
    \rm{s.t.} \ \ \ & {\mathrm{FOV}} \geq \mathrm{FOV}_{\mathrm{min}} \label{Eq:Max_1b}
\end{align}
\end{subequations}
where $P_\mathrm{r}$ is given by \eqref{Eq:Pr2} which is a function of $B$ and $\mathrm{FOV}$. Due to the non-convexity of \eqref{Eq:Max_1a} in $(B,\mathrm{FOV})$, this is a non-convex optimisation problem. Let the objective function be denoted by $R = f(B,\mathrm{FOV})$ for brevity. The following proposition brings out a key characteristic of $R$.
%-----------------------------------------
% Proposition 1
%-----------------------------------------
\begin{proposition}\label{Proposition:1}
$R$ is a monotonically decreasing function of $\mathrm{FOV}$.
\end{proposition}
\begin{proof}
    See Appendix~\ref{App:1}.
\end{proof}
\noindent The following lemma establishes a key result that is used to simplify the non-convex optimisation problem in \eqref{Eq:Max_1}.
%-----------------------------------------
% Lemma 1
%-----------------------------------------
\begin{lemma}\label{Lemma:1}
The solution to the rate maximisation problem in \eqref{Eq:Max_1} lies at the boundary of the feasible region.
\end{lemma}
%-----------------------------------------
% Proof of Lemma 1
%-----------------------------------------
\begin{proof}
From Proposition~\ref{Proposition:1}, it can be concluded that for any given value of $B = B_0$, the objective function $R=f(B_0,\mathrm{FOV})$ under $\mathrm{FOV}\geq\mathrm{FOV}_\mathrm{min}$ takes its maximum value at $(B_0,\mathrm{FOV}_\mathrm{min})$. Therefore, $R$ is always maximised at the boundary for $\mathrm{FOV} = \mathrm{FOV}_\mathrm{min}$.
\end{proof}
%-----------------------------------------
% Theorem 1
%-----------------------------------------
\begin{theorem}\label{Theorem:1}
The optimisation problem in \eqref{Eq:Max_1} is simplified to an unconstrained single variable optimisation problem:
\begin{equation}
    \operatorname*{arg\;max}_{B} \ \ \  R = f(B, \mathrm{FOV}_{\mathrm{min}}).
    \label{Eq:Max_1s}
\end{equation}
\end{theorem}
%-----------------------------------------
% Proof of Theorem 1
%-----------------------------------------
\begin{proof}
It readily follows from Lemma~\ref{Lemma:1} that \eqref{Eq:Max_1} for $\mathrm{FOV}=\mathrm{FOV}_{\mathrm{min}}$ reduces to \eqref{Eq:Max_1s}.
\end{proof}
\noindent Although the optimisation problem of \eqref{Eq:Max_1s} is still non-convex, its solution can be numerically computed using a one-dimensional search over the range of interest for $B$.
%---------------------------------------------------------------------------------------------------
\subsubsection{With FOV and overall dimensions constraint}
When there are additional constraints on the physical dimensions of the \ac{ADR} as well as the \ac{FOV} constraint, the rate maximisation problem is formulated as:
\begin{subequations}
\label{Eq:Max_2}
\begin{align}   
    \operatorname*{arg\;max}_{(B,\mathrm{FOV})} \ \ \ & R = f(B,\mathrm{FOV}) \label{Eq:Max_2a}\\
    \rm{s.t.} \ \ \ & {\mathrm{FOV}} \geq \mathrm{FOV}_{\mathrm{min}} \label{Eq:Max_2b}\\
    & L_\mathrm{ADR} \leq L_{\mathrm{max}} \label{Eq:Max_2c}\\
    & A_\mathrm{ADR} \leq A_{\mathrm{max}} \label{Eq:Max_2d}
\end{align}
\end{subequations}
The second and the third constraints of \eqref{Eq:Max_2} are intended to fulfil the design objectives for the top area and the overall height of the \ac{ADR}. The objective function in \eqref{Eq:Max_2a} and the two constraints in \eqref{Eq:Max_2c} and \eqref{Eq:Max_2d} are all non-convex. As a result, the optimisation problem in \eqref{Eq:Max_2} is non-convex.
By substituting \eqref{Eq:L_ADR} in \eqref{Eq:Max_2c} and \eqref{Eq:A_ADR} in \eqref{Eq:Max_2d}, both with equality, the optimisation variable $B$ is separately derived as an explicit function of $\mathrm{FOV}$ as follows:
\begin{equation}
B = f_\mathrm{L}(\mathrm{FOV}) = \dfrac{K_1}{L_{\mathrm{max}}} \left[ \dfrac{n_\mathrm{CPC} + \sin\left(\dfrac{\mathrm{FOV}}{2N_\mathrm{tier}+1}\right)}{\sin\left(\dfrac{\mathrm{FOV}}{2N_\mathrm{tier}+1}\right) \tan\left(\dfrac{\mathrm{FOV}}{2N_\mathrm{tier}+1}\right)} \right] \raisepunct{,}
\label{Eq:f_L}
\end{equation}
\begin{equation}
B = f_\mathrm{A}(\mathrm{FOV}) = \dfrac{1}{\sin\left(\dfrac{\mathrm{FOV}}{2N_\mathrm{tier}+1}\right)}\sqrt{\dfrac{K_2}{A_\mathrm{max}} \left[ 1 + \sum_{i=1}^{N_\mathrm{tier}} 6i\cos{\left(\dfrac{\mathrm{2FOV}}{2N_\mathrm{tier}+1}\right)} \right]} \raisepunct{.}
\label{Eq:f_A}
\end{equation}
Note that \eqref{Eq:f_L} and \eqref{Eq:f_A} represent boundaries of the two constraints on $L_\mathrm{ADR}$ and $A_\mathrm{ADR}$.
%-----------------------------------------
% Proposition 2
%-----------------------------------------
\begin{proposition}\label{Proposition:2}
$L_{\mathrm{ADR}}$ and $A_{\mathrm{ADR}}$ are monotonically decreasing functions of $B$ and $\mathrm{FOV}$.
\end{proposition}
\begin{proof}
    See Appendix~\ref{App:2}.
\end{proof}
\noindent The following lemma aims to unify the three constraints in \eqref{Eq:Max_2b}--\eqref{Eq:Max_2d} based on Proposition~\ref{Proposition:2}.
%-----------------------------------------
% Lemma 2
%-----------------------------------------
\begin{lemma}
The optimisation problem in \eqref{Eq:Max_2} can be reformulated with a single \ac{FOV} constraint:
\begin{subequations}
\label{Eq:Max_2r}
\begin{align}   
    \operatorname*{arg\;max}_{(B,\mathrm{FOV})} \ \ \ & R = f(B,\mathrm{FOV}) \label{Eq:Max_2ra}\\
    \rm{s.t.} \ \ \ & \mathrm{FOV} \geq f_\mathrm{FOV}(B) \label{Eq:Max_2rb}
\end{align}
\end{subequations}
with the boundary function $f_\mathrm{FOV}(B) = \max\left\{ \mathrm{FOV}_\mathrm{min}, f_\mathrm{L}^{-1}(B), f_\mathrm{A}^{-1}(B) \right\}$, where $f_\mathrm{L}^{-1}(B)$ and $f_\mathrm{A}^{-1}(B)$ are the inverse functions of $B = f_\mathrm{L}(\mathrm{FOV})$ and $B = f_\mathrm{A}(\mathrm{FOV})$ given in \eqref{Eq:f_L} and \eqref{Eq:f_A}.
\end{lemma}
%-----------------------------------------
% Proof of Lemma 2
%-----------------------------------------
\begin{proof}
Based on Proposition~\ref{Proposition:2}, the dimensions constraints in \eqref{Eq:Max_2c} and \eqref{Eq:Max_2d} can be transformed into their equivalent \ac{FOV} constraints. By defining:
\begin{subequations}
\begin{align}
    \mathcal{S}_1 &= \left\{ {(B,\mathrm{FOV}) \Big\lvert \mathrm{FOV} \geq \mathrm{FOV}_\mathrm{min}}\right\} \raisepunct{,} \\
    \mathcal{S}_2 &= \left\{ {(B,\mathrm{FOV}) \Big\lvert L_\mathrm{ADR}(B,\mathrm{FOV}) \leq L_\mathrm{max}}\right\} = \left\{ {(B,\mathrm{FOV}) \Big\lvert  \mathrm{FOV} \geq {f_\mathrm{L}^{-1}}(B) }\right\} \label{Eq:App3_S1} \raisepunct{,} \\
    \mathcal{S}_3 &= \left\{ {(B,\mathrm{FOV}) \Big\lvert A_\mathrm{ADR}(B,\mathrm{FOV}) \leq A_\mathrm{max}}\Big\} = \Big\{ {(B,\mathrm{FOV}) \Big\lvert  \mathrm{FOV} \geq {f_\mathrm{A}^{-1}}(B) }\right\} \label{Eq:App3_S2} \raisepunct{,}
\end{align}
\end{subequations}
the feasible set of the optimisation problem in \eqref{Eq:Max_2} takes the following form:
\begin{equation}
    \mathcal{S} = \mathcal{S}_1 \cap \mathcal{S}_2 \cap \mathcal{S}_3 = 
    \left\{ (B,\mathrm{FOV}) \Big\lvert \mathrm{FOV} \geq \max\left\{ \mathrm{FOV}_\mathrm{min}, f_\mathrm{L}^{-1}(B), f_\mathrm{A}^{-1}(B) \right\} \right\} \raisepunct{.}
\end{equation}
The boundary of $\mathcal{S}$ as a function of $B$ is:
\begin{equation}
f_\mathrm{FOV}(B) = \max\left\{ \mathrm{FOV}_\mathrm{min}, f_\mathrm{L}^{-1}(B), f_\mathrm{A}^{-1}(B) \right\} \raisepunct{.}
\label{Eq:FOV_unified}
\end{equation}
This completes the proof.
\end{proof}
\noindent The following lemma offers further simplification for the optimisation problem in \eqref{Eq:Max_2r}.
%-----------------------------------------
% Lemma 3
%-----------------------------------------
\begin{lemma}\label{Lemma:2}
The solution to the rate maximisation problem in \eqref{Eq:Max_2r} lies at the boundary of the feasible region.
\end{lemma}
%-----------------------------------------
% Proof of Lemma 3
%-----------------------------------------
\begin{proof}
Since $R=f(B,\mathrm{FOV})$ is a decreasing function of $\mathrm{FOV}$ as shown in Proposition~\ref{Proposition:1}, for any given value of $B = B_0$, $R=f(B_0,\mathrm{FOV})$ under $\mathrm{FOV}\geq f_\mathrm{FOV}(B_0)$ is maximised at $(B_0,f_\mathrm{FOV}(B_0))$. Hence, the maximum of $R$ is located at the boundary of the feasible region.
\end{proof}
%-----------------------------------------
% Theorem 2
%-----------------------------------------
\begin{theorem} \label{Theorem:2}
The optimisation problem in \eqref{Eq:Max_2r} is simplified to an unconstrained single variable optimisation problem:
\begin{equation}
    \operatorname*{arg\;max}_{B} \ \ \  R = f(B, f_\mathrm{FOV}(B))
    \label{Eq:Max_2s}
\end{equation}
\end{theorem}
%-----------------------------------------
% Proof of Theorem 2
%-----------------------------------------
\begin{proof}
With the aid of Lemma~\ref{Lemma:2}, it suffices to evaluate \eqref{Eq:Max_2r} for $\mathrm{FOV}=f_\mathrm{FOV}(B)$.
\end{proof}
\noindent The solution to the non-convex optimisation problem in \eqref{Eq:Max_2s} can be efficiently computed using numerical methods by way of a one-dimensional search along the $B$ axis.
%%%%%%%%%%%%%%%%%%%%%%%%%%%%%%%%%%%%%%%%%%%%%%%%%%%%%%%%%%%%%%%%%%%%%%%%%%%%%%%%%%%%%%%%%%%%%%%%%%%%
%%%%%%%%%%%%%%%%%%%%%%%%%%%%%%%%%%%%%%%%%%%%%%%%%%%%%%%%%%%%%%%%%%%%%%%%%%%%%%%%%%%%%%%%%%%%%%%%%%%%
\begin{table}[t!]
    \begin{minipage}{.5\linewidth}
    \centering
    \caption{Simulation Parameters}
    \resizebox{\columnwidth}{!}{%
	\begin{tabular}{|c|l|l|}
	\hline
		\textbf{Parameter}    & \textbf{Description}                & \textbf{Value} \\ \hline\hline
		$D$                   & Link distance                       & $3$ m          \\ 
		$w_0$                 & Incident beam waist radius          & $10$ {\textmu}m  \\
% 		$n_\mathrm{Lens}$     & Refractive index of the lens        & $1.55$         \\
% 		$w' _0$               & Transformed beam waist radius       & $4.1$ {\textmu}m  \\
		$\lambda$             & Laser wavelength                    & $950$ nm       \\ 
            $P_\mathrm{t}$        & Optical power of the VCSEL          & $10$ mW        \\
   		$n_\mathrm{CPC}$      & CPC refractive index                & $1.7$          \\
            $R_\mathrm{PD}$       & PD responsivity                     & $0.6$ A$/$W    \\
            $F_\mathrm{n}$        & TIA noise figure                    & $5$ dB   \\		
		$\Gamma$              & SNR gap                             & $2.6$       \\
		$\mathrm{BER}$        & Bit error ratio                     & $3.8\times10^{-3}$      \\ \hline
	\end{tabular}
	\label{Tab:1}}
    \end{minipage}%
    \begin{minipage}{.49\linewidth}
    \renewcommand{\arraystretch}{1.4}%
    \centering
	\caption{ADR Configurations}
	\resizebox{\columnwidth}{!}{%
    \begin{tabular}{l|c|c|c|c|c|}
    \cline{2-6}
        & \textbf{Configuration} & $\boldsymbol{N_\mathrm{tier}}$ & $\boldsymbol{N_\mathrm{ADR}}$ & \textbf{PD Array Size} & \textbf{Total PDs} \\ 
        \hhline{:-::=====:}
        \multicolumn{1}{|l||}{\multirow{3}{*}{\rotatebox[origin=c]{90}{\textbf{Single-Tier}}}}
        & Config. 1              & $1$   & $7$              & $2 \times 2$    & $28$      \\ \cline{2-6}
        \multicolumn{1}{|l||}{}
        & Config. 2              & $1$   & $7$              & $4 \times 4$    & $112$     \\ \cline{2-6}                              
        \multicolumn{1}{|l||}{}
        & Config. 3              & $1$   & $7$              & $8 \times 8$    & $448$     \\
        \hhline{:=::=====:}
        \multicolumn{1}{|l||}{\multirow{3}{*}{\rotatebox[origin=c]{90}{\textbf{Multi-Tier}}}}
        & Config. 4              & $2$   & $19$             & $2 \times 2$    & $76$      \\ \cline{2-6} 
        \multicolumn{1}{|l||}{}
        & Config. 5              & $2$   & $19$             & $4 \times 4$    & $304$     \\ \cline{2-6}                                
        \multicolumn{1}{|l||}{}
        & Config. 6              & $3$   & $37$             & $2 \times 2$    & $148$     \\
        \hhline{|-||-----|}
\end{tabular}
\label{Tab:2}}
\end{minipage} 
\vspace{-20pt}
\end{table}

%%%%%%%%%%%%%%%%%%%%%%%%%%%%%%%%%%%%%%%%%%%%%%%%%%%%%%%%%%%%%%%%%%%%%%%%%%%%%%%%%%%%%%%%%%%%%%%%%%%%
%%%%%%%%%%%%%%%%%%%%%%%%%%%%%%%%%%%%%%%%%%%%%%%%%%%%%%%%%%%%%%%%%%%%%%%%%%%%%%%%%%%%%%%%%%%%%%%%%%%%
\section{Numerical Results and Discussions} \label{Sec:5}
This section provides numerical results for the narrow-beam \ac{OWC} link configuration described in Section~\ref{Sec:3}. The transmitter is composed of a single \ac{VCSEL} followed by a plano-convex lens of focal length $f = 33$~mm. The transmitter design is adopted from \cite{ESarbazi2020tb}. For eye safety assessment, the \ac{MHP} is specified as $10$~cm and the maximum permissible optical power of the \ac{VCSEL} is calculated as $P_\mathrm{t,max} = 16$~mW \cite{Henderson2003}. Here, we assume $P_\mathrm{t} = 10$~mW, which satisfies the eye safety condition. The \ac{VCSEL}-based transmitter provides a Gaussian beam spot of diameter $20$~cm at the receiver, which is located at distance $D=3$~m from the transmitter. Also, without loss of generality, we assume the use of dielectric CPCs, since the optical gain of a dielectric CPC is improved by a factor of $n_\mathrm{CPC} ^2$ compared to a reflective hollow CPC based on \eqref{Eq:CPC_gain}. Moreover, $\mathrm{FF}=0.7$ for each PD array. The rest of the simulation parameters are listed in Table~\ref{Tab:1}. In order to evaluate the receiver design, we use various ADR configurations with different number of tiers and PD array sizes, as introduced in Table~\ref{Tab:2}. For each configuration, the total number of PDs is calculated as $N_\mathrm{ADR}\times N_\mathrm{PD}$.
%---------------------------------------------------------------------------------------------------

\begin{figure*}[t!]
    \centering
    \begin{minipage}[b]{\linewidth}
    \subfloat[\label{Fig:Surf_a} $R(B,\mathrm{FOV})$] {\includegraphics[width=0.33\textwidth, keepaspectratio=true]{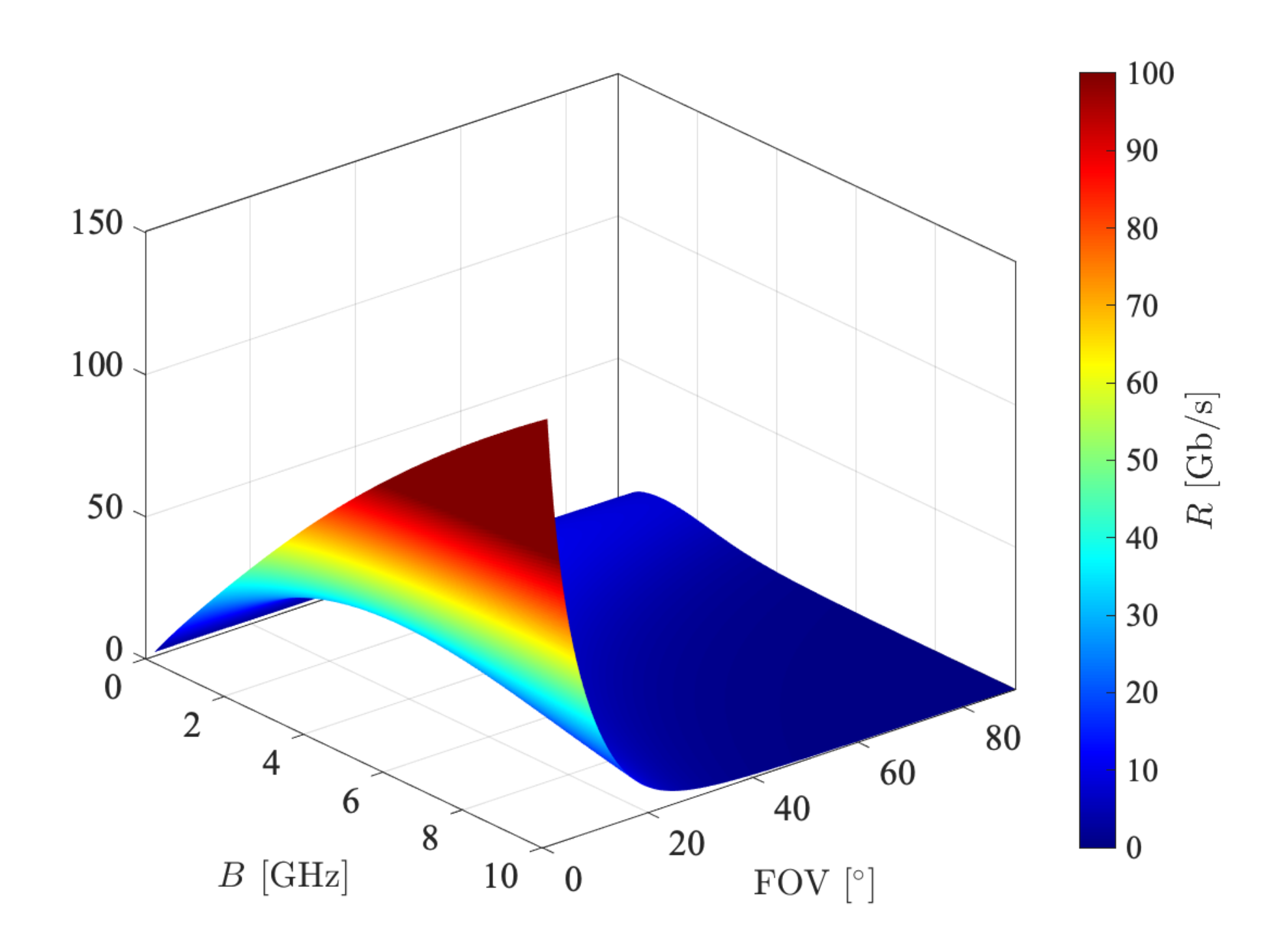}}
    \subfloat[\label{Fig:Surf_b} $L_\mathrm{ADR}(B,\mathrm{FOV})$] {\includegraphics[width=0.33\textwidth, keepaspectratio=true]{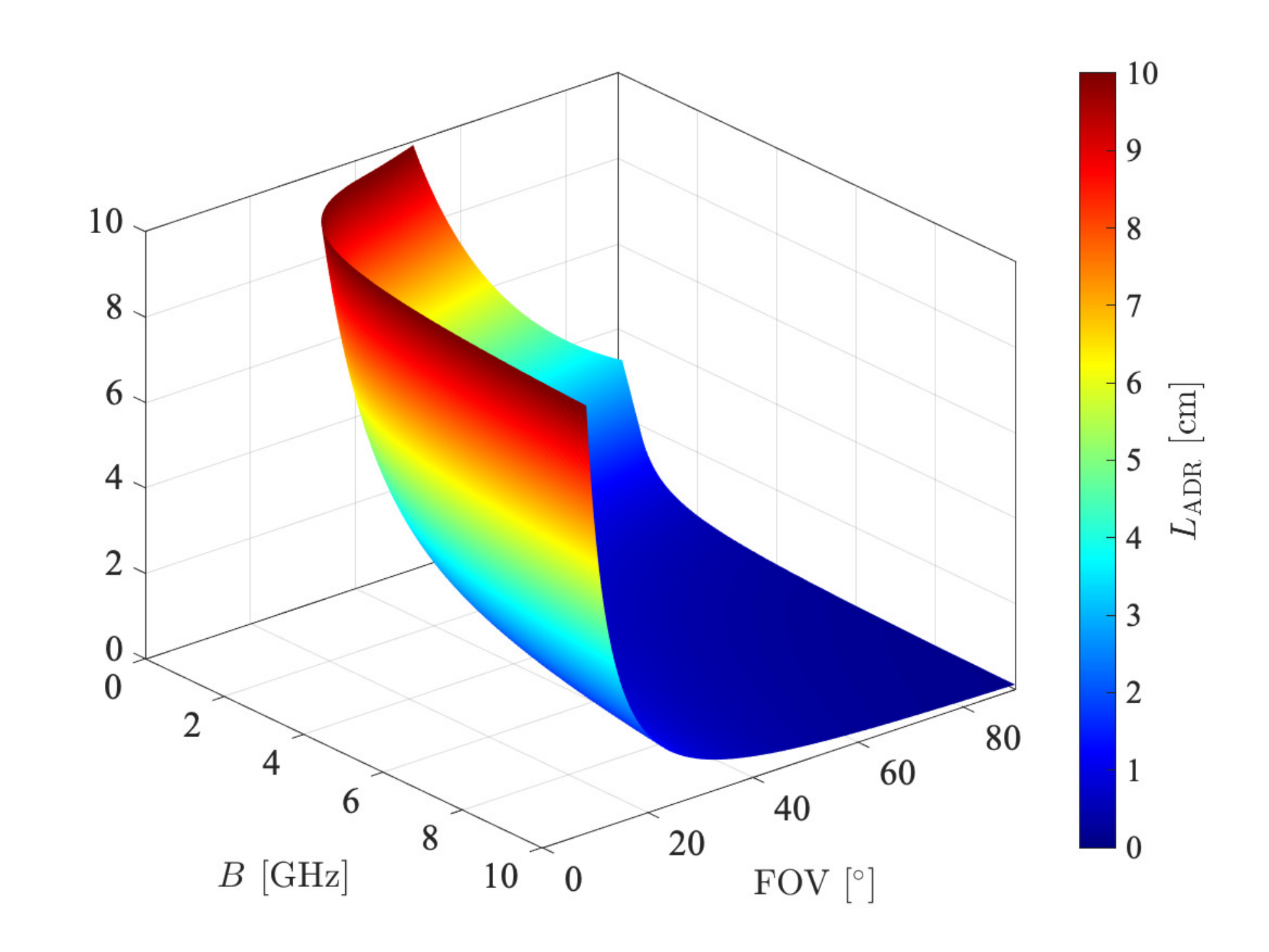}}
    \subfloat[\label{Fig:Surf_c} $A_\mathrm{ADR}(B,\mathrm{FOV})$] {\includegraphics[width=0.33\textwidth, keepaspectratio=true]{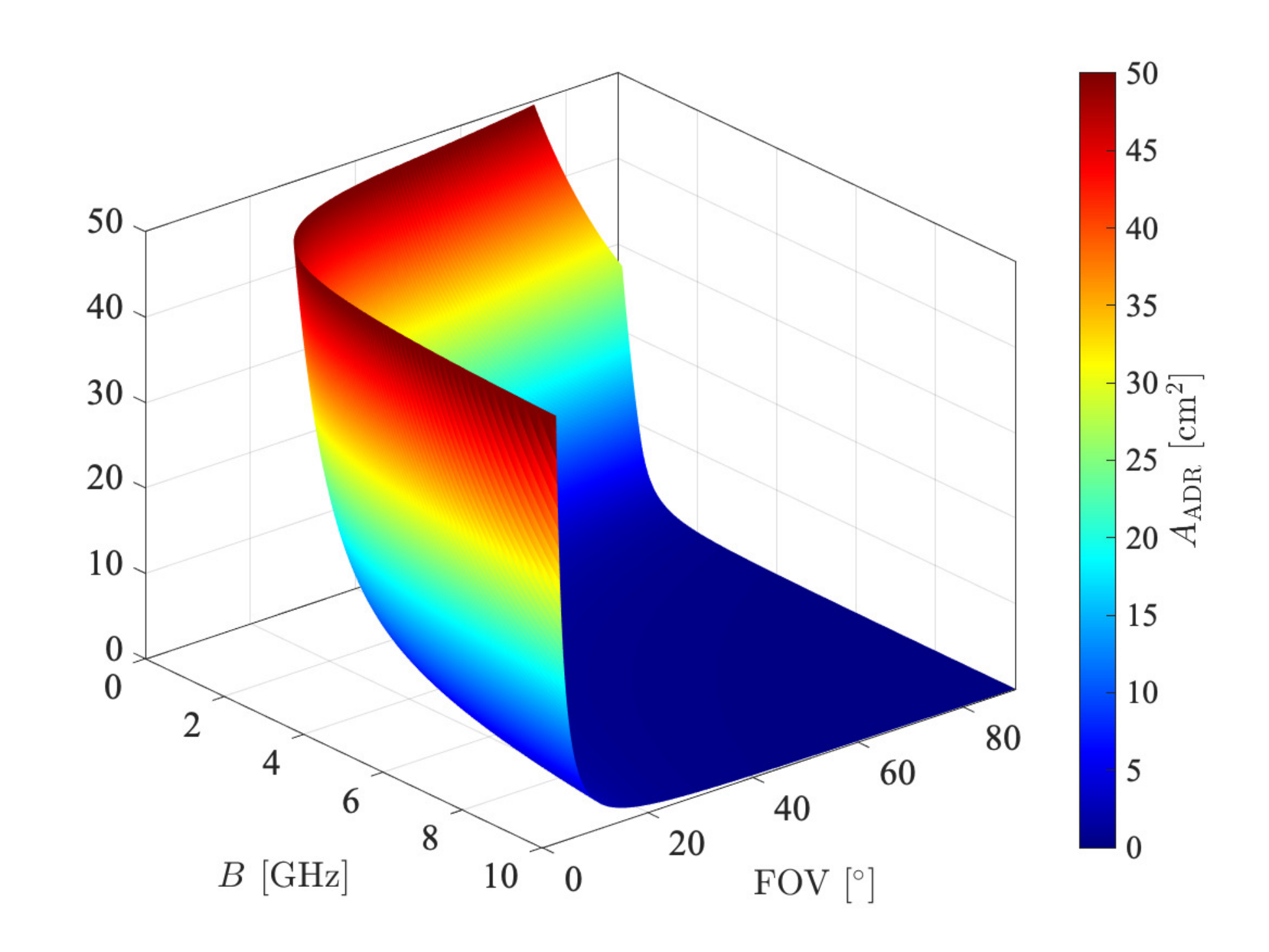}}
    \end{minipage}
    \begin{minipage}[b]{\linewidth}
    \subfloat[\label{Fig:Cont_a} $R(B,\mathrm{FOV})$] {\includegraphics[width=0.33\textwidth, keepaspectratio=true]{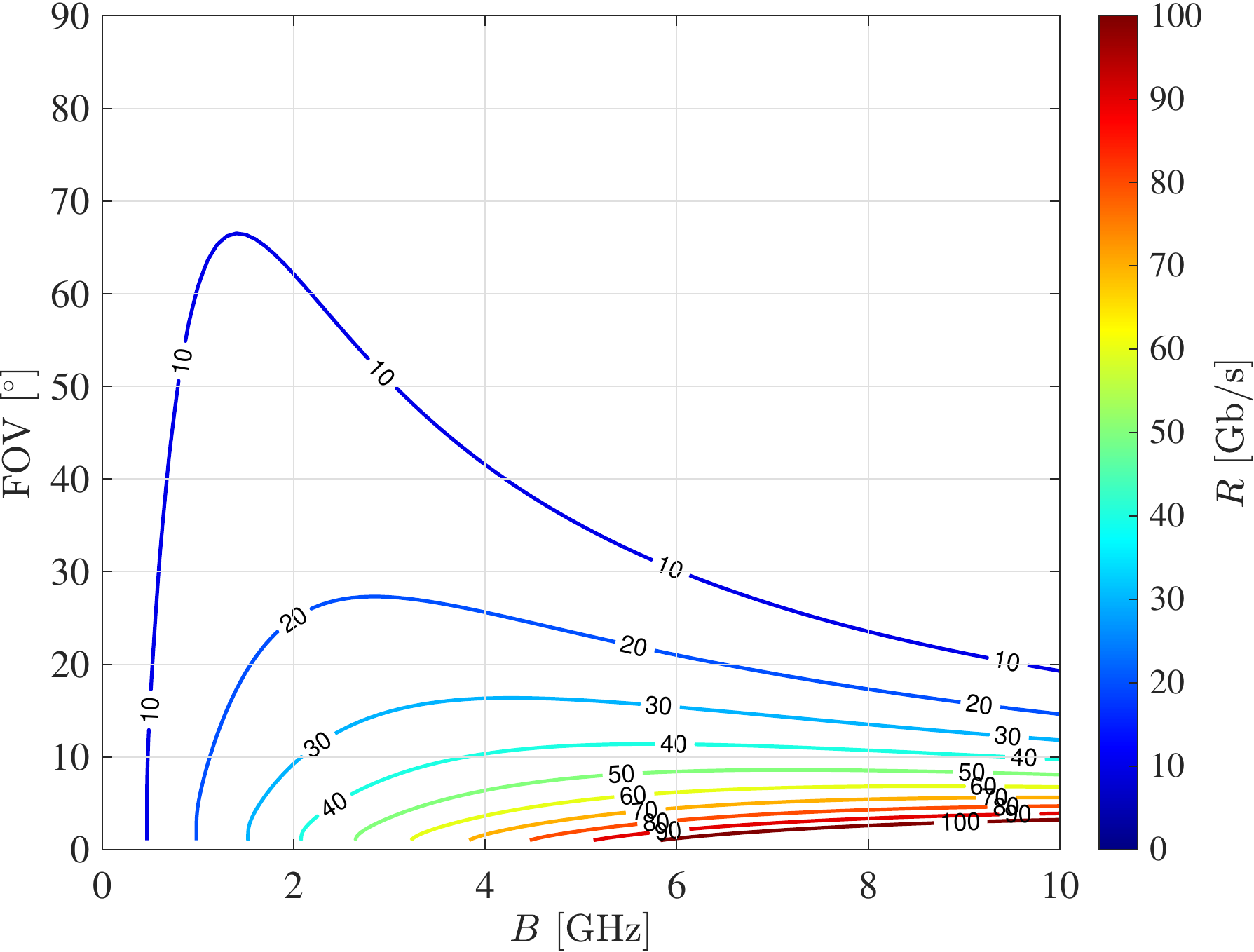}}
    \subfloat[\label{Fig:Cont_b} $L_\mathrm{ADR}(B,\mathrm{FOV})$] {\includegraphics[width=0.33\textwidth, keepaspectratio=true]{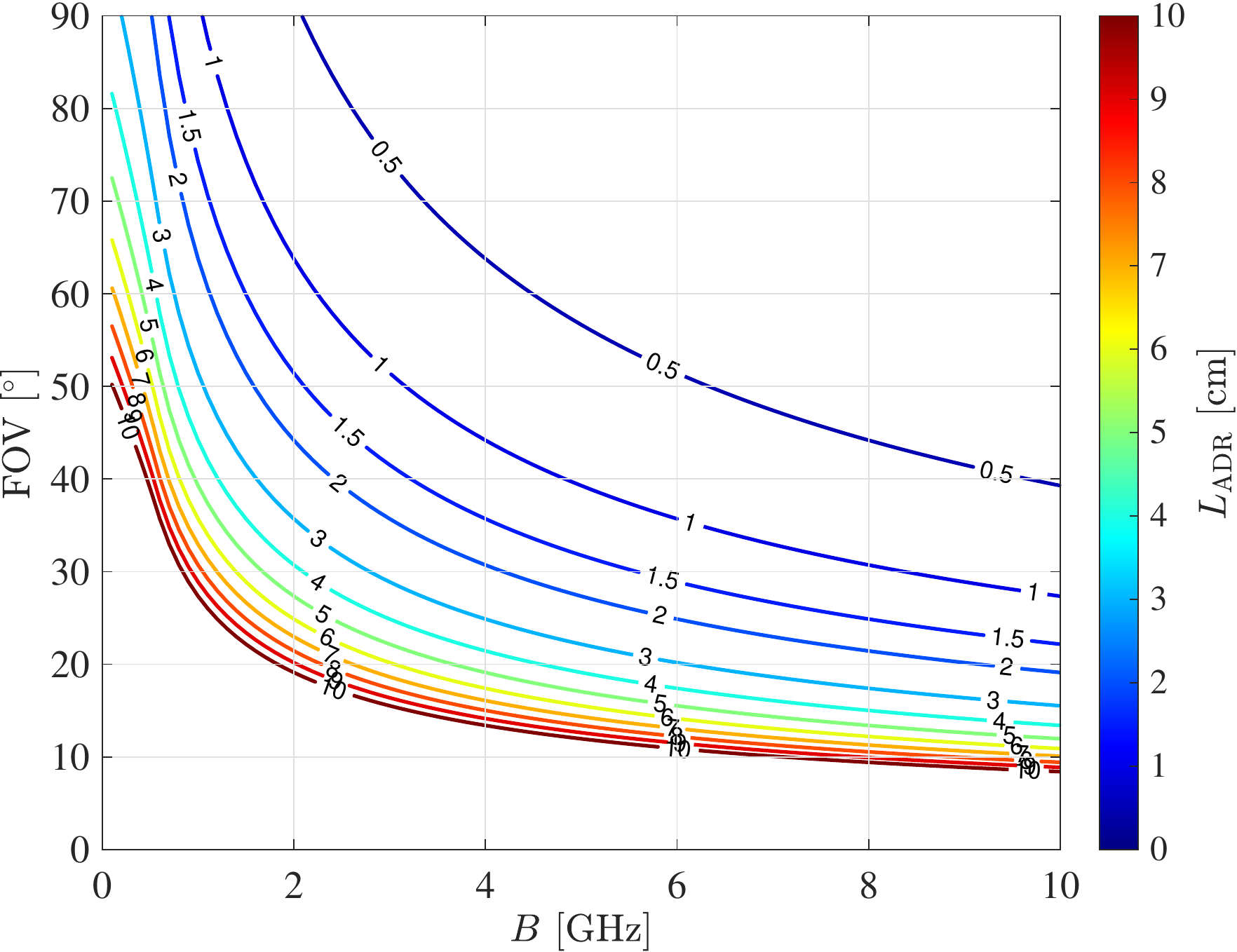}}
    \subfloat[\label{Fig:Cont_c} $A_\mathrm{ADR}(B,\mathrm{FOV})$] {\includegraphics[width=0.33\textwidth, keepaspectratio=true]{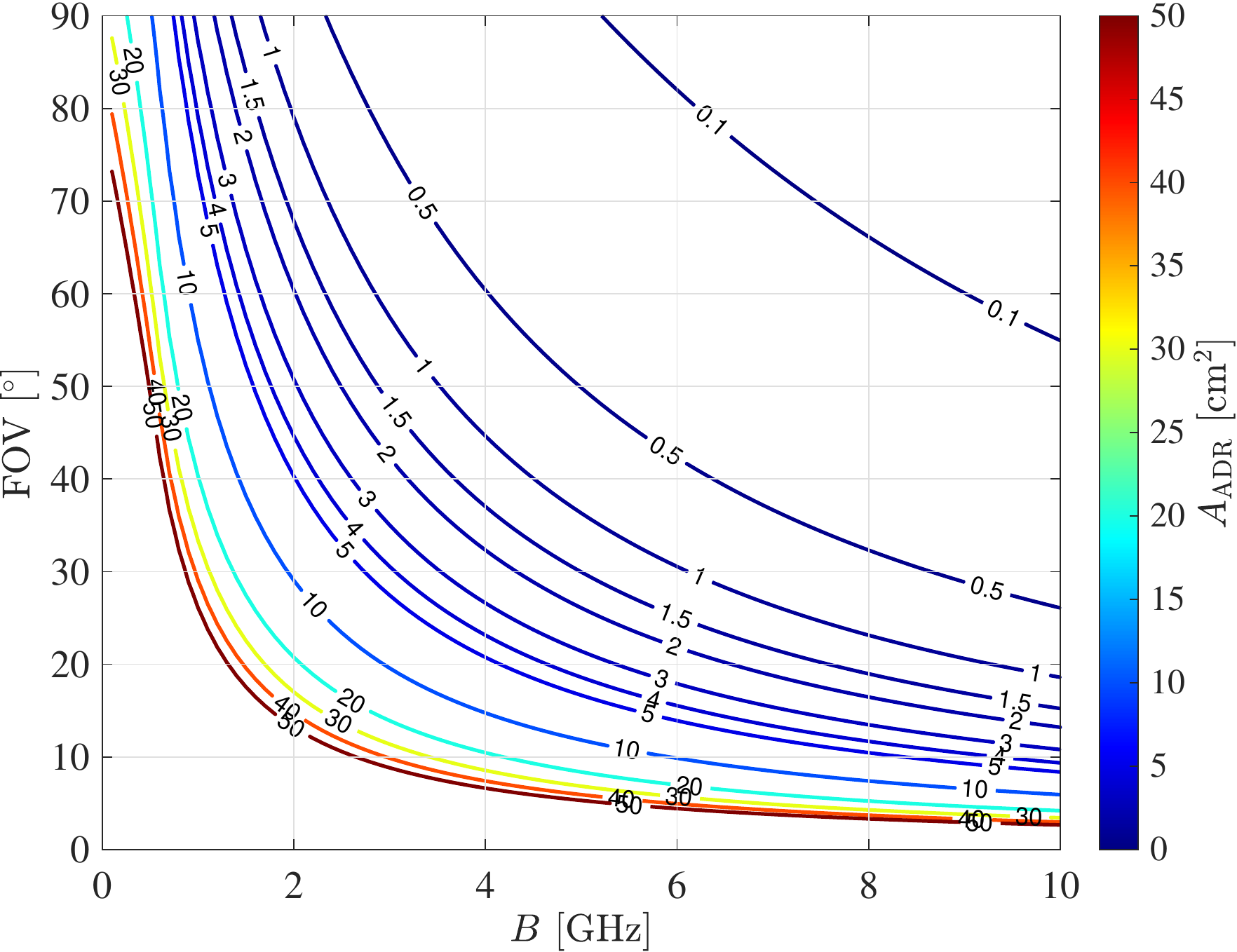}}
    \caption{The achievable data rate $R$, the overall height $L_\mathrm{ADR}$ and the total effective area $A_\mathrm{ADR}$ as a function of the bandwidth $B$ and $\mathrm{FOV}$ for Config.~2: (a)--(c) surface plots, (d)--(f) contour plots.}
    \label{Fig:Cont}
    \end{minipage}
    \vspace{-20pt}
\end{figure*}

%---------------------------------------------------------------------------------------------------
\subsection{Achievable Data Rate and Overall ADR Dimensions}
Fig.~\subref*{Fig:Surf_a}--\subref*{Fig:Surf_c} illustrate \ac{3D} surface plots of the achievable data rate $R$, the total effective area $A_\mathrm{ADR}$ and the overall height of the receiver $L_\mathrm{ADR}$, respectively, as a function of $B$ and $\mathrm{FOV}$ for a $1$-tier ADR with Config.~2, which uses $4\times4$ PD arrays (i.e., $N_\mathrm{PD} = 16$). The results are obtained by evaluating $R$, $L_\mathrm{ADR}$ and $A_\mathrm{ADR}$ based on \eqref{Eq:Rate}, \eqref{Eq:L_ADR} and \eqref{Eq:A_ADR}, as per the four-step procedure laid out in Section~\ref{Sec:4}. To interpret the information contained in these \ac{3D} surfaces, the corresponding contour plots\footnote{A contour plot is a graphical representation of a \ac{3D} surface on a \ac{2D} plane. To clarify, consider a dependent variable $Z$ as a function of two independent variables $X$ and $Y$. A contour line on the $X$-$Y$ plane interconnects all the $(X,Y)$ coordinates corresponding to the same value of $Z$. Alternatively, contour lines can be viewed as intersections of the 3D surface with planes parallel to the $X$-$Y$ plane for different values of $Z$, which are projected onto the $X$-$Y$ plane.} are presented in Fig.~\subref*{Fig:Cont_a}--\subref*{Fig:Cont_c}. From Fig.~\subref*{Fig:Cont_a}, it can be observed how the achievable data rate evolves, indicating the point that attaining higher data rates necessitates lower \acp{FOV}, which in fact underlines a tradeoff between the achievable rate and \ac{FOV}. For example, at $R=10$~Gb/s, $\mathrm{FOV}\leq65^\circ$ is realisable, whereas when $R=20$~Gb/s, the resultant \ac{FOV} is no greater than $28^\circ$. By comparison, the constant-height contours specify regions towards the top-right corner of the $B$-$\mathrm{FOV}$ plane where $L_\mathrm{ADR}$ is less than the given values, as shown in Fig.~\subref*{Fig:Cont_b}. The constant-area contours follow a similar trend, as shown in Fig.~\subref*{Fig:Cont_c}. Comparing this trend with that in Fig.~\subref*{Fig:Cont_a} reveals another tradeoff between scaling down the receiver dimensions and pushing up the achievable data rate. This tradeoff is elucidated via the following discussion on the rate maximisation under joint FOV and dimensions constraints.

\begin{figure}[t!]
    \centering
    \subfloat[\label{Fig:R_B_FixedFOV} ]{\includegraphics[width=0.33\textwidth, keepaspectratio=true]{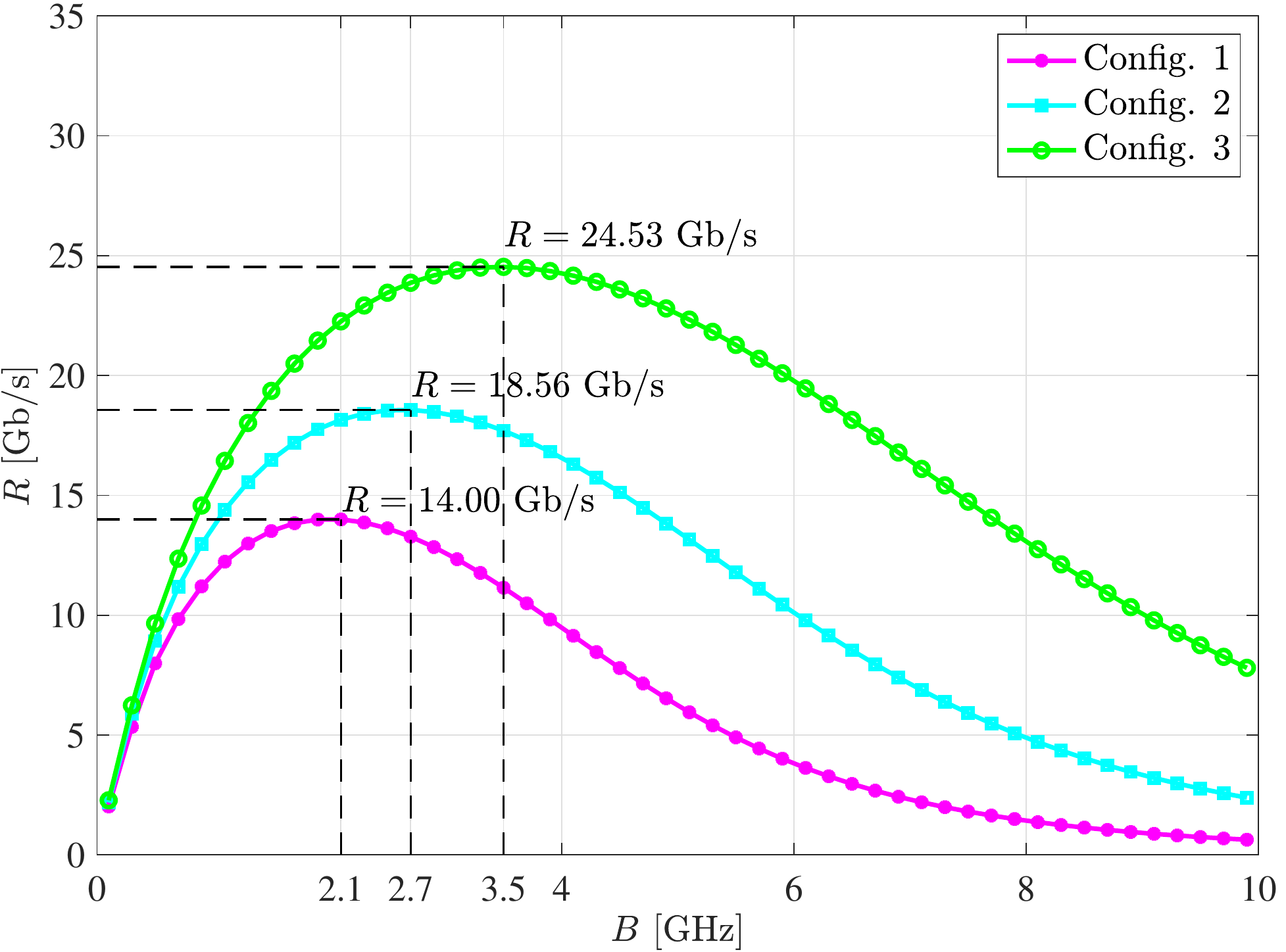}}\hfill
    \subfloat[\label{Fig:L_B_FixedFOV} ]{\includegraphics[width=0.33\textwidth, keepaspectratio=true]{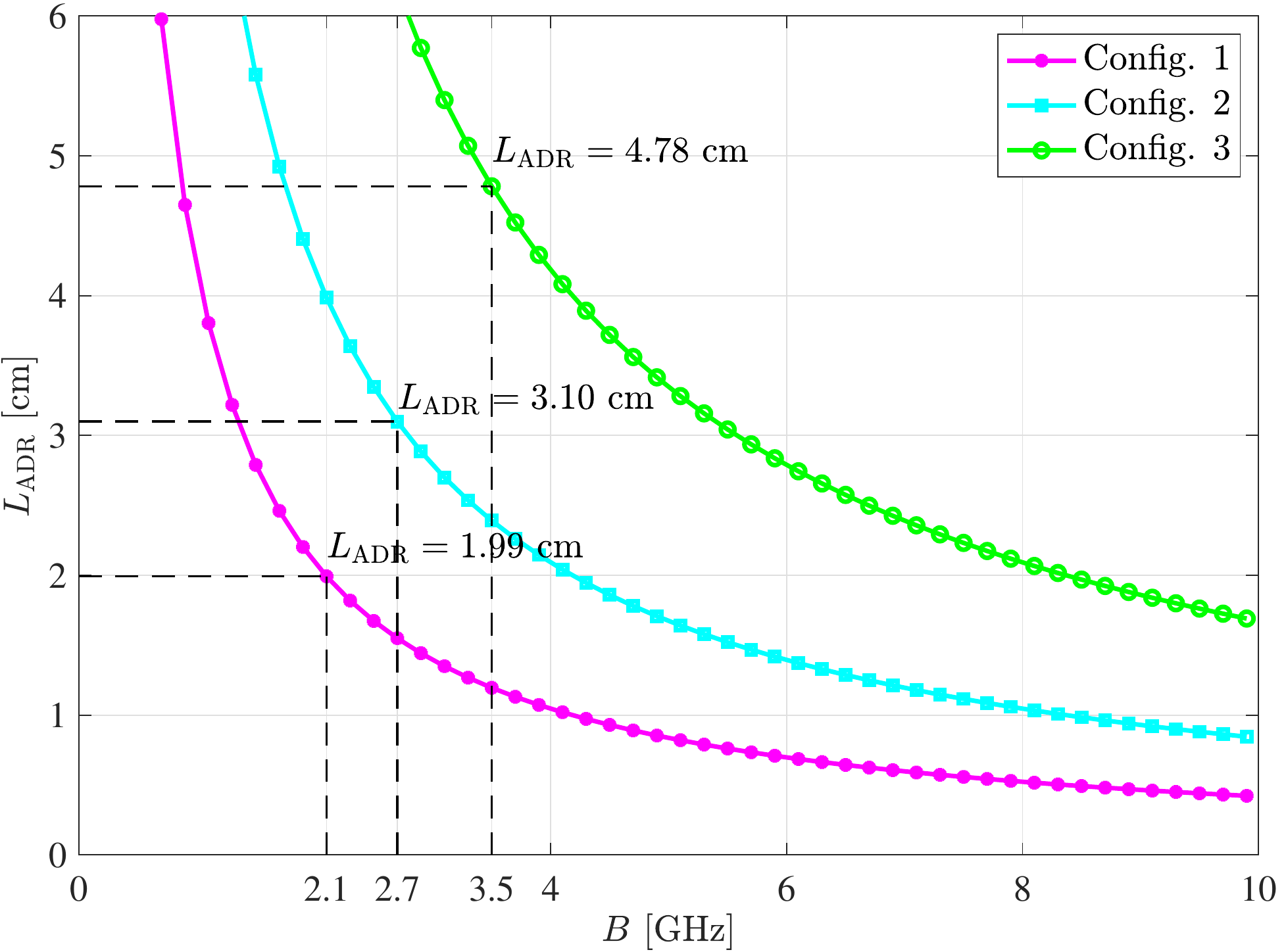}}\hfill
    \subfloat[\label{Fig:A_B_FixedFOV} ]{\includegraphics[width=0.33\textwidth, keepaspectratio=true]{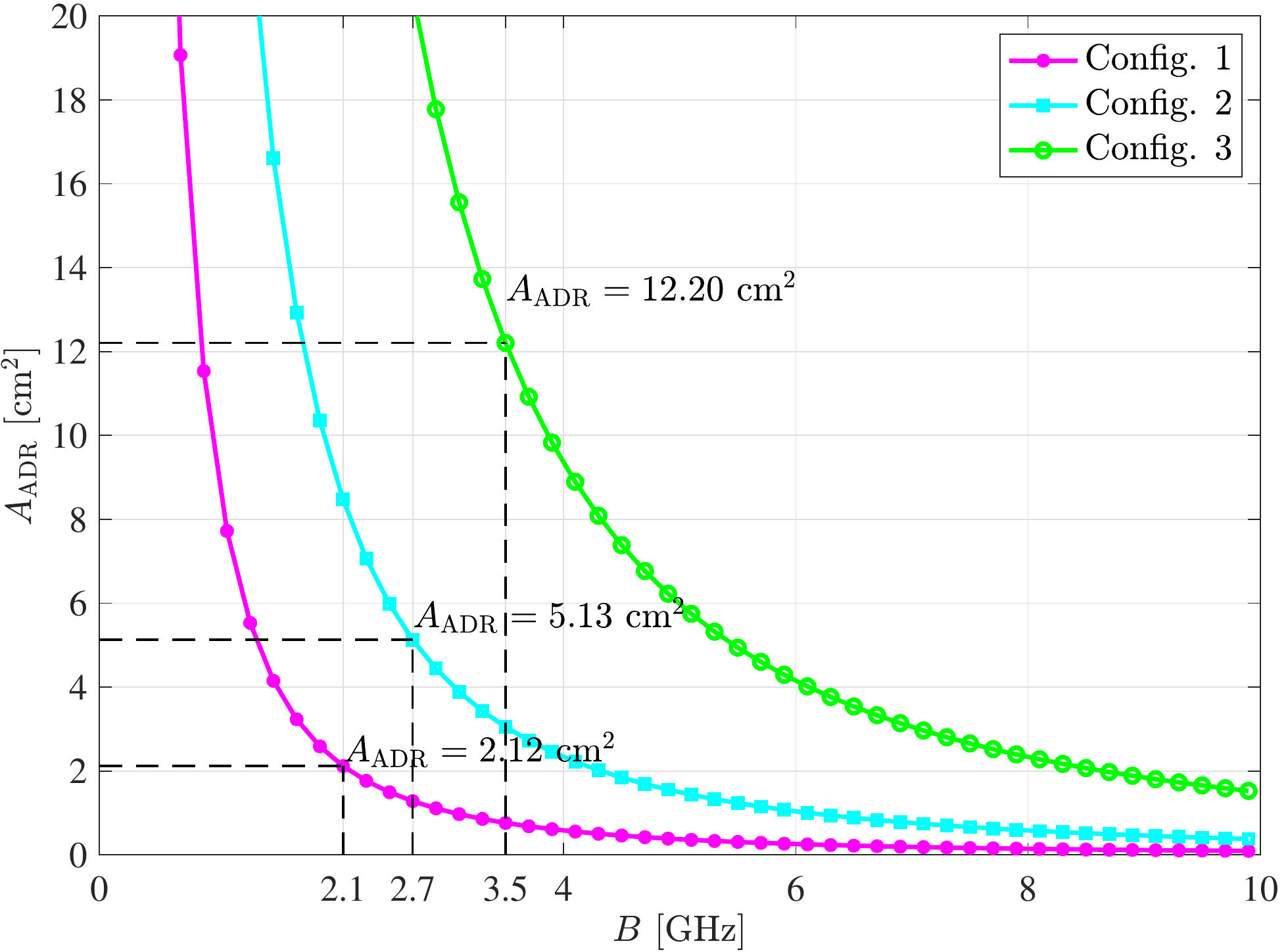}}
    \caption{The achievable data rate $R$, the overall height $L_\mathrm{ADR}$ and the total effective area $A_\mathrm{ADR}$ as a function of the bandwidth $B$ for $\mathrm{FOV}=30^\circ$ for Configs.~1, 2 and 3 based on $2\times2$, $4\times4$ and $8\times8$ \ac{PD} arrays, respectively.}
    \label{Fig:SurfCrossSec_FixedFOV}
    \vspace{-20pt}
\end{figure}

Another perspective to look at the results in Figs.~\subref*{Fig:Surf_a}--\subref*{Fig:Surf_c} is through vertical cross sections of the surface plots for a fixed value of $\mathrm{FOV}$. In Fig.~\ref{Fig:SurfCrossSec_FixedFOV}, $R$, $L_\mathrm{ADR}$ and $A_\mathrm{ADR}$ are plotted against $B$ for $\mathrm{FOV}=30^\circ$, based on \ac{ADR} Configs.~1, 2 and 3 with $2\times2$, $4\times4$ and $8\times8$ \ac{PD} arrays (i.e., $N_\mathrm{PD} = 4,16,64$). From Fig.~\subref*{Fig:R_B_FixedFOV}, it can be seen that for each array size, the achievable data rate has a peak in $2<B<4$~GHz; occurring at $B=2.1,2.7,3.5$~GHz for Configs.~1, 2 and 3, respectively. Increasing the bandwidth beyond these values does not help to improve the achievable data rate, but rather brings about a performance degradation. This phenomenon happens as a result of the underlying tradeoffs governing the \ac{ADR} design. For a given \ac{FOV}, the optical gain of \acp{CPC} and hence the ratio between the entrance and exit aperture areas is fixed in line with the gain-\ac{FOV} tradeoff. The bandwidth is increased by reducing the \ac{PD} area according to the area-bandwidth tradeoff, which in turn lessens the total collection area of the \ac{PD} array for a given array size. Consequently, there is a threshold at which the loss in the power collection efficiency starts to compromise the advantage of a higher bandwidth. Furthermore, it can be observed that the data rate requirement of $R\geq 10$~Gb/s is met by using a $2\times2$ array. The \ac{PD} array size can be increased as a means to improve the overall rate performance, as shown in Fig.~\subref*{Fig:R_B_FixedFOV}, at the expense of enlarging the overall height of the receiver as well as the total effective area, as shown in Figs.~\subref*{Fig:L_B_FixedFOV} and \subref*{Fig:A_B_FixedFOV}. The triple $(R,L_\mathrm{ADR},A_\mathrm{ADR})=(14.00~\text{Gb/s},1.99~\text{cm},2.12~\text{cm}^2)$ corresponds to Config.~1 at $B=2.1$~GHz, which is the best operating point for this configuration in terms of the data rate performance. In this case, it turns out that Config.~1 offers the appropriate \ac{ADR} design taking account of the hardware complexity, as it uses $4$ \acp{PD} per receiver element.
%---------------------------------------------------------------------------------------------------
\subsection{Design Spaces with Rate and FOV Requirements: Single-Tier ADR vs. Multi-Tier ADR}
Fig.~\ref{Fig:Cont_Multi} illustrates the contour plots of the achievable data rate $R$. Figs.~\subref*{Fig:Cont_Multi_a}--\subref*{Fig:Cont_Multi_c} correspond to three $1$-tier ADR designs composed of $7$ CPCs in conjunction with $2\times2$, $4\times4$, and $8\times8$ PD arrays, respectively (i.e., $N_\mathrm{PD}=4,16,64$). Comparing Configs.~1, 2 and 3 confirms that a larger PD array size leads to a higher overall rate performance. To exemplify the rate maximisation under a minimum FOV constraint only according to the optimisation problem in \eqref{Eq:Max_1}, suppose a $\mathrm{FOV}\geq30^{\circ}$ is required. The contour plots in Fig.~\ref{Fig:Cont_Multi} show that the achievable rate is maximised at the point where its curve is tangent to the line $\mathrm{FOV}=30^\circ$, as indicated by dashed black curves. The number displayed next to each one represents the maximum value of $R$ for $\mathrm{FOV} \geq 30^\circ$. It can be observed that the maximum data rate of the $1$-tier ADR design reaches values greater than $10$~Gb/s by means of $2\times2$ , $4\times4$ and $8\times8$ PD arrays, respectively.
%---------------------------------------------------------------------------------------------------
There are three different multi-tier ADR configurations labelled as Configs.~4, 5 and 6 in Table~\ref{Tab:2}. Configs.~4 and 5 are $2$-tier ADRs with $76$ and $304$ PDs in total, respectively, while Config.~6 is a $3$-tier ADR having $148$ PDs. The contour plots of $R$ for these ADRs are shown and compared in Fig.~\ref{Fig:Cont_Multi}. Config.~3, which is a $1$-tier ADR with $448$ PDs, is also included as a benchmark for comparison. Suppose the design requirements are $R \geq 10$~Gb/s and $\mathrm{FOV}\geq 30^\circ$, in line with the desired specifications for \ac{6G} optical wireless networks \cite{ESarbazi2020tb}. These inequalities jointly define a design space on the $B$-$\mathrm{FOV}$ plane as highlighted in a grey shade in Fig.~\ref{Fig:Cont_Multi}. The design space brings forth the flexibility to choose the PD array parameters while fulfilling the design objectives.

\begin{figure}[t!]
    \captionsetup[subfigure]{justification=centering}
    \centering
    \subfloat[\label{Fig:Cont_Multi_a} Config. 1 \\($N_\mathrm{tier}=1$ and $2\times2$ PD array)] {\includegraphics[width=0.33\textwidth, keepaspectratio=true]{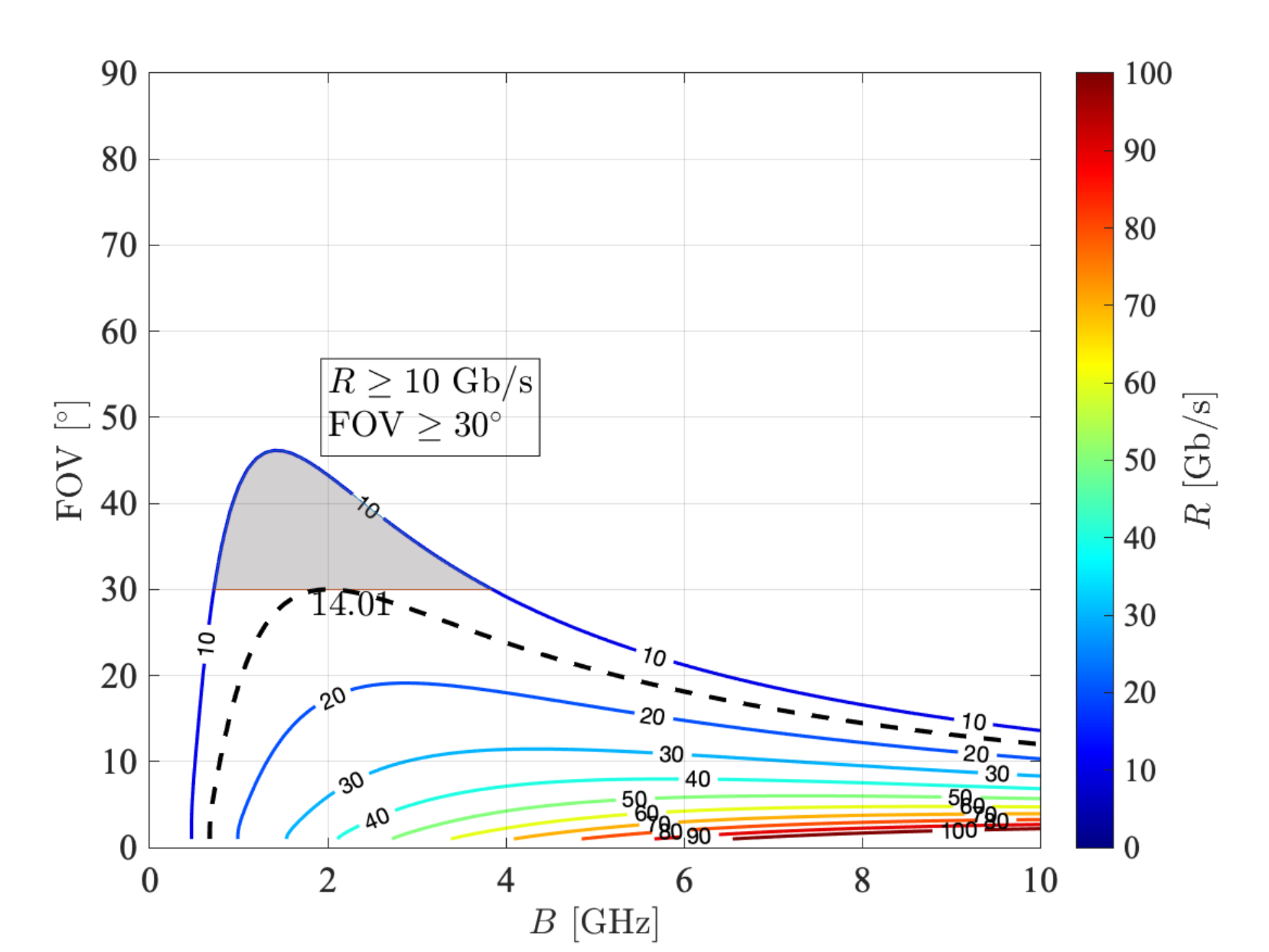}}
    \subfloat[\label{Fig:Cont_Multi_b} Config. 2 \\($N_\mathrm{tier}=1$ and $4\times4$ PD array)] {\includegraphics[width=0.33\textwidth, keepaspectratio=true]{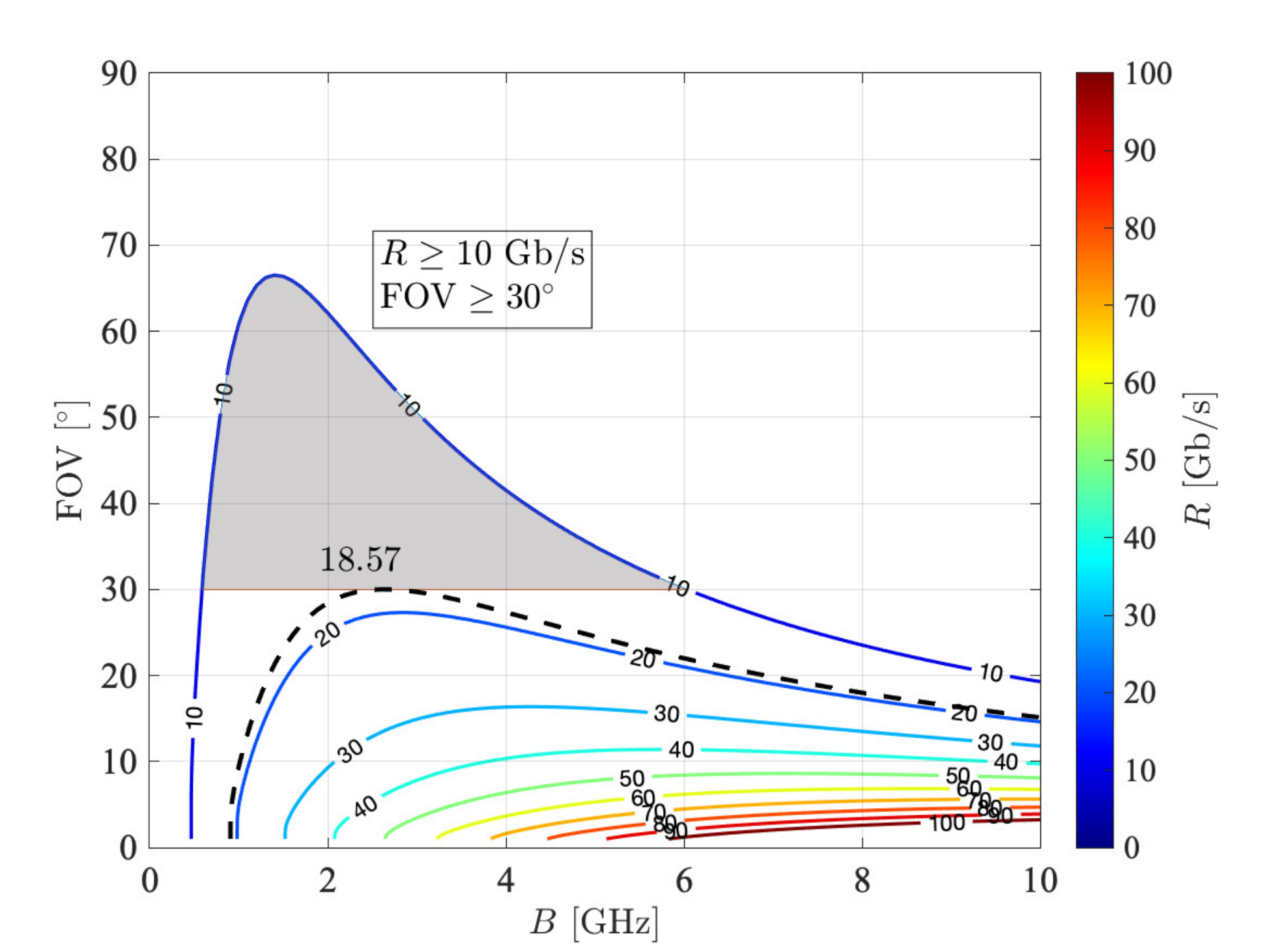}}
    \subfloat[\label{Fig:Cont_Multi_c} Config. 3 \\($N_\mathrm{tier}=1$ and $8\times8$ PD array)]{\includegraphics[width=0.33\textwidth, keepaspectratio=true]{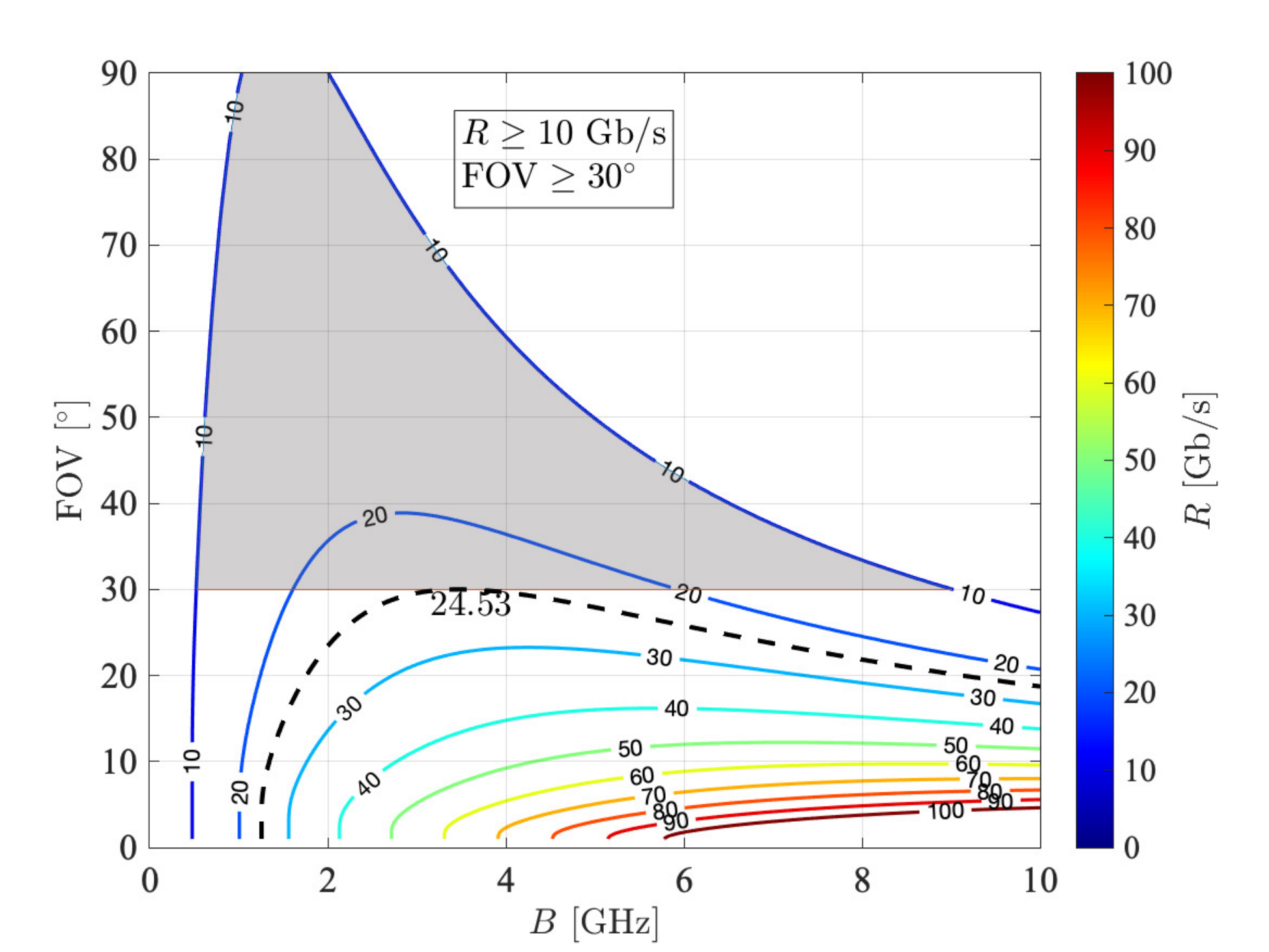}}\\
    \subfloat[\label{Fig:Cont_Multi_d} Config. 4 \\($N_\mathrm{tier}=2$ and $2\times2$ PD array)]{\includegraphics[width=0.33\textwidth, keepaspectratio=true]{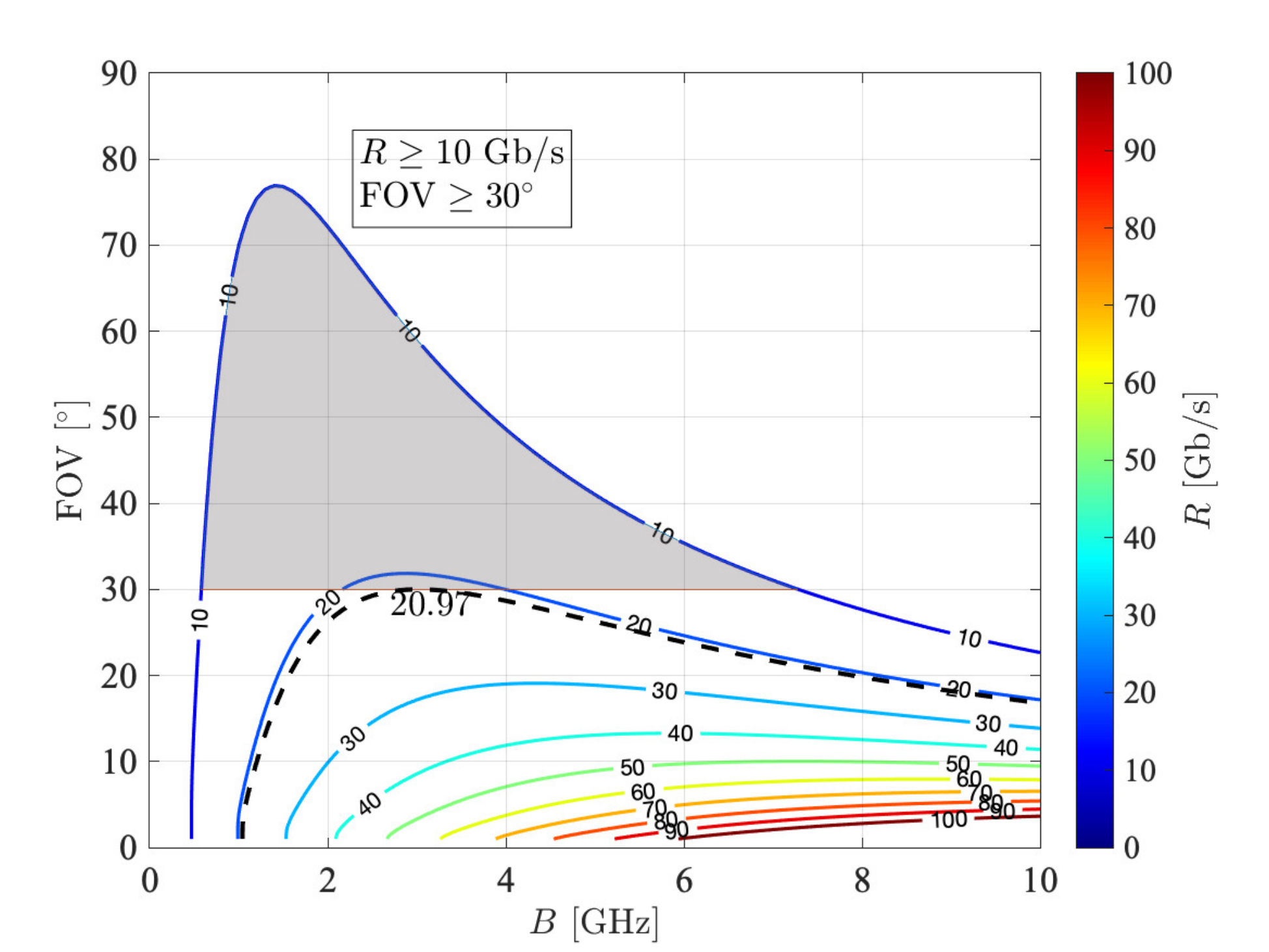}}
    \subfloat[\label{Fig:Cont_Multi_e} Config. 5 \\($N_\mathrm{tier}=2$ and $4\times4$ PD array)]{\includegraphics[width=0.33\textwidth, keepaspectratio=true]{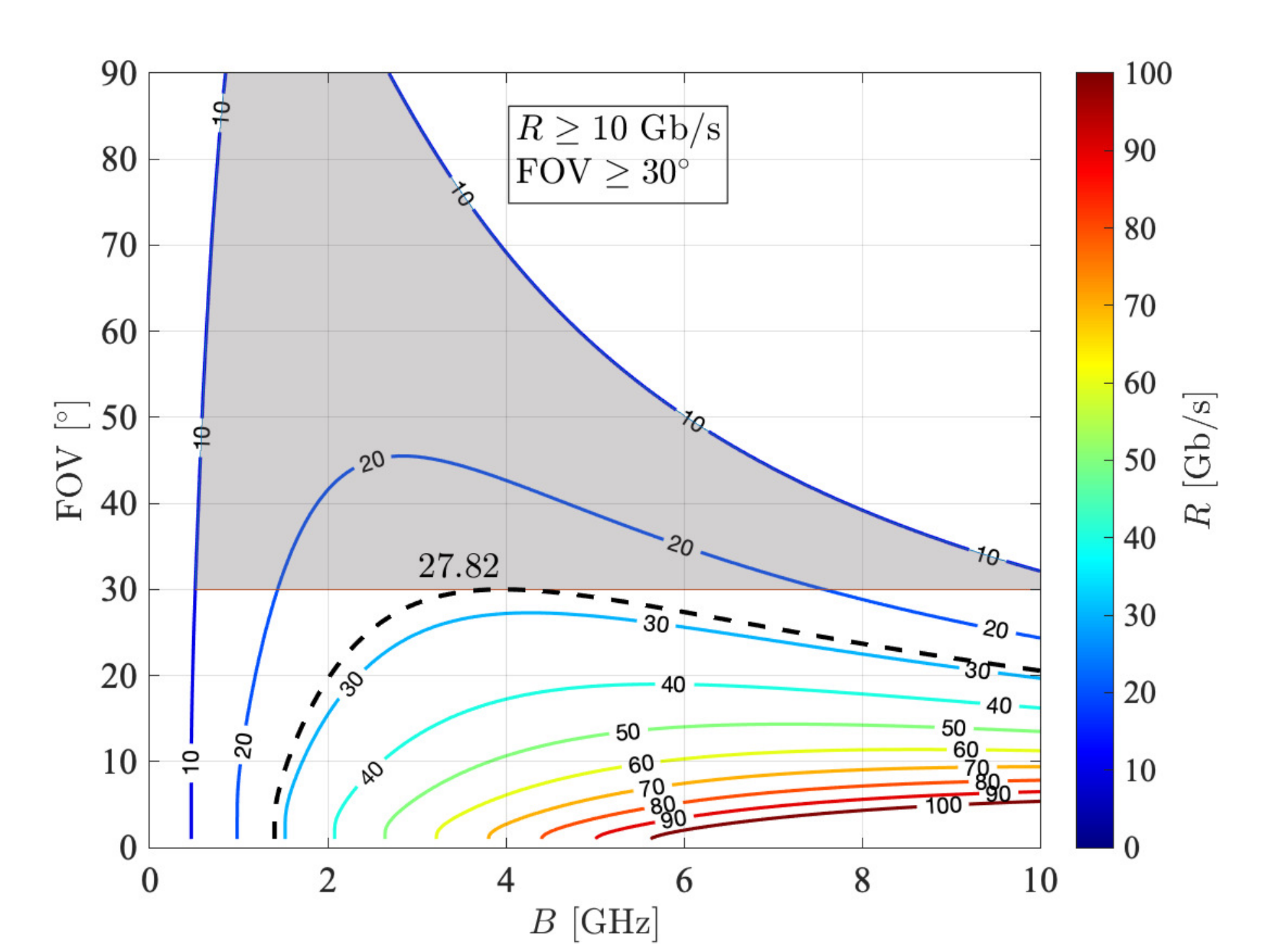}}
    \subfloat[\label{Fig:Cont_Multi_f} Config. 6 \\($N_\mathrm{tier}=3$ and $2\times2$ PD array)]{\includegraphics[width=0.33\textwidth, keepaspectratio=true]{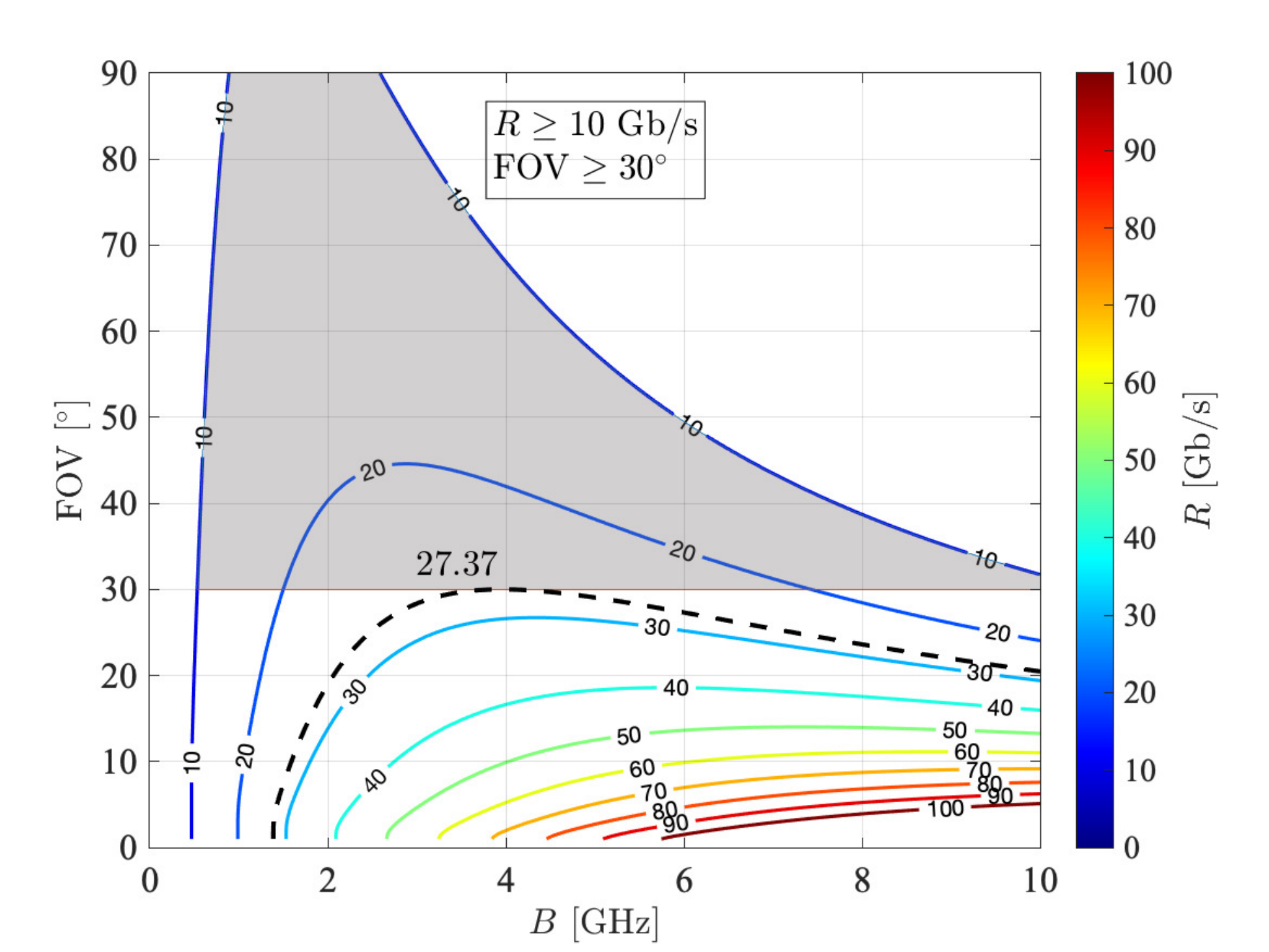}}
    \caption{Contour plots of the achievable data rate $R$ as a function of the bandwidth $B$ and $\mathrm{FOV}$ for Configs.~1--6.}
    \label{Fig:Cont_Multi}
    \vspace{-20pt}
\end{figure}

Comparing $2$-tier ADR Configs.~4 and 5, as shown in Figs.~\subref*{Fig:Cont_Multi_d} and \subref*{Fig:Cont_Multi_e}, it can be observed that for the same number of tiers, the design space expands when the PD array size and hence the total number of PDs used for each receiver element increases. Besides, comparing $1$-tier ADR Config.~3 in Fig.~\subref*{Fig:Cont_Multi_c} with $2$-tier ADR Config.~5 in Fig.~\subref*{Fig:Cont_Multi_e} points out that the incorporation of an additional tier results in a larger design space, although a smaller PD array size is used in Config.~5. This is also evident when moving on to $3$-tier ADR Config.~6, as shown in Fig.~\subref*{Fig:Cont_Multi_f}. Therefore, in a multi-tier \ac{ADR}, increasing the number of tiers allows the use of a smaller PD array size to meet the same data rate and \ac{FOV} requirements.

%---------------------------------------------------------------------------------------------------
\subsection{Rate Maximisation Under Joint FOV and Overall Dimensions Constraints}
Fig.~\ref{Fig:FeasibleRegion2} displays various realisations of the feasible region for the optimisation problem in \eqref{Eq:Max_2}. The feasible region on the $B$-$\mathrm{FOV}$ plane is the intersection of the overlapping areas formed by the three constraints in \eqref{Eq:Max_2b}, \eqref{Eq:Max_2c} and \eqref{Eq:Max_2d}. The boundary of the area resulting from the constraint $A_{\mathrm{ADR}} \leq A_{\max}$ is determined by the nonlinear function $\mathrm{FOV}= f_A ^{-1}(B)$ in \eqref{Eq:f_L}. Similarly, the nonlinear function $\mathrm{FOV}= f_L ^{-1}(B)$ in \eqref{Eq:f_A} defines the boundary of the area due to the constraint $L_{\mathrm{ADR}} \leq L_{\max}$.
The boundary curves can have up to three intersection points within the range of interest and one or two of the boundary conditions can be dominant depending on the values of $\mathrm{FOV}_\mathrm{min}$, $A_{\max}$ and $L_{\max}$.

\begin{figure}[t!]
    \captionsetup[subfigure]{justification=centering}
    \centering
    \subfloat[\label{Fig:FeasibleRegion2_a} $\mathrm{FOV}_\mathrm{min}=20^{\circ}$, \\$L_{\max}=1$~cm, $A_{\max}=10$~cm$^2$] {\includegraphics[width=0.33\textwidth, keepaspectratio=true] {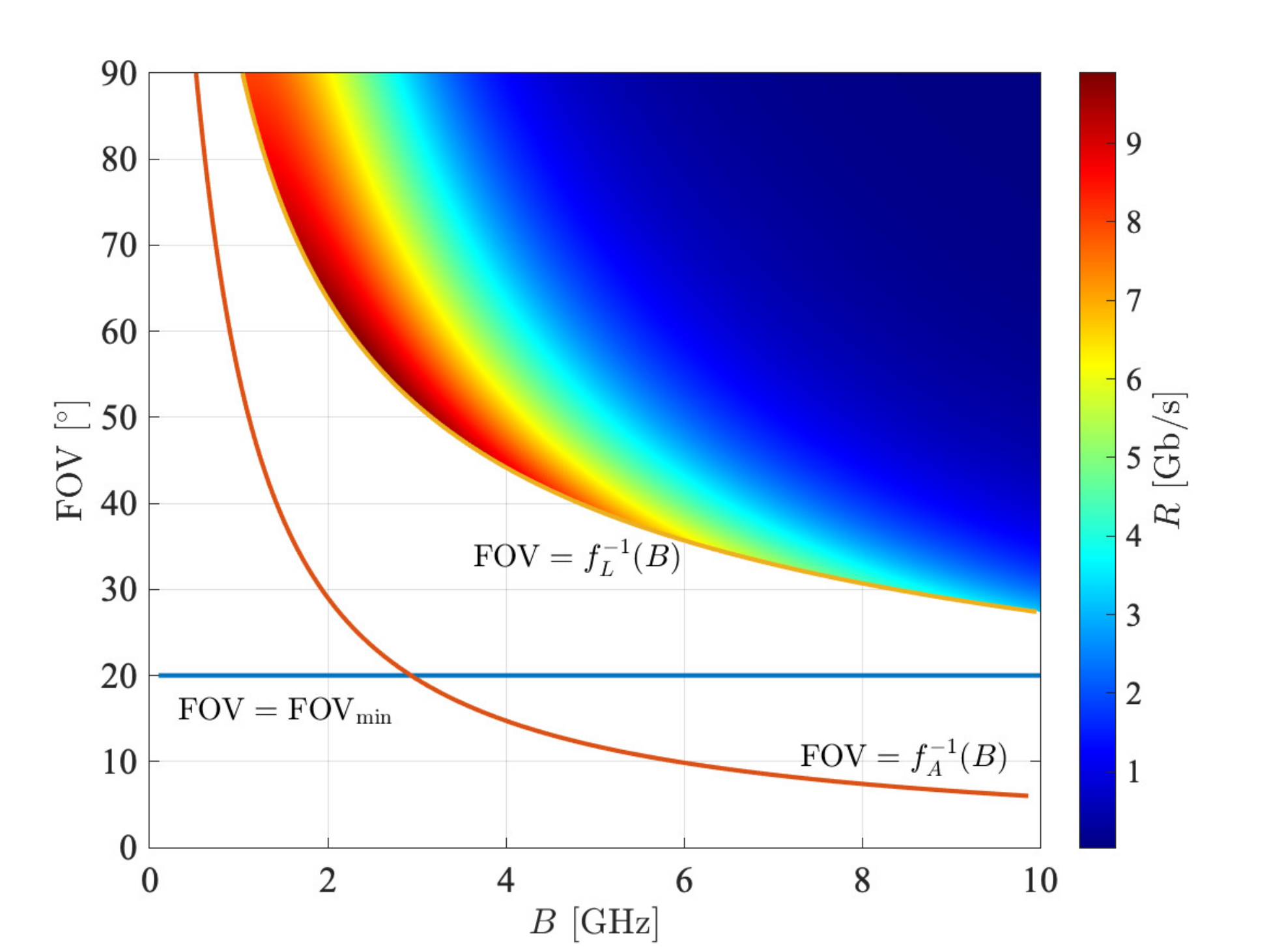}}\hfill \subfloat[\label{Fig:FeasibleRegion2_b} $\mathrm{FOV}_\mathrm{min}=15^{\circ}$, \\$L_{\max}=5$~cm, $A_{\max}=1$~cm$^2$] {\includegraphics[width=0.33\textwidth, keepaspectratio=true] {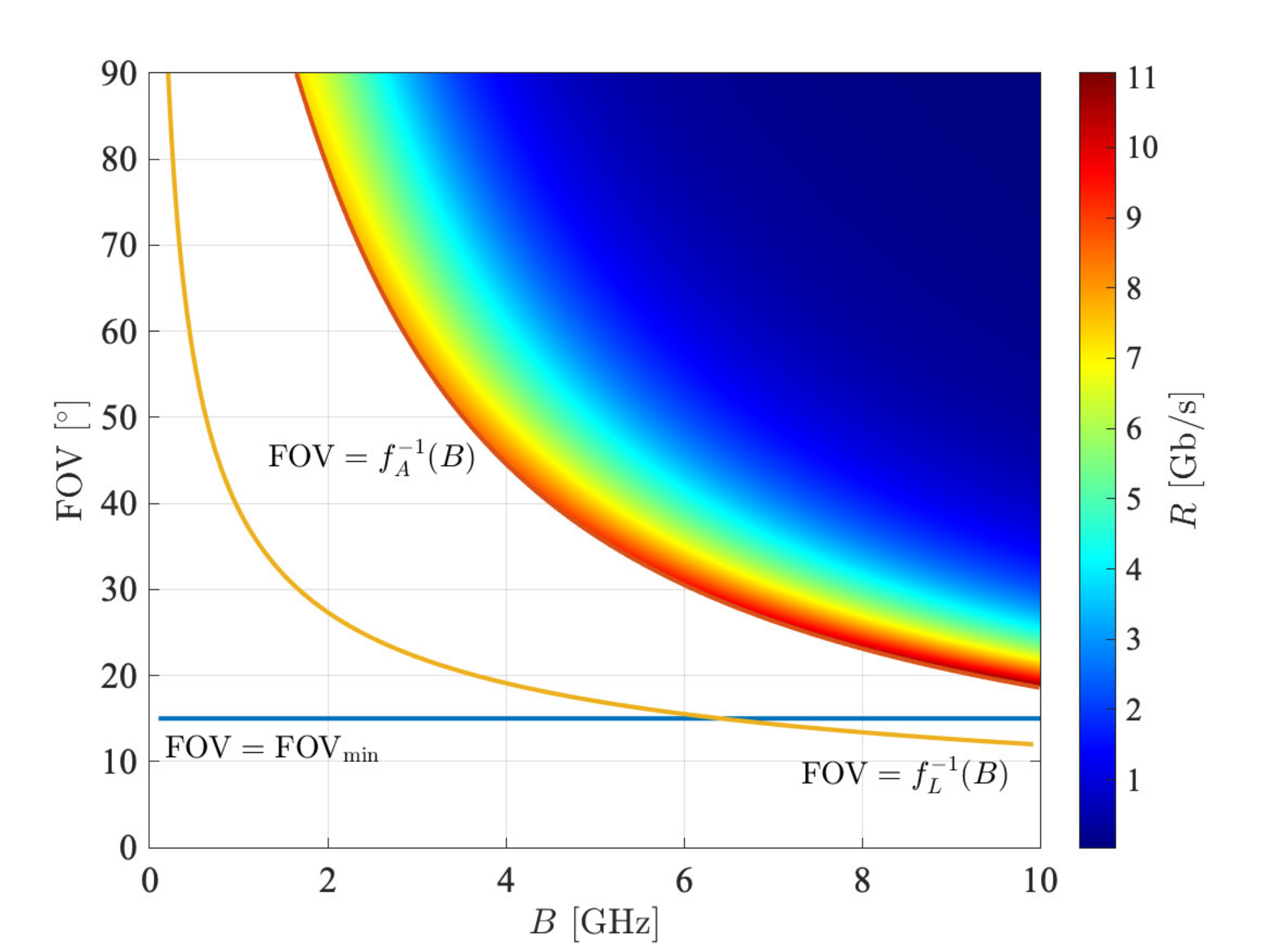}}\hfill
    \subfloat[\label{Fig:FeasibleRegion2_c} $\mathrm{FOV}_\mathrm{min}=30^{\circ}$, \\$L_{\max}=1$~cm, $A_{\max}=10$~cm$^2$] {\includegraphics[width=0.33\textwidth, keepaspectratio=true] {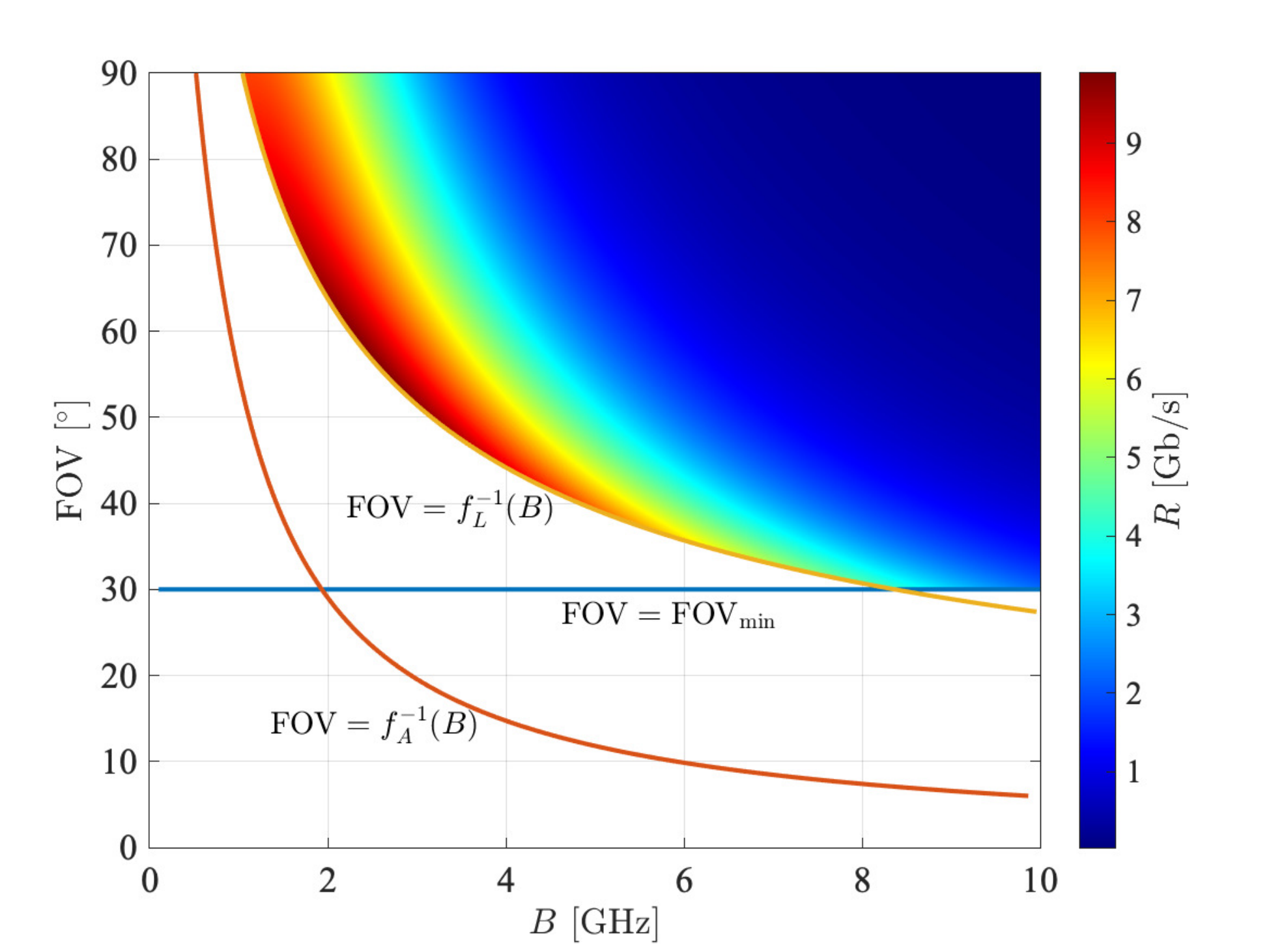}}\\
    \subfloat[\label{Fig:FeasibleRegion2_d} $\mathrm{FOV}_\mathrm{min}=30^{\circ}$, \\$L_{\max}=5$~cm, $A_{\max}=2$~cm$^2$] {\includegraphics[width=0.33\textwidth, keepaspectratio=true] {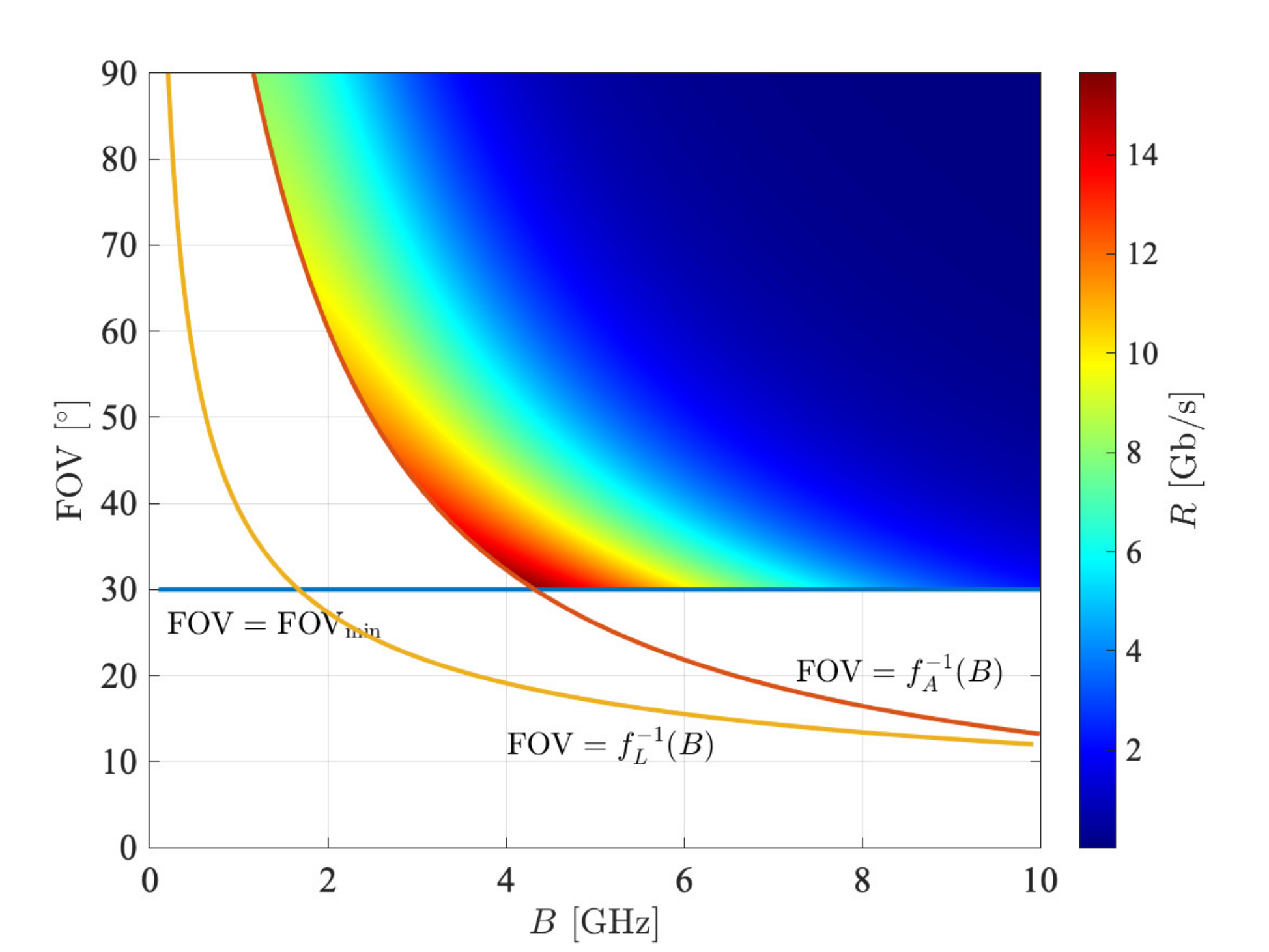}}\hfill
    \subfloat[\label{Fig:FeasibleRegion2_e} $\mathrm{FOV}_\mathrm{min}=30^{\circ}$, \\$L_{\max}=2$~cm, $A_{\max}=5$~cm$^2$] {\includegraphics[width=0.33\textwidth, keepaspectratio=true] {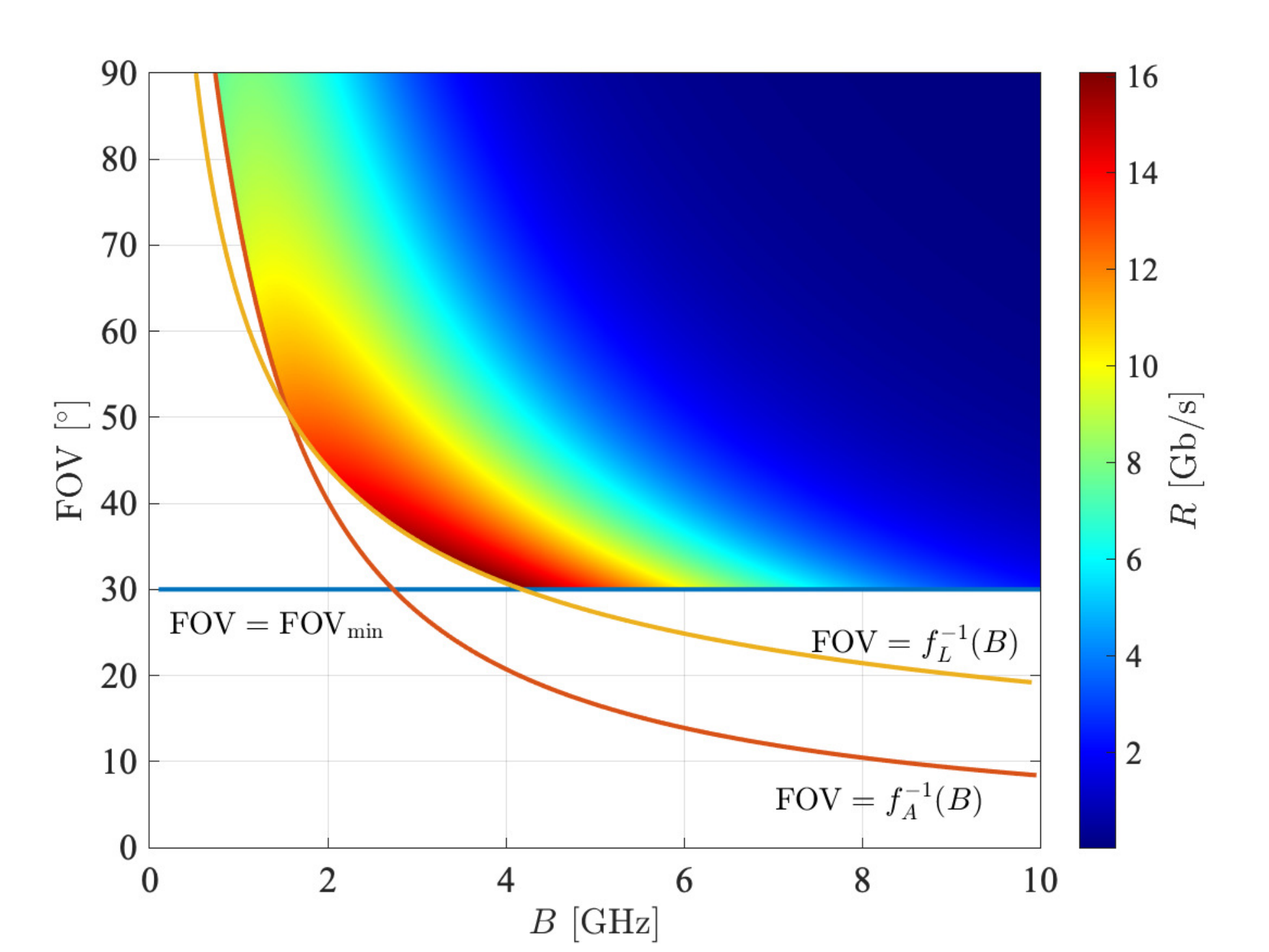}}\hfill
    \subfloat[\label{Fig:FeasibleRegion2_f} $\mathrm{FOV}_\mathrm{min}=30^{\circ}$, \\$L_{\max}=5$~cm, $A_{\max}=10$~cm$^2$] {\includegraphics[width=0.33\textwidth, keepaspectratio=true] {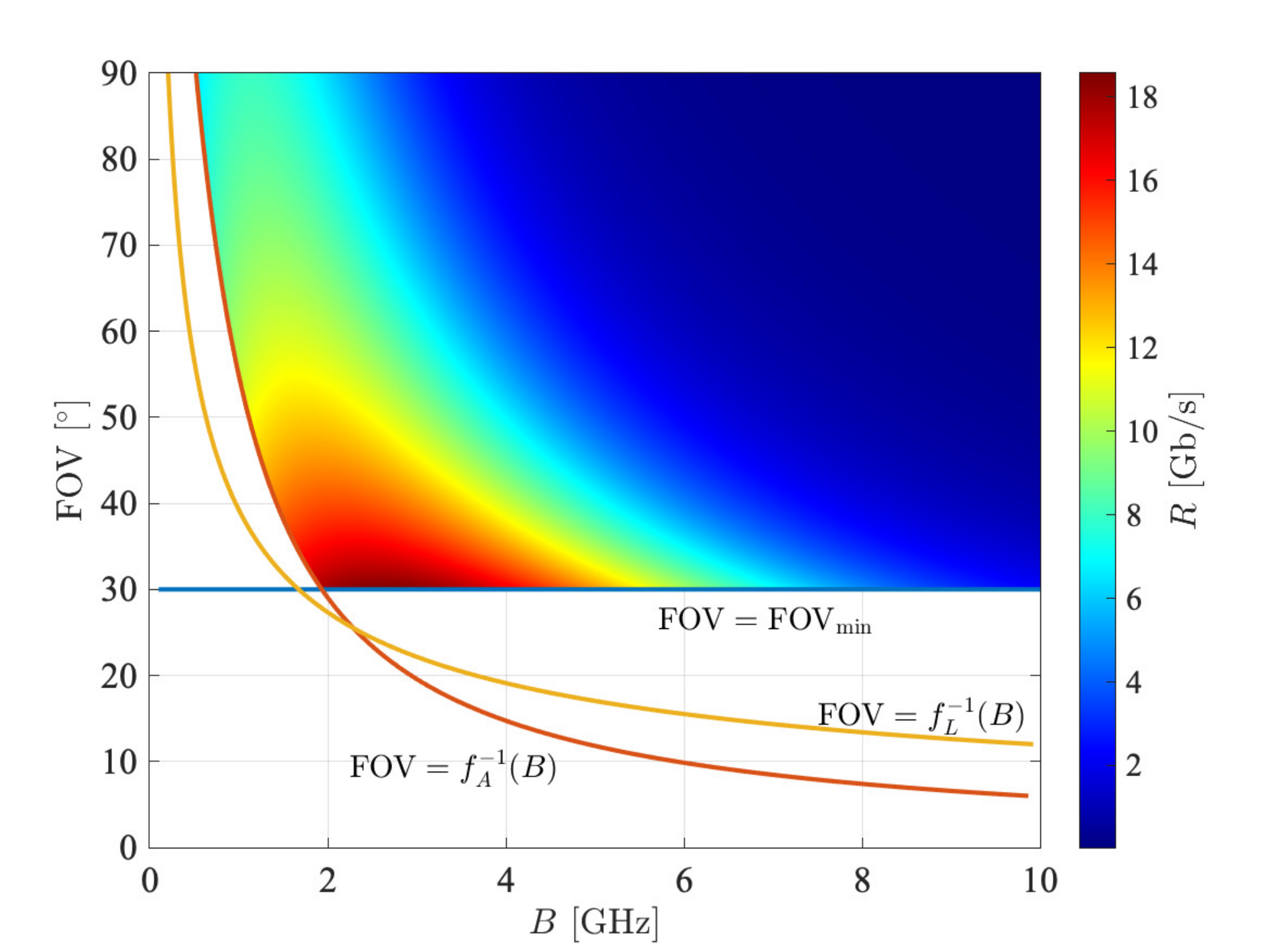}}
    \caption{Various realisations of the feasible region for the optimisation problem in \eqref{Eq:Max_2} under Config.~2.}
    \label{Fig:FeasibleRegion2}
    \vspace{-20pt}
\end{figure}

The examples provided in Fig.~\ref{Fig:FeasibleRegion2} represent six possible outcomes of the feasible region arising from different configurations of the constraints. In Figs.~\subref*{Fig:FeasibleRegion2_a} and~\subref*{Fig:FeasibleRegion2_b}, two of the three constraints are inactive and one of them solely controls the feasible region. In Figs.~\subref*{Fig:FeasibleRegion2_c} and~\subref*{Fig:FeasibleRegion2_d}, on the other hand, one of the constraints is inactive and the intersection of the other two constitutes the feasible region. Figs.~\subref*{Fig:FeasibleRegion2_e} and~\subref*{Fig:FeasibleRegion2_f} illustrate two more examples for the case where the constraints intersect pairwise at three distinct points. In this case, all the three constrains actively participate in determining the feasible region. Note that there are other possible realisations for the feasible region that are not shown in Fig.~\ref{Fig:FeasibleRegion2}. For instance, in each case, depending on how large $\mathrm{FOV}_\mathrm{min}$ is chosen, it may or may not cross the boundary curves of the other two constraints. Furthermore, it can be visually verified that the solution to the rate maximisation problem always lies on the boundary of the feasible region, as already anticipated by Lemma~\ref{Lemma:2}.

\begin{figure}[t!]
    \centering
    \subfloat[\label{Fig:Surf_Rmax_FOV15} $\mathrm{FOV}_\mathrm{min}=15^{\circ}$] {\includegraphics[width=0.48\textwidth, keepaspectratio=true] {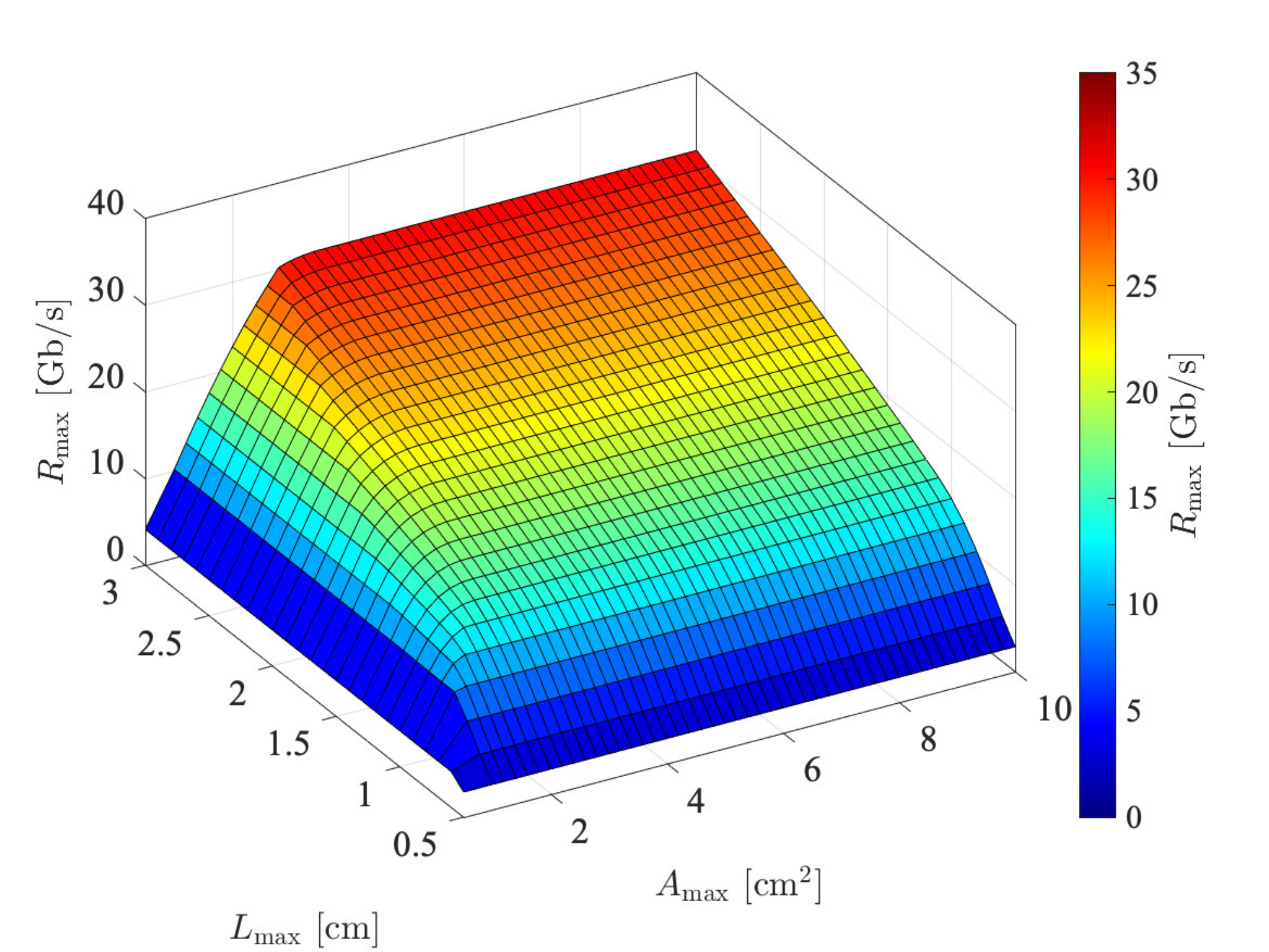}}\hfill
    \subfloat[\label{Fig:Surf_Rmax_FOV30} $\mathrm{FOV}_\mathrm{min}=30^{\circ}$] {\includegraphics[width=0.48\textwidth, keepaspectratio=true] {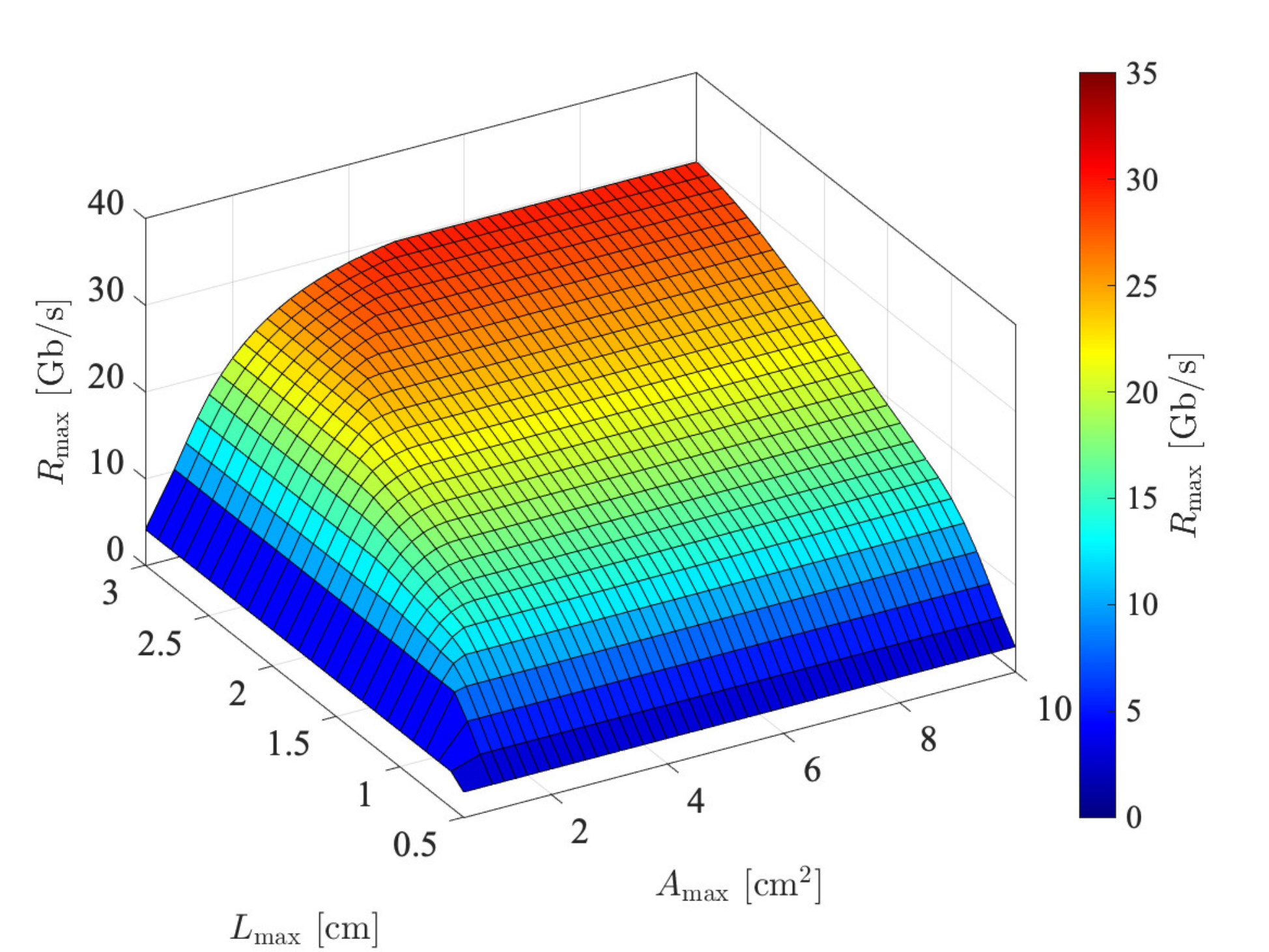}}
    \caption{The maximum achievable data rate $R_\mathrm{max}$ as a function of the dimensions $A_\mathrm{max}$ and $L_\mathrm{max}$ for Config.~3 and $w_0=10$~{\textmu}m.}
    \label{Fig:Surf_Rmax}
    \vspace{-20pt}
\end{figure}

Fig.~\ref{Fig:Surf_Rmax} shows the maximum achievable rate $R_\mathrm{max}$ for Config.~3 with $8\times8$ \ac{PD} arrays when $L_\mathrm{max}$ and $A_\mathrm{max}$ are variable, for two values of $\mathrm{FOV}_\mathrm{min}=15^\circ,30^\circ$, assuming $w_0=10$~{\textmu}m. This figure can be used to identify conditions under which one of the two constraints $L_\mathrm{ADR}\leq L_\mathrm{max}$ or $A_\mathrm{ADR}\leq A_\mathrm{max}$ becomes inactive; for instance, for a given $L_\mathrm{max} = \hat{L}$, there exists $\hat{A}$ such that $A_\mathrm{ADR}\leq A_\mathrm{max}$ is inactive for $A_\mathrm{max}\geq \hat{A}$. Based on Fig.~\subref*{Fig:Surf_Rmax_FOV15}, for $L_\mathrm{max}=2$~cm, the constraint on $A_\mathrm{ADR}$ turns out to be inactive for $A_\mathrm{max}\geq 2.5$~cm$^2$, as $R_\mathrm{max}$ is constant when $A_\mathrm{max}\geq 2.5$~cm$^2$. In addition, it can be observed that $R_\mathrm{max}$ is an increasing function of both $A_\mathrm{max}$ and $L_\mathrm{max}$. For $\mathrm{FOV}_\mathrm{min}=15^\circ$, peak data rates greater than $25$~Gb/s and up to $30$~Gb/s are achieved at the cost of large \ac{ADR} dimensions with $L_\mathrm{max}\geq2.5$~cm and $A_\mathrm{max}\geq3$~cm$^2$. When $\mathrm{FOV}_\mathrm{min}$ goes up to $30^\circ$, $R_\mathrm{max}$ does not exceed $30$~Gb/s for $L_\mathrm{max}\geq2.5$~cm and $A_\mathrm{max}\geq4$~cm$^2$.

%%%%%%%%%%%%%%%%%%%%%%%%%%%%%%%%%%%%%%%%%%%%%%%%%%%%%%%%%%%%%%%%%%%%%%%%%%%%%%%%%%%%%%%%%%%%%%%%%%%%
%%%%%%%%%%%%%%%%%%%%%%%%%%%%%%%%%%%%%%%%%%%%%%%%%%%%%%%%%%%%%%%%%%%%%%%%%%%%%%%%%%%%%%%%%%%%%%%%%%%%
\section{Compact ADR Design Using CPC Length Truncation} \label{Sec:6}
Although \acp{CPC} are superior to other types of concentrators in terms of optical gain, the main drawback of \acp{CPC} lies in their relatively long length compared to the diameter of the collecting aperture. A practical and cost effective solution for reducing the length of \acp{CPC} is truncation \cite{RWinston2005,P_DTsonev2019}. Length truncation slightly reduces the size of the entrance aperture. This effect has been investigated in detail for \ac{2D} and \ac{3D} \acp{CPC} in \cite{Winston1975Principles,ARabl1976Optical,Welford1978Optics}. Since the parabolic surface of a \ac{CPC} near the entrance aperture is almost parallel to the optical axis, it can be truncated well short of the full \ac{CPC} length without a significant reduction in the entrance aperture size. As a result, truncated \acp{CPC} exhibit a minor decrease in the concentration gain, and yet they reach a better optical efficiency because the light undergoes less number of internal reflections as compared to full-length \acp{CPC}.
%---------------------------------------------------------------------------------------------------
\subsection{Modified ADR Design Using Truncated CPCs}
Let $D_1^\mathrm{T}$, $L_{\mathrm{CPC}}^{\mathrm{T}}$, $G_{\mathrm{CPC}}^{\mathrm{T}}$ and $\theta_{\mathrm{CPC}}^{\mathrm{T}}$ denote the diameter of the entrance aperture, the length, the concentration gain and the acceptance angle of a truncated CPC, while $D_1$, $L_{\mathrm{CPC}}$, $G_{\mathrm{CPC}}$ and $\theta_{\mathrm{CPC}}$ are the respective parameters of a full-length CPC. The truncation ratio is denoted by $\tau$ so that $L^{\mathrm{T}}_{\mathrm{CPC}} = \tau L_{\mathrm{CPC}}$ for $0<\tau<1$. Besides the length reduction, truncation rather increases the acceptance angle of a CPC, albeit its impact on the overall angular collection performance is insignificant. The findings in \cite{ARabl1976Optical} have corroborated that $\theta^{\mathrm{T}}_{\mathrm{CPC}}\approx\theta_{\mathrm{CPC}}$ is the case for $\tau \geq 0.5$. Fig.~\ref{Fig:Truncation} shows the normalised concentration gain of a truncated CPC as a function of $\tau$ for different values of $\theta_{\mathrm{CPC}}$, assuming $n_{\mathrm{CPC}}=1.7$ and $D_2=1.5$~mm. It can be observed that the four characteristic curves are nearly overlapping for $\tau\geq0.4$, for which $G_\mathrm{CPC}^\mathrm{T}\geq0.9G_\mathrm{CPC}$. Also, $G_\mathrm{CPC}^\mathrm{T}\approx0.9G_\mathrm{CPC}$ for $\tau=0.6$ for all four values of $\theta_{\mathrm{CPC}}$, which means a $40\%$ length reduction leads to only $10\%$ gain loss. In the following, we use this as a simplifying assumption.

\begin{figure}[t!]
    \centering
    \includegraphics[width=0.7\textwidth, keepaspectratio=true] {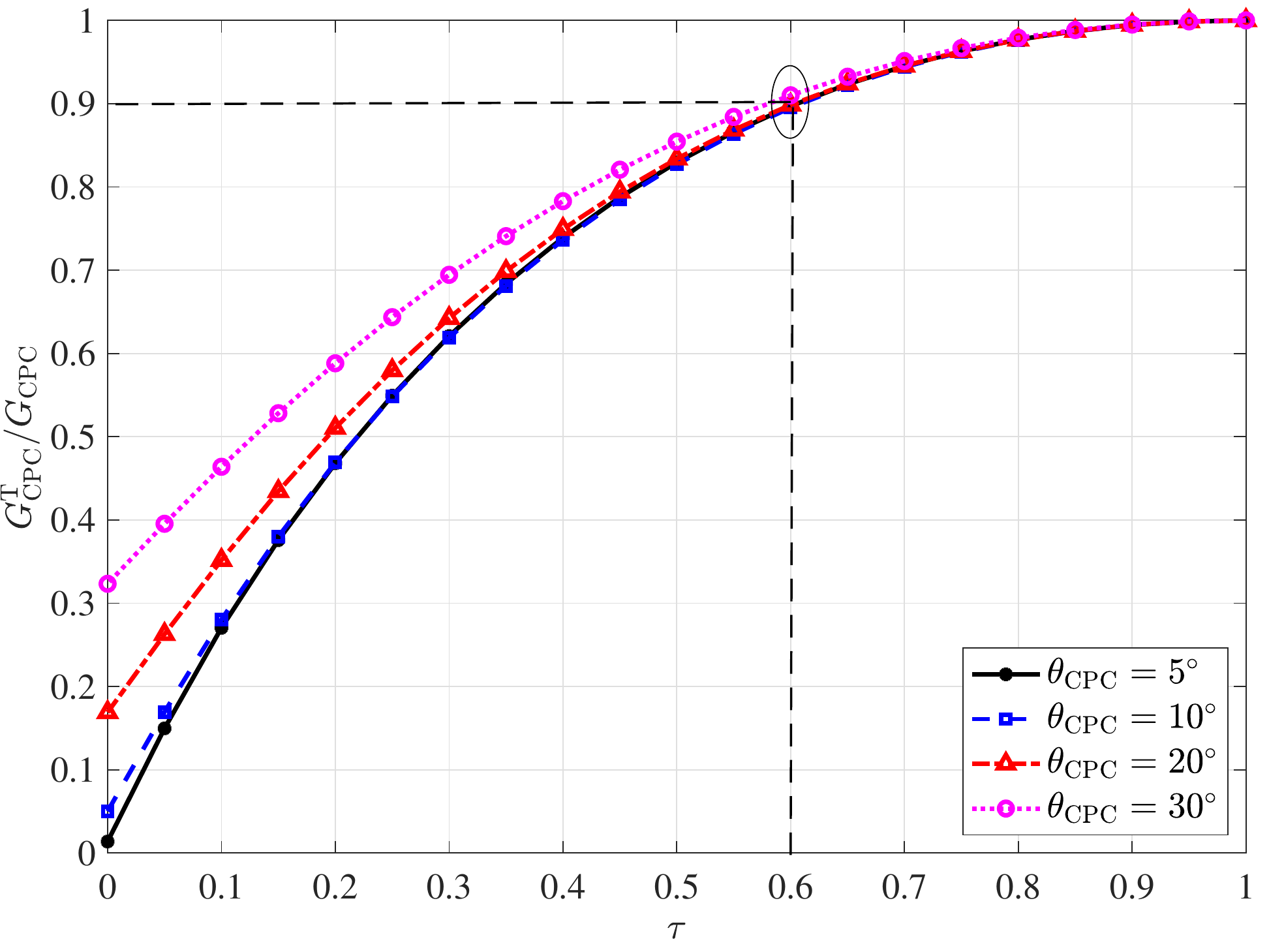}
    \caption{The CPC concentration gain as a function of the truncation ratio $\tau$ for $n_{\mathrm{CPC}}=1.7$ and $D_2=1.5$~mm.}
    \label{Fig:Truncation}
    \vspace{-20pt}
\end{figure}

In this case, the problem statements in \eqref{Eq:Max_1} and \eqref{Eq:Max_2} can be modified by introducing appropriate coefficients into the ADR parameters. In particular, for the modified ADR, it suffices to replace $D_1$ and $L_\mathrm{CPC}$ by $D_1^\mathrm{T}=\sqrt{0.9}D_1$ and $L_\mathrm{CPC}^\mathrm{T}=0.6L_\mathrm{CPC}$ through the use of \eqref{Eq:CPC_D1} and \eqref{Eq:L_CPC2}. Denoting the total effective area and the overall height of this modified ADR by $A_\mathrm{ADR}^\mathrm{T}$ and $L_\mathrm{ADR}^\mathrm{T}$, respectively, $L_\mathrm{ADR}^\mathrm{T}\approx 0.6L_\mathrm{CPC}$ and $A_\mathrm{ADR}^\mathrm{T}=0.9A_\mathrm{ADR}$ based on \eqref{Eq:L_ADR} and \eqref{Eq:A_ADR}. Subsequently, the objective function in \eqref{Eq:Max_1a} and \eqref{Eq:Max_2a} is updated by recalculating the received power $P_\mathrm{r}$ in \eqref{Eq:Pr2} for $D_1^\mathrm{T}$. In addition, the constraints in \eqref{Eq:Max_2c} and \eqref{Eq:Max_2d} change to $L_\mathrm{ADR}^\mathrm{T}\leq L_\mathrm{max}$ and $A_\mathrm{ADR}^\mathrm{T}\leq A_\mathrm{max}$. Note that the \ac{FOV} constraint in \eqref{Eq:Max_1b} and \eqref{Eq:Max_2b} does not change, since $\theta^{\mathrm{T}}_{\mathrm{CPC}}\approx\theta_{\mathrm{CPC}}$. With these modifications, although the main parameters affected including $D_1^\mathrm{T}$, $L_{\mathrm{ADR}}^\mathrm{T}$ and $A_{\mathrm{ADR}}^\mathrm{T}$ are involved in the achievable rate analysis, the essence of the rate maximisation problem is preserved. This means the same solutions already developed for the original problems in \eqref{Eq:Max_1} and \eqref{Eq:Max_2} by means of Theorem~\ref{Theorem:1} and Theorem~~\ref{Theorem:2} would apply to the modified ADR optimisation problems. Now, let us proceed to the maximum achievable rate performance of the modified ADR design based on truncated CPCs with $\tau=0.6$ as discussed.

%---------------------------------------------------------------------------------------------------
\begin{figure}[t!]
    \centering
    \subfloat[\label{Fig:Rmax_Amax}] {\includegraphics[width=0.48\textwidth, keepaspectratio=true] {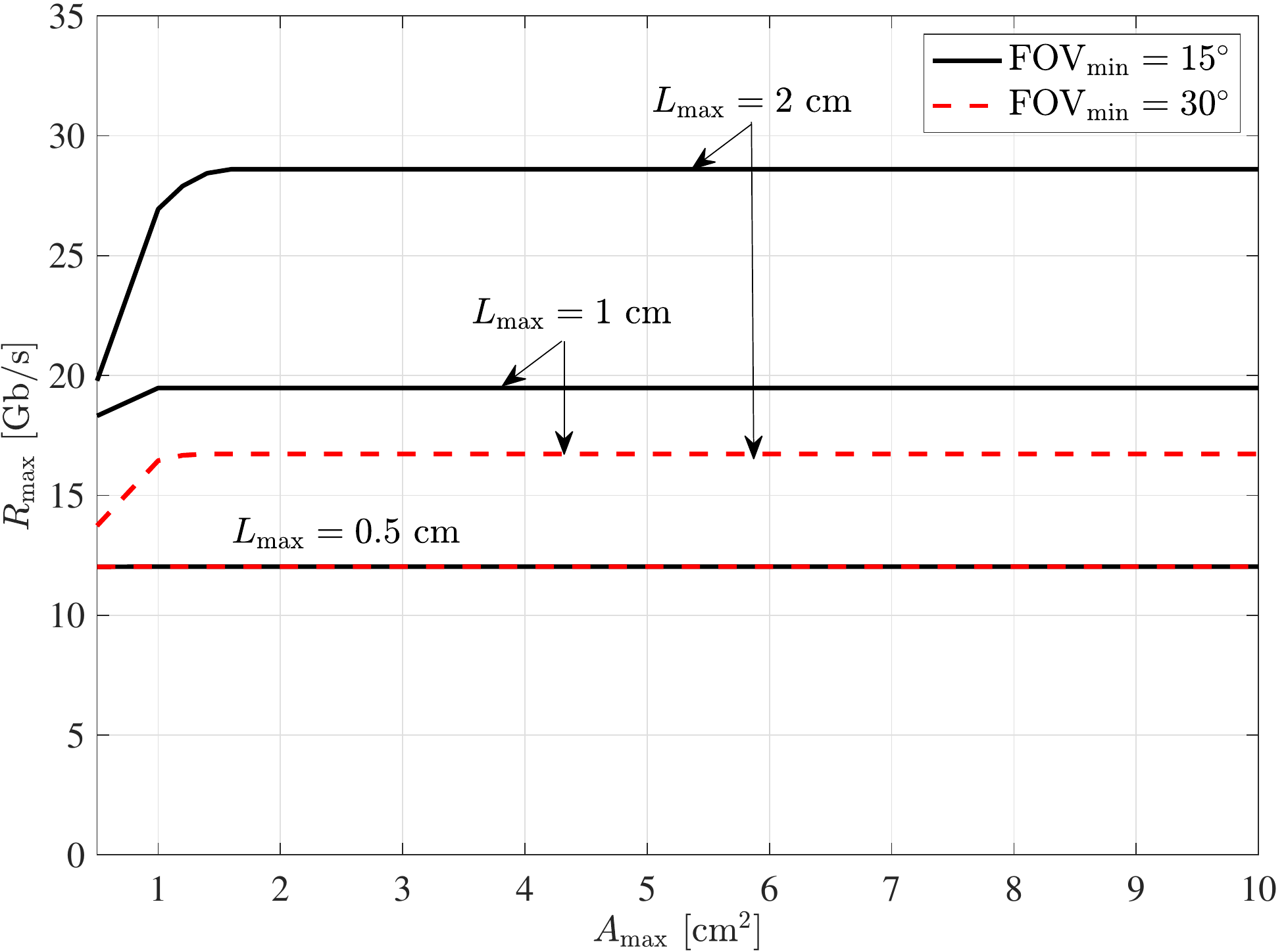}}\hfill
    \subfloat[\label{Fig:Rmax_Lmax}] {\includegraphics[width=0.48\textwidth, keepaspectratio=true] {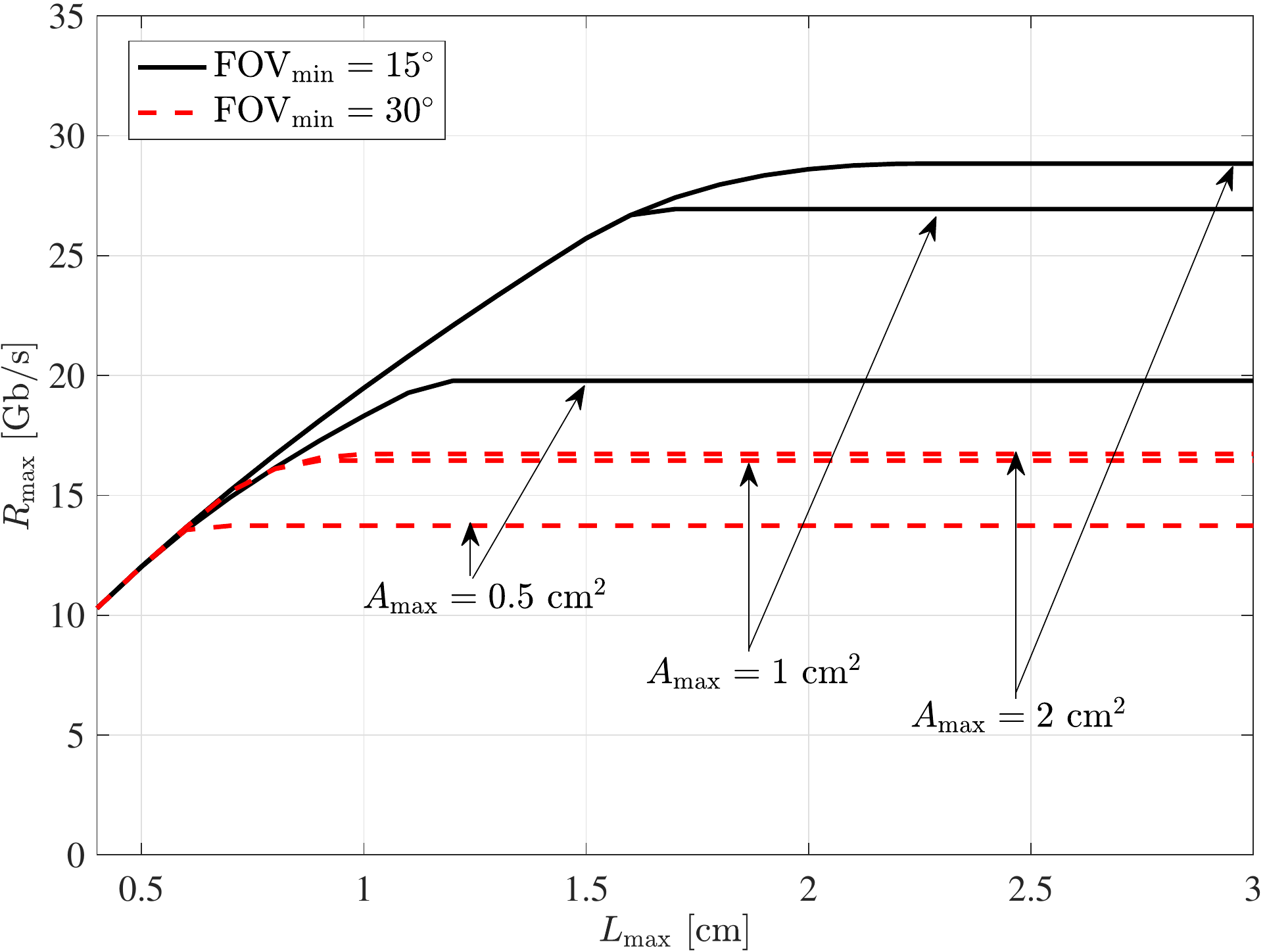}}
    \caption{The maximum achievable data rate $R_\mathrm{max}$ as a function of (a) $A_\mathrm{max}$ and (b) $L_\mathrm{max}$ for Config.~1 with $\tau=0.6$}
    \label{Fig:Rmax}
    \vspace{-20pt}
\end{figure}

\subsection{Performance Evaluation}
We consider the rate maximisation problem under joint \ac{FOV} and overall dimensions constraints as outlined in \eqref{Eq:Max_2}. For a given set of the constrains, the resulting feasible region may fall within one of the various categories shown in Fig.~\ref{Fig:FeasibleRegion2}. To shed light on the impact of each constraint on the maximum achievable data rate, Fig.~\ref{Fig:Rmax} demonstrates $R_\mathrm{max}$ for different values of the constraints, based on Config.~1 in Table~\ref{Tab:2}. In Fig.~\subref*{Fig:Rmax_Amax}, $R_\mathrm{max}$ is plotted against $A_\mathrm{max}$ for $L_\mathrm{max}=0.5,1,2$~cm and $\mathrm{FOV}_\mathrm{min} = 15^\circ,30^\circ$. It can be seen that in the case of $L_\mathrm{max} = 0.5$~cm, the change in $\mathrm{FOV}_\mathrm{min}$ does not make any difference in $R_\mathrm{max}$, since one of the dimensions constraints already dominates the maximum rate performance. Specifically, $L\leq L_\mathrm{max}$ is the most restrictive constraint which limits the performance at $R_\mathrm{max}=12$~Gb/s no matter how large $A_\mathrm{max}$ is chosen. For both cases of $L_\mathrm{max} = 1,2$~cm, on the other hand, the \ac{FOV} constraint does have an influence on the performance in that $R_\mathrm{max}$ improves by decreasing $\mathrm{FOV}_\mathrm{min}$. For $L_\mathrm{max}=1$~cm, the performance is limited at $R_\mathrm{max}=17,19.5$~Gb/s for $\mathrm{FOV}_\mathrm{min}=15^\circ,30^\circ$, respectively, when $A_\mathrm{max}\geq1$~cm$^2$. In the case of $L_\mathrm{max}=2$~cm, $R_\mathrm{max}$ reaches constant values of $17$~Gb/s for $A_\mathrm{max} \geq 1$~cm$^2$ and $28.5$~Gb/s for $A_\mathrm{max}\geq1.5$~cm$^2$, corresponding to $\mathrm{FOV}_\mathrm{min}=15^\circ,30^\circ$.

In Fig.~\subref*{Fig:Rmax_Lmax}, $R_\mathrm{max}$ is evaluated as a function of $L_\mathrm{max}$ for $A_\mathrm{max}=0.5,1,2$~cm$^2$ and $\mathrm{FOV}_\mathrm{min}=15^\circ,30^\circ$. For $A_\mathrm{max}=1,2$~cm$^2$, it is evident that the choice of $\mathrm{FOV}_\mathrm{min}$ has no effect on the maximum rate performance when $L_\mathrm{max}\leq0.7$~cm. By decreasing $A_\mathrm{max}$ to $0.5$~cm$^2$, this is the case as long as $L_\mathrm{max}\leq0.6$~cm. After $L_\mathrm{max}$ exceeds these thresholds, $R_\mathrm{max}$ varies with all three constraints. For sufficiently large values of $L_\mathrm{max}$, however, the rate maximisation is mainly controlled by the \ac{FOV} constraint as well as the total effective area constraint. For $\mathrm{FOV}_\mathrm{min} = 30^\circ$, the performance reaches no greater than $R_\mathrm{max}=13.8,16.8,17$~Gb/s for $A_\mathrm{max}=0.5,1,2$~cm$^2$, respectively. With $\mathrm{FOV}_\mathrm{min}=15^\circ$, these limits are improved to $R_\mathrm{max}=20,27,29$~Gb/s at the cost of halving the minimum \ac{FOV}. In this case, $R_\mathrm{max}$ retains its growing trend with $L_\mathrm{max}$ over a wider range of $L_\mathrm{max}$ while approaching its upper limit. In fact, the constraint on $\mathrm{FOV}$ is relaxed enough allowing the two constraints on $L_\mathrm{ADR}$ and $A_\mathrm{ADR}$ to take control of the rate maximisation across the boundary of the feasible region.

\begin{figure}[t!]
    \captionsetup[subfigure]{justification=centering}
    \centering
    \subfloat[Modified ADR: \\No constraint on dimensions \label{Fig:Rmax_FOVmin_a}] {\includegraphics[width=0.32\textwidth, keepaspectratio=true] {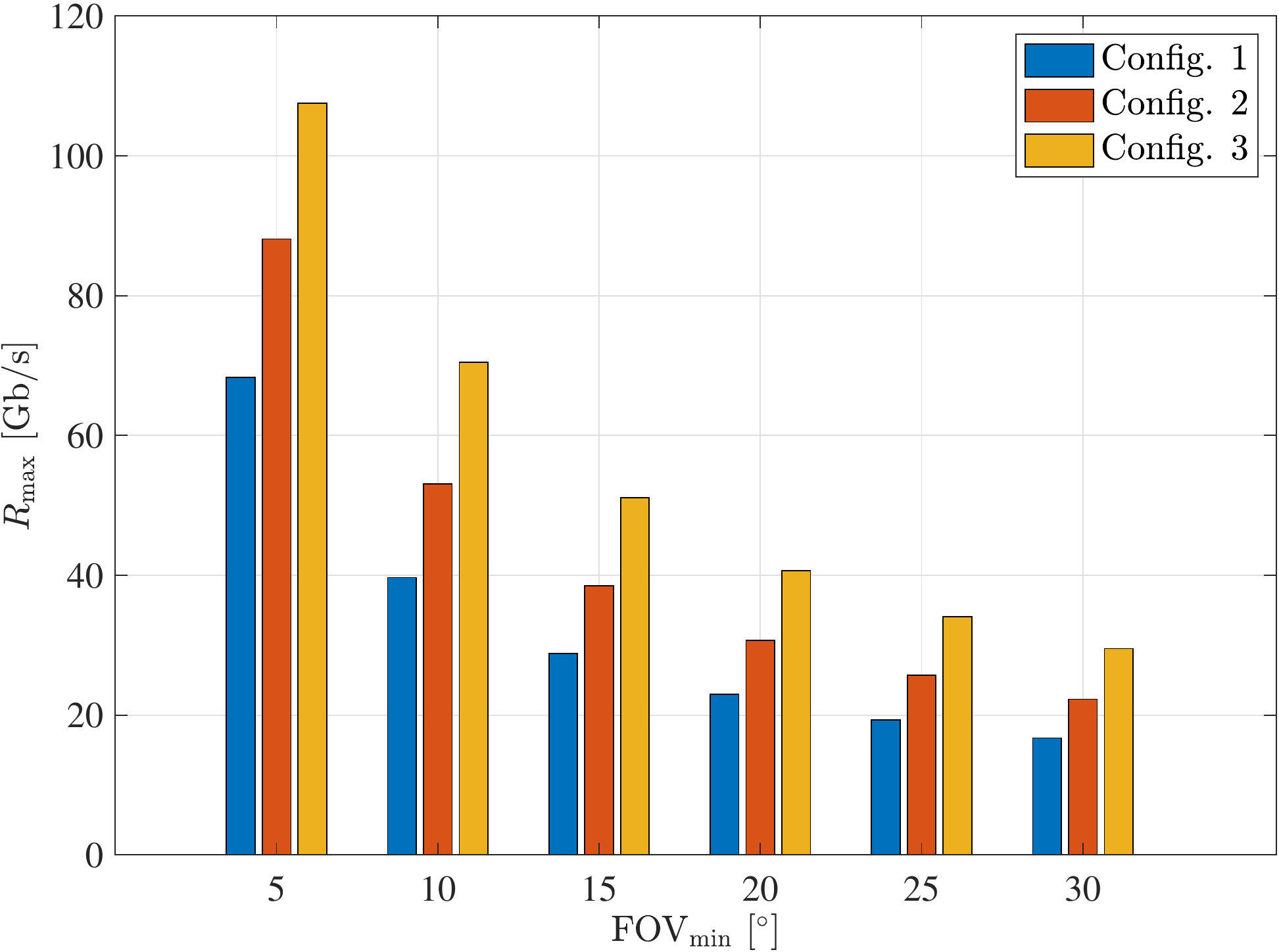}}\hfill
    \subfloat[Modified ADR: \\$L_{\max}=2$~cm, $A_{\max}=4$~cm$^2$ \label{Fig:Rmax_FOVmin_c}] {\includegraphics[width=0.32\textwidth, keepaspectratio=true] {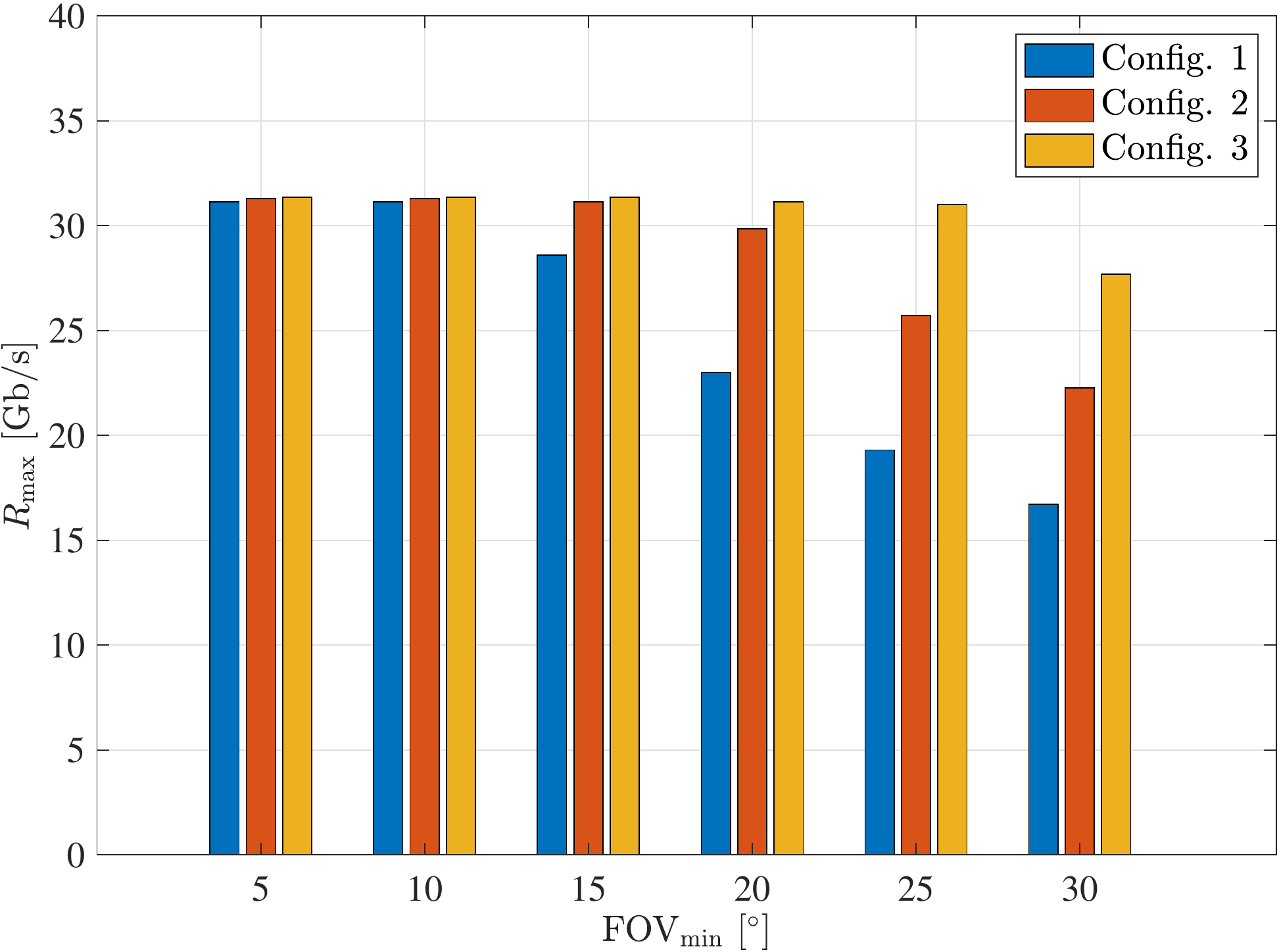}}\hfill
    \subfloat[Modified ADR: \\$L_{\max}=0.5$~cm, $A_{\max}=0.5$~cm$^2$ \label{Fig:Rmax_FOVmin_e}] {\includegraphics[width=0.32\textwidth, keepaspectratio=true] {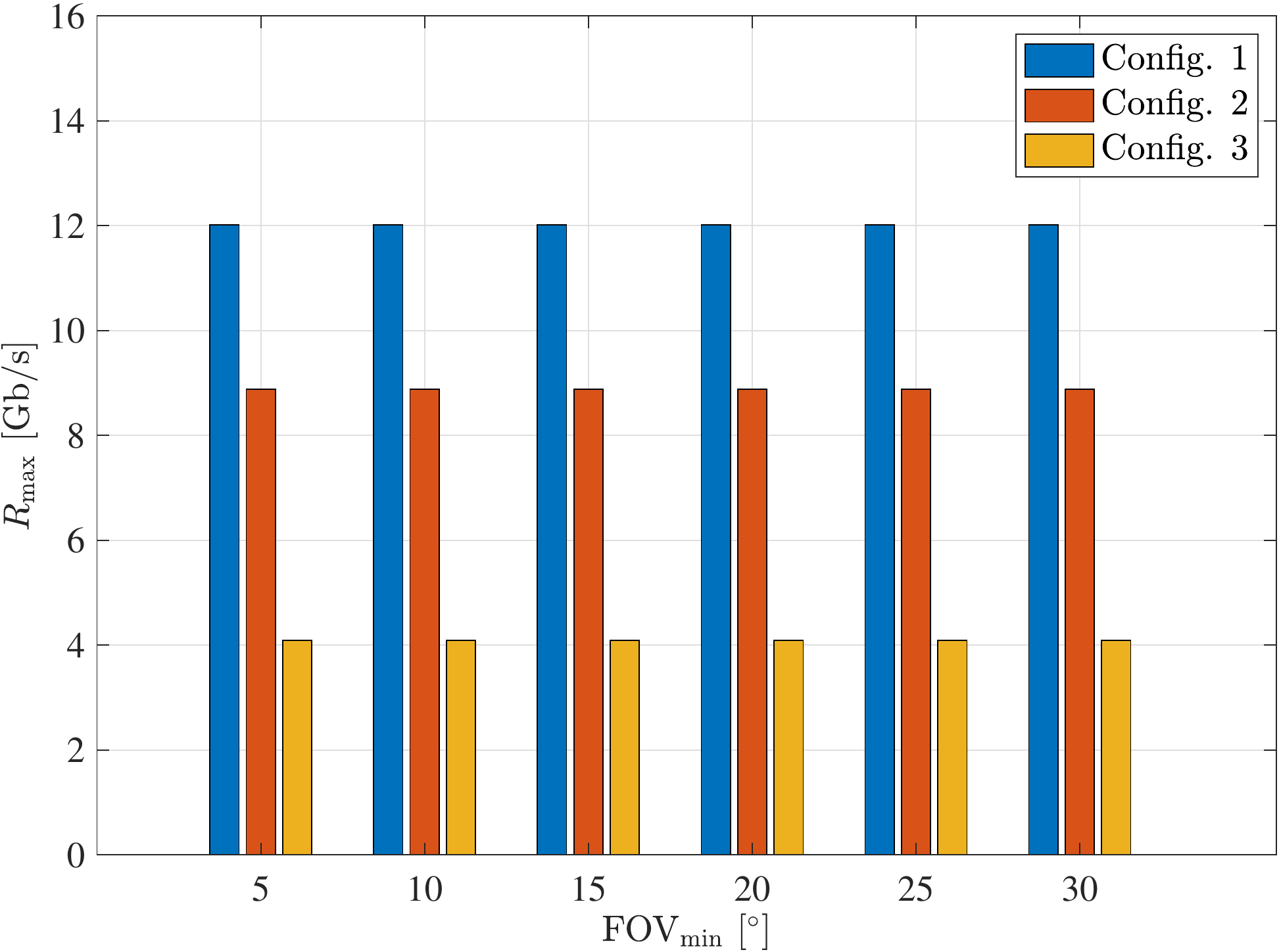}}\\
    \subfloat[Original ADR: \\No constraint on dimensions \label{Fig:Rmax_FOVmin_b}] {\includegraphics[width=0.32\textwidth, keepaspectratio=true] {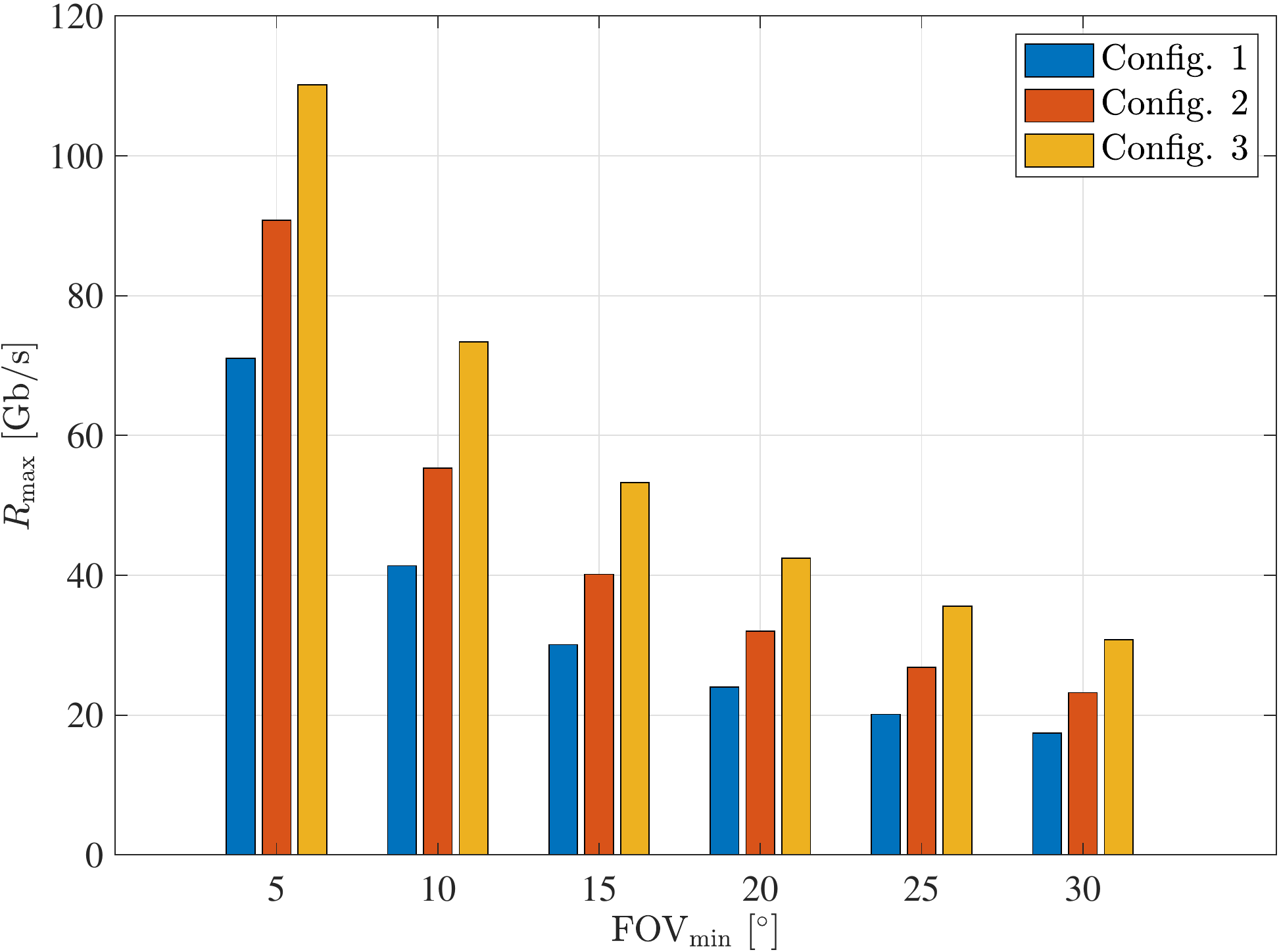}}\hfill
    \subfloat[Original ADR: \\$L_{\max}=2$~cm, $A_{\max}=4$~cm$^2$ \label{Fig:Rmax_FOVmin_d}] {\includegraphics[width=0.32\textwidth, keepaspectratio=true] {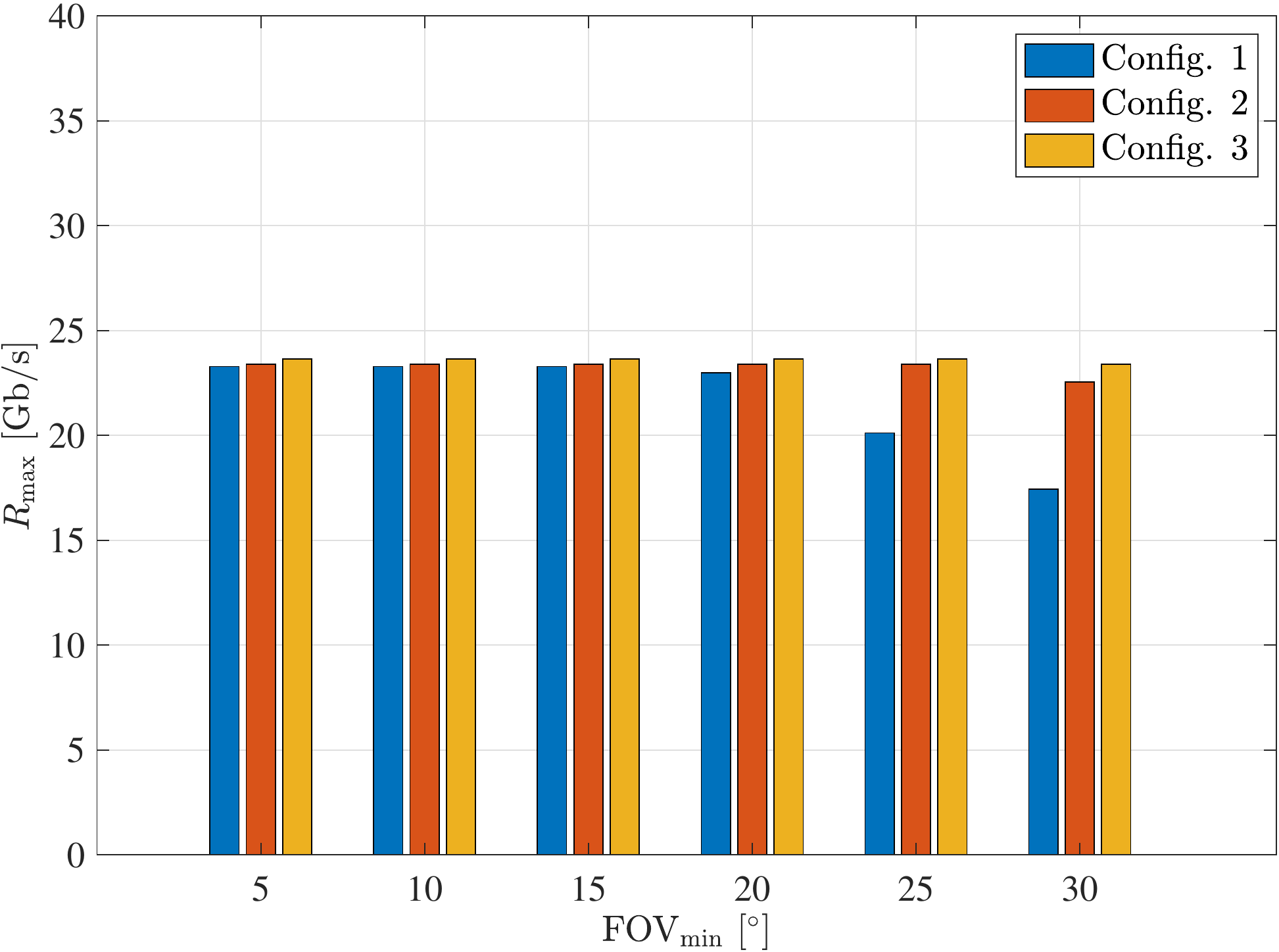}}\hfill
    \subfloat[Original ADR: \\$L_{\max}=0.5$~cm, $A_{\max}=0.5$~cm$^2$ \label{Fig:Rmax_FOVmin_f}] {\includegraphics[width=0.32\textwidth, keepaspectratio=true] {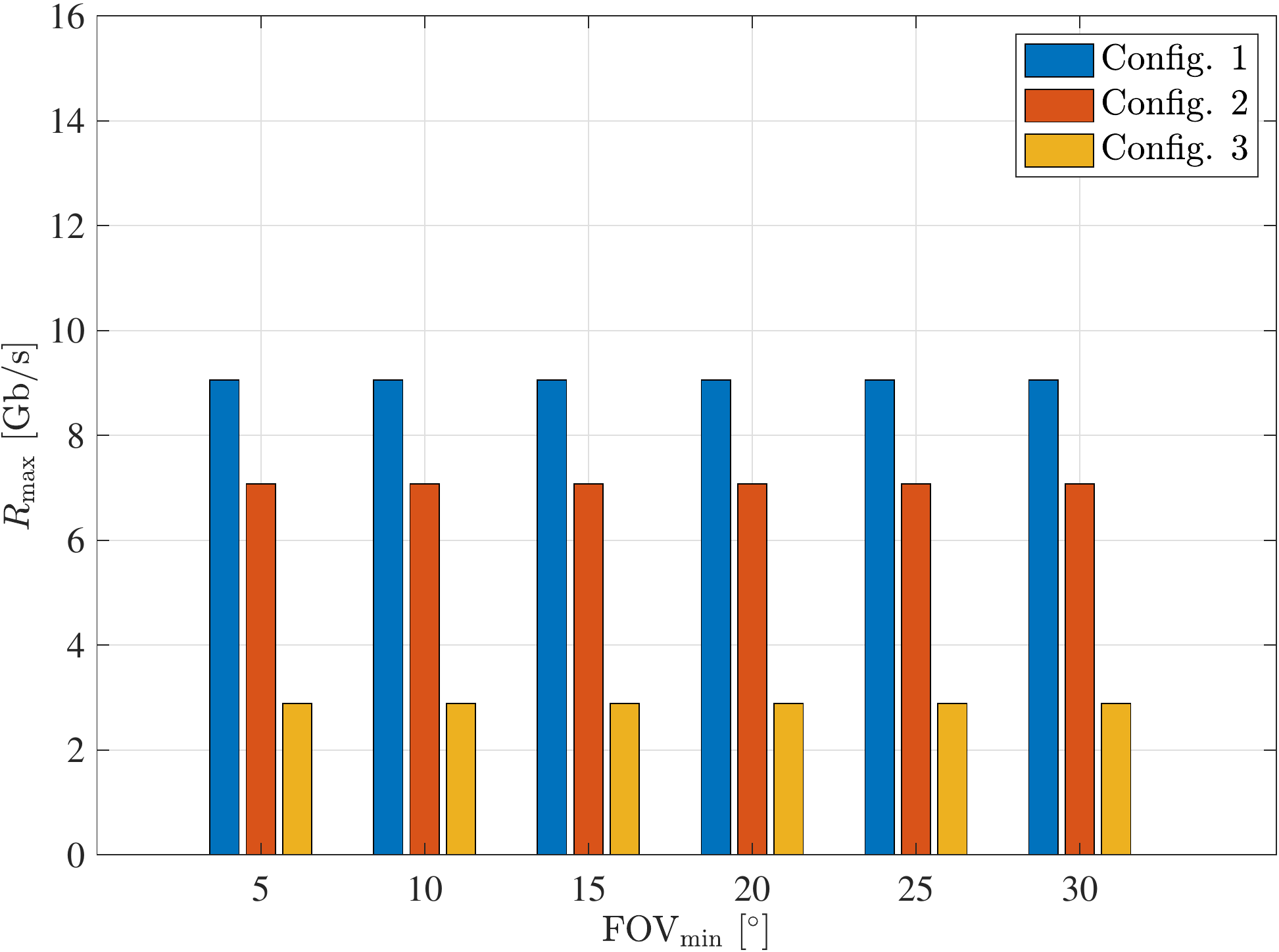}}
    \caption{The maximum achievable rate $R_\mathrm{max}$ vs. $\mathrm{FOV}_\mathrm{min}$ for different PD array sizes and $w_0=10$~{\textmu}m.}
    \label{Fig:Rmax_FOVmin}
    \vspace{-20pt}
\end{figure}

Fig.~\ref{Fig:Rmax_FOVmin} compares the modified \ac{ADR} based on truncated \acp{CPC} versus the original \ac{ADR} using full-length \acp{CPC} in terms of the maximum achievable rate performance. Recall that Configs.~1--3 represent a $1$-tier \ac{ADR} with $2\times2$, $4\times4$ and $8\times8$ \ac{PD} arrays, respectively, per receiver element. To make a fair comparison, we use the same set of $\mathrm{FOV}_\mathrm{min}$, $L_\mathrm{max}$ and $A_\mathrm{max}$ for both \ac{ADR} designs. Depending on the values of these parameters, three operating regimes are distinguished as \ac{NCD} as shown in Figs.~\subref*{Fig:Rmax_FOVmin_a} and \subref*{Fig:Rmax_FOVmin_b}, \ac{MCD} as shown in Figs.~\subref*{Fig:Rmax_FOVmin_c} and \subref*{Fig:Rmax_FOVmin_d}, and \ac{SCD} as shown in Figs.~\subref*{Fig:Rmax_FOVmin_e} and \subref*{Fig:Rmax_FOVmin_f}. When there is no constraint on dimensions, both \ac{ADR} designs achieve almost the same performance and follow an identical decreasing trend with $\mathrm{FOV}_\mathrm{min}$. In this case, the performance of the modified \ac{ADR} is only slightly lower, which is attributed to the $10\%$ loss in the optical gain of \acp{CPC} because of the $40\%$ length truncation. Also, Configs.~1 and 3 mark the lowest and the highest performance levels, respectively. Notwithstanding, Configs.~1--3 attain $R_\mathrm{max}>15$~Gb/s with $\mathrm{FOV}_\mathrm{min}=30^\circ$. The modified \ac{ADR} design manifests its advantage when there are constraints on dimensions. Under \ac{MCD}, for $L_{\max}=2$~cm and $A_{\max}=4$~cm$^2$, the modified \ac{ADR} realises considerably higher performance bounds as compared to the original \ac{ADR}. In the case of \ac{SCD}, for $L_{\max}=0.5$~cm and $A_{\max}=0.5$~cm$^2$, each configuration for each \ac{ADR} design yields an equal performance level for all values of $\mathrm{FOV}_\mathrm{min}$. It is noteworthy that, in contrast to \ac{NCD}, the order of Configs.~1--3 in holding the lowest-to-highest performance levels is reversed under \ac{SCD}. This indicates that under extremely limiting constraints on dimensions, increasing the \ac{PD} array size in fact adversely affects the receiver performance. The modified \ac{ADR} of $L_{\max}=0.5$~cm and $A_{\max}=0.5$~cm$^2$ is a compact \ac{ADR} design that achieves $R_\mathrm{max}=12$~Gb/s with $\mathrm{FOV}_\mathrm{min}=30^\circ$ based on a $2\times2$ \ac{PD} array. The highest performance achieved by the original \ac{ADR} under the same conditions is just above $9$~Gb/s.

%%%%%%%%%%%%%%%%%%%%%%%%%%%%%%%%%%%%%%%%%%%%%%%%%%%%%%%%%%%%%%%%%%%%%%%%%%%%%%%%%%%%%%%%%%%%%%%%%%%%
%%%%%%%%%%%%%%%%%%%%%%%%%%%%%%%%%%%%%%%%%%%%%%%%%%%%%%%%%%%%%%%%%%%%%%%%%%%%%%%%%%%%%%%%%%%%%%%%%%%%
\section{Conclusions} \label{Sec:7}
An in-depth study of the fundamental design tradeoffs was conducted for non-imaging \acp{ADR} based on \acp{CPC}. Building on a two-stage diversity combining scheme that reduces the complexity, a unified and tractable analytical framework was developed to meet the challenging requirements for high performance laser-based optical wireless receivers. A non-convex optimisation problem was formulated for maximising the achievable data rate $R$ under constraints on $\mathrm{FOV}$, the overall height $L_\mathrm{ADR}$ and the total top area $A_\mathrm{ADR}$, to find optimum values of $\mathrm{FOV}$ and the \ac{PD} bandwidth $B$ for a given \ac{PD} array size. First, with an objective to achieve $R\geq10$~Gb/s under the minimum \ac{FOV} constraint only, the results evince a tradeoff between $R$ and $\mathrm{FOV}$. In particular, $R=10$~Gb/s is realised with $\mathrm{FOV}\leq65^\circ$, while $R=20$~Gb/s is achieved with $\mathrm{FOV}\leq28^\circ$. Alternatively, when $\mathrm{FOV}$ is fixed, $R$ has an absolute maximum with respect to $B$, indicating that the performance is indeed degraded if the bandwidth is increased beyond a certain threshold. For $\mathrm{FOV}=30^\circ$, the peak data rates of $R=14.00,18.56,24.53$~Gb/s are observed at $B=2.1,2.7,3.5$~GHz for $2\times2$, $4\times4$ and $8\times8$ \ac{PD} arrays, respectively. Hence, when there is no constraint on the receiver dimensions, increasing the \ac{PD} array size improves the maximum rate performance, even though a $2\times2$ \ac{PD} array serves the purpose of fulfilling $R\geq10$~Gb/s. For jointly satisfying $R\geq10$~Gb/s and $\mathrm{FOV}\geq30^\circ$, the design space enlarges by increasing the \ac{PD} array size used for each \ac{ADR} element as well as by adding to the number of \ac{ADR} tiers. It was shown that the design spaces provided by a $2$-tier \ac{ADR} with $4\times4$ \ac{PD} arrays and a $3$-tier \ac{ADR} using $2\times2$ \ac{PD} arrays are almost equally large. 

In the presence of the constraints on dimensions, the results bring out an exacting tradeoff of practical importance for designing wide \ac{FOV} \acp{ADR}: acquiring a small footprint and achieving a high data rate are opposing objectives. For a $1$-tier \ac{ADR} based on $8\times8$ \ac{PD} arrays with $\mathrm{FOV}_\mathrm{min}=30^\circ$, $R_\mathrm{max}=30$~Gb/s provided that $L_\mathrm{max}\geq3$~cm and $A_\mathrm{max}\geq4$~cm$^2$. The same \ac{ADR} configuration reaches $R_\mathrm{max}=10$~Gb/s if $L_\mathrm{max}=1$~cm and $A_\mathrm{max}=2$~cm$^2$. The results also demonstrate the significant impact of choosing small values for $L_\mathrm{max}$ such that the overall height constraint takes over the performance when $L_\mathrm{max}$ reduces to $0.5$~cm. To overcome this challenge, a modified \ac{ADR} solution by means of $40\%$ length truncation for \acp{CPC} was proposed. It turns out that this truncation has only a marginal impact on the overall performance when the dimensions are unconstrained. The modified \ac{ADR} has an advantage over the original \ac{ADR} made of full-length \acp{CPC} under moderately to strictly constrained dimensions. It was found that under extremely limiting constraints on dimensions, the use of larger \ac{PD} arrays significantly degrades the maximum achievable rate. However, with the proposed modification, a compact \ac{ADR} design of $L_{\max}=0.5$~cm and $A_{\max}=0.5$~cm$^2$ is able to achieve $R_\mathrm{max}=12$~Gb/s with $\mathrm{FOV}_\mathrm{min}=30^\circ$ based on a $2\times2$ \ac{PD} array. Future research includes a comparison of imaging versus non-imaging receivers, and performance optimisation of multi-beam optical wireless networks by using \acp{ADR} with mobility and random orientation.

%%%%%%%%%%%%%%%%%%%%%%%%%%%%%%%%%%%%%%%%%%%%%%%%%%%%%%%%%%%%%%%%%%%%%%%%%%%%%%%%%%%%%%%%%%%%%%%%%%%%
%%%%%%%%%%%%%%%%%%%%%%%%%%%%%%%%%%%%%%%%%%%%%%%%%%%%%%%%%%%%%%%%%%%%%%%%%%%%%%%%%%%%%%%%%%%%%%%%%%%%

%%%%%%%%%%%%%%%%%%%%%%%%%%%%%%%%%%%%%%%%%%%%%%%%%%%%%%%%%%%%%%%%%%%%%%%%%%%%%%%%%%%%%%%%%%%%%%%%%%%%
%%%%%%%%%%%%%%%%%%%%%%%%%%%%%%%%%%%%%%%%%%%%%%%%%%%%%%%%%%%%%%%%%%%%%%%%%%%%%%%%%%%%%%%%%%%%%%%%%%%%
\section*{Acknowledgement}
The authors acknowledge financial support from the Engineering and Physical Sciences Research Council (EPSRC) under grant EP/S016570/1 `Terabit Bidirectional Multi-User Optical Wireless System (TOWS) for 6G LiFi'. 

%%%%%%%%%%%%%%%%%%%%%%%%%%%%%%%%%%%%%%%%%%%%%%%%%%%%%%%%%%%%%%%%%%%%%%%%%%%%%%%%%%%%%%%%%%%%%%%%%%%%

%%%%%%%%%%%%%%%%%%%%%%%%%%%%%%%%%%%%%%%%%%%%%%%%%%%%%%%%%%%%%%%%%%%%%%%%%%%%%%%%%%%%%%%%%%%%%%%%%%%%
%%%%%%%%%%%%%%%%%%%%%%%%%%%%%%%%%%%%%%%%%%%%%%%%%%%%%%%%%%%%%%%%%%%%%%%%%%%%%%%%%%%%%%%%%%%%%%%%%%%%
\appendices
%-----------------------------------------
% Appendix A: Proof of Proposition 1
%-----------------------------------------
\section{Proof of Proposition~\ref{Proposition:1}} \label{App:1}
The partial derivative of $R$ with respect to $\mathrm{FOV}$ is obtained by using \eqref{Eq:Rate}:
\begin{equation}
    \dfrac{\partial R}{\partial \mathrm{FOV}} =
    \dfrac{2B}{\ln2} \left(\dfrac{{R_\mathrm{PD}^2 P_\mathrm{r}}}{{\Gamma N_0 B} + {R_\mathrm{PD}^2 P_\mathrm{r} ^2}}\right) \dfrac{\partial P_\mathrm{r}}{\partial \mathrm{FOV}} \raisepunct{.}
    \label{Eq:PD_R_FOV}
\end{equation}
Based on \eqref{Eq:FOV_theta} and \eqref{Eq:Pr2}, it follows that:
\begin{equation}
    \dfrac{\partial P_\mathrm{r}}{\partial \mathrm{FOV}} = 
    \left(\dfrac{1}{2N_\mathrm{tier}+1}\right) \dfrac{\partial P_\mathrm{r}}{\partial \theta_\mathrm{CPC}} = 
    \left(\dfrac{1}{2N_\mathrm{tier}+1}\right) \dfrac{P_{\mathrm{t}} D_1 \mathrm{FF}}{(w'(D))^2} \exp{\left( -\dfrac{D_1 ^2}{2(w'(D))^2} \right )} \dfrac{\partial D_1}{\partial \theta_\mathrm{CPC}} \raisepunct{.}
    \label{Eq:PD_Pr_FOV}
\end{equation}
According to \eqref{Eq:PD_R_FOV} and \eqref{Eq:PD_Pr_FOV}, the polarity of $\dfrac{\partial R}{\partial \mathrm{FOV}}$ depends on the last product term, $\dfrac{\partial D_1}{\partial \theta_\mathrm{CPC}}$. Based on \eqref{Eq:CPC_D1}, the partial derivative of $D_1$ with respect to $\theta_\mathrm{CPC}$ is given by:
\begin{equation}
\dfrac{\partial D_1}{\partial \theta_\mathrm{CPC}} = \dfrac{2K_1}{ B} \left(-\dfrac{n_\mathrm{CPC} \cos\theta_\mathrm{CPC}}{\sin^2 \theta_\mathrm{CPC}}\right)<0\raisepunct{,}
\end{equation}
because $\cos\theta_\mathrm{CPC}>0$ due to Corollary~\ref{Corollary:1}. As a result, $\dfrac{\partial R}{\partial\mathrm{FOV}}<0$, hence the proof is completed.
%-----------------------------------------
% Appendix B: Proof of Proposition 2
%-----------------------------------------
\section{Proof of Proposition~\ref{Proposition:2}} \label{App:2}
It is clear from \eqref{Eq:L_ADR} that $\dfrac{\partial L_{\mathrm{ADR}}}{\partial B}<0$. By using \eqref{Eq:FOV_theta} and \eqref{Eq:L_ADR}, the partial derivative of $L_{\mathrm{ADR}}$ with respect to $\mathrm{FOV}$ is derived in the form:
\begin{equation}
\begin{aligned}
    \mkern-18mu \dfrac{\partial L_{\mathrm{ADR}}} {\partial \mathrm{FOV}} = 
    &\left(\dfrac{1}{2N_\mathrm{tier}+1}\right) \dfrac{\partial L_{\mathrm{ADR}}}{\partial \theta_\mathrm{CPC}} = \left(\dfrac{1}{2N_\mathrm{tier}+1}\right) \times\\
    &\dfrac{K_1}{B} \left(-\dfrac{2n_\mathrm{CPC} \sin\theta_\mathrm{CPC} +  n_\mathrm{CPC} \sin\theta_\mathrm{CPC}\tan^2\theta_\mathrm{CPC} + \tan^2\theta_\mathrm{CPC}}{\sin^2\theta_\mathrm{CPC}\tan^2\theta_\mathrm{CPC}}\right) \raisepunct{.}
\end{aligned}
\end{equation}
Considering the fact that $\theta_\mathrm{CPC}\leq\dfrac{\pi}{6}$ from Corollary~\ref{Corollary:1}, $\sin\theta_\mathrm{CPC}>0$. Consequently, $\dfrac{\partial L_{\mathrm{ADR}}}{\partial\mathrm{FOV}}<0$.

Also, based on \eqref{Eq:A_ADR}, it can be readily verified that $\dfrac{\partial A_{\mathrm{ADR}}}{\partial B}<0$. By using \eqref{Eq:FOV_theta} and \eqref{Eq:A_ADR}, the partial derivative of $A_{\mathrm{ADR}}$ with respect to $\mathrm{FOV}$ is obtained as follows:
\begin{equation}
\begin{aligned}
    \mkern-18mu &\dfrac{\partial A_{\mathrm{ADR}}} {\partial \mathrm{FOV}}
    = \left(\dfrac{1}{2N_\mathrm{tier}+1}\right) \times \\
    &\frac{K_2}{B^2} \left[ -\dfrac{2 \sin\theta_\mathrm{CPC} \cos\theta_\mathrm{CPC}}{\sin^4\theta_\mathrm{CPC}} \left( {1 + \sum_{i=1}^{N_\mathrm{tier}} 6i\cos(2i\theta_\mathrm{CPC})} \right)
    - \dfrac{1}{\sin^2\theta_\mathrm{CPC}} {\sum_{i=1}^{N_\mathrm{tier}} 12i^2 \sin(2i\theta_\mathrm{CPC})} \right].
\end{aligned}
\end{equation}
Since $\theta_\mathrm{CPC}\leq\dfrac{\pi}{6}$, $\sin\theta_\mathrm{CPC}>0$ and $\cos\theta_\mathrm{CPC}>0$. Also, $2i\theta_\mathrm{CPC}<\dfrac{\pi}{2}$, thus $\sin(2i\theta_\mathrm{CPC})>0$ and $\cos(2i\theta_\mathrm{CPC})>0$. Therefore, $\dfrac{\partial A_{\mathrm{ADR}}}{\partial \mathrm{FOV}}<0$. This completes the proof.

%%%%%%%%%%%%%%%%%%%%%%%%%%%%%%%%%%%%%%%%%%%%%%%%%%%%%%%%%%%%%%%%%%%%%%%%%%%%%%%%%%%%%%%%%%%%%%%%%%%%
%%%%%%%%%%%%%%%%%%%%%%%%%%%%%%%%%%%%%%%%%%%%%%%%%%%%%%%%%%%%%%%%%%%%%%%%%%%%%%%%%%%%%%%%%%%%%%%%%%%%
\bibliographystyle{IEEEtran}
\bibliography{IEEEabrv,refs}
\end{document}